%% file: main_arxive.tex
\documentclass[pdflatex,sn-mathphys-num]{sn-jnl}
 
\usepackage{amssymb,amsmath,amsfonts,mathrsfs,amsthm}
\usepackage[utf8]{inputenc}
\usepackage[T1]{fontenc}
\usepackage[english]{babel}
\usepackage{graphicx}
\usepackage[title]{appendix}
\usepackage{color}
\usepackage{doi}
\usepackage{commath}
\usepackage{array}
\usepackage[normalem]{ulem}
\usepackage{xcolor}%
\usepackage{textcomp}%
\usepackage{manyfoot}%
\usepackage{booktabs}%
\usepackage{dsfont}
\usepackage{algorithm}
\usepackage{algpseudocode}
\usepackage{anyfontsize}

\usepackage{thmtools}
\usepackage{thm-restate}
\usepackage{hyperref}
\usepackage{varioref}
\usepackage{cleveref}

\theoremstyle{plain}
\newtheorem{theorem}{Theorem}%  meant for continuous numbers
\newtheorem{lemma}{Lemma}% 
\newtheorem{corollary}{Corollary}
\newtheorem{proposition}{Proposition}% to get separate numbers for theorem and proposition etc.

\theoremstyle{definition}%
\newtheorem{definition}{Definition}%

\theoremstyle{remark}%
\newtheorem{remark}{Remark}%

\input{parts/grib_defs}

\DeclareMathOperator{\ILPSF}{ILP-SF}
\DeclareMathOperator{\ILPCF}{ILP-CF}
\DeclareMathOperator{\GenILPSF}{Generalized-ILP-SF}

\DeclareMathOperator{\GroupMin}{Group-Min}
\DeclareMathOperator{\LocalILPSF}{Local-ILP-SF}
\DeclareMathOperator{\LocalILPCF}{Local-ILP-CF}
\DeclareMathOperator{\ModularILPSF}{Modular-ILP-SF}

\newenvironment{myproof}[1][Proof]{\begin{proof}[#1]}{\end{proof}}

\newcommand{\TABLE}[3]{
\caption{#1 {(#3)}}
#2
}

\begin{document}

\title[Delta-modular Problems of Bounded Codimension \ldots]{Delta-modular ILP Problems of Bounded Codimension, Discrepancy, and Convolution
%\thanks{Section \ref{main_result_sec} was prepared within the framework of the Basic Research Program at the National Research University Higher School of Economics (HSE). Subsection \ref{standard_ILP_subs} and Sections \ref{comb_prob_sec}, \ref{exp_ILP_complexity_sec} were prepared under financial support of Russian Science Foundation grant No ?????????.}
}
\author[1]{\fnm{Mikhail} \sur{Cherniavskii}}\email{cherniavskii.miu@phystech.edu}

\author*[1]{\fnm{Dmitry V.} \sur{Gribanov}}\email{dimitry.gribanov@gmail.com}

\author[2]{\fnm{Dmitry S.} \sur{Malyshev}}\email{dsmalyshev@rambler.ru}

\author[3]{\fnm{Panos M.} \sur{Pardalos}}\email{panos.pardalos@gmail.com}

% \author[3]{\fnm{Panos M.} \sur{Pardalos}}\email{panos.pardalos@gmail.com}

% \author[2]{\fnm{Nikolai Yu.} \sur{Zolotykh}}\email{nikolai.zolotykh@itmm.unn.ru}

\affil*[1]{\orgdiv{Laboratory of Discrete and Combinatorial Optimization}, \orgname{Moscow Institute of Physics and Technology}, \orgaddress{\street{Institutsky lane 9}, \city{Dolgoprudny, Moscow region}, \postcode{141700}, \country{Russian Federation}}}

\affil[2]{\orgdiv{Laboratory of Algorithms and Technologies for Network Analysis}, \orgname{HSE University}, \orgaddress{\street{136 Rodionova Ulitsa}, \city{Nizhny Novgorod}, \postcode{603093}, \country{Russian Federation}}}

\affil[3]{\orgdiv{Department of Industrial and Systems Engineering}, \orgname{University of Florida}, \orgaddress{\street{401 Weil Hall}, \city{Gainesville}, \postcode{116595}, \state{Florida}, \country{USA}}}

% \affil[2]{\orgname{Lobachevsky State University of Nizhny Novgorod}, \orgaddress{\street{23 Gagarina Avenue}, \city{Nizhny Novgorod}, \postcode{603950}, \country{Russian Federation}}}

% \affil[3]{\orgdiv{Department of Industrial and Systems Engineering}, \orgname{University of Florida}, \orgaddress{\street{401 Weil Hall}, \city{Gainesville}, \postcode{116595}, \state{Florida}, \country{USA}}}

\abstract{
\input{parts/abstract}

}%abstract

\maketitle              % typeset the header of the contribution

\keywords{Bounded Sub-determinants, Bounded Codimension, Integer Programming, Convolution, Group Ring, Discrepancy}

\tableofcontents

\input{parts/text}

\backmatter

\addcontentsline{toc}{section}{Acknowledgements}
\section*{Acknowledgements}

\input{parts/acknowledgements}

\begin{appendices}

\input{parts/appendices}

\end{appendices}

\addcontentsline{toc}{section}{References}
\bibliography{parts/grib_biblio}

\end{document}

%% file: parts/grib_defs.tex
\DeclareMathOperator{\CC}{\mathbb{C}}

\DeclareMathOperator{\RR}{\mathbb{R}}

\DeclareMathOperator{\QQ}{\mathbb{Q}}

\DeclareMathOperator{\ZZ}{\mathbb{Z}}

\DeclareMathOperator{\BB}{\mathbb{B}}

\DeclareMathOperator{\CCQ}{\CC\!\QQ}

\DeclareMathOperator{\GC}{\mathcal{G}}
\DeclareMathOperator{\RC}{\mathcal{R}}
\DeclareMathOperator{\BC}{\mathcal{B}}

\DeclareMathOperator{\WC}{\mathcal{W}}

\DeclareMathOperator{\PC}{\mathcal{P}}

\DeclareMathOperator{\SC}{\mathcal{S}}

\DeclareMathOperator{\LC}{\mathcal{L}}
\DeclareMathOperator{\HC}{\mathcal{H}}
\DeclareMathOperator{\MC}{\mathcal{M}}
\DeclareMathOperator{\AC}{\mathcal{A}}
\DeclareMathOperator{\FC}{\mathcal{F}}

\DeclareMathOperator{\JC}{\mathcal{J}}
\DeclareMathOperator{\IC}{\mathcal{I}}
\DeclareMathOperator{\CCal}{\mathcal{C}}
\DeclareMathOperator{\QC}{\mathcal{Q}}

\DeclareMathOperator{\NotBC}{\overline{\BC}}

\DeclareMathOperator{\rank}{rank}

\DeclareMathOperator{\size}{size}
\DeclareMathOperator{\poly}{poly}
\DeclareMathOperator{\polylog}{polylog}

\DeclareMathOperator{\const}{const}

\DeclareMathOperator{\vol}{vol}

\DeclareMathOperator{\disc}{disc}
\DeclareMathOperator{\herdisc}{herdisc}

\DeclareMathOperator{\detlb}{detlb}

\DeclareMathOperator{\BZero}{\mathbf 0}

\DeclareMathOperator{\prob}{\Pr}

\DeclareMathOperator{\LP}{LP}

\newcommand*{\intint}[2][1]{\{#1, \dots, #2\}}

\DeclareMathOperator{\DP}{D\!P}

\DeclareMathOperator{\Hom}{Hom}

\newcommand\restr[2]{{% we make the whole thing an ordinary symbol
  \left.\kern-\nulldelimiterspace % automatically resize the bar with \right
  #1 % the function
  \vphantom{\big|} % pretend it's a little taller at normal size
  \right|_{#2} % this is the delimiter
  }}

%% file: parts/abstract.tex
For integers $k,n \geq 0$ and a cost vector $c \in \ZZ^n$, we study two fundamental integer linear programming (ILP) problems:
\begin{gather*}
    \text{(Standard Form)} \quad \max\bigl\{c^\top x \colon Ax = b,\ x \in \ZZ^n_{\geq 0}\bigr\} \text{ with } A \in \ZZ^{k \times n}, \rank(A) = k, b \in \ZZ^k, \\
    \text{(Canonical Form)} \quad \max\bigl\{c^\top x \colon Ax \leq b,\ x \in \ZZ^n\bigr\} \text{ with } A \in \ZZ^{(n+k) \times n}, \rank(A) = n, b \in \ZZ^{n+k}.
\end{gather*}
We present improved algorithms for both problems and their feasibility versions, parameterized by $k$ and $\Delta$, where $\Delta$ denotes the maximum absolute value of $\rank(A) \times \rank(A)$ subdeterminants of $A$. Our main complexity results, stated in terms of required arithmetic operations, are:
\begin{itemize}
    \item Optimization:
    $O(\log k)^{2k} \cdot \Delta^2 / 2^{\Omega(\sqrt{\log \Delta})} + 2^{O(k)} \cdot \poly(\varphi)$,
    
    \item Feasibility:
    $O(\log k)^k \cdot \Delta \cdot (\log \Delta)^3 + 2^{O(k)} \cdot \poly(\varphi)$,
\end{itemize}
where $\varphi$ represents the input size measured by the bit-encoding length of $(A,b,c)$.

We also examine several special cases when $k \in \{0,1\}$, which have applications in: expected computational complexity of ILP with varying right-hand side $b$; ILP problems with generic constraint matrices; ILP problems on simplices. Our results yield improved complexity bounds for these specific scenarios. 

As independent contributions, we present: an $n^2/2^{\Omega(\sqrt{\log n})}$-time algorithm for the tropical convolution problem on sequences indexed by elements of a finite Abelian group of order $n$; a complete and self-contained error analysis of the generalized DFT over Abelian groups in the Word-RAM model.

%% file: parts/text.tex
\section{Introduction}\label{intro_sec}

Let us first define two ILP problems of our interest:
\begin{definition}
Let $A \in \ZZ^{k\times n}$, $\rank(A) = k$, $c \in \ZZ^n$, $b \in \ZZ^{k}$. 
Assume additionally that $k \times k$ sub-determinants of $A$ are co-prime, we will clarify this assumption later (see Remark \ref{ILPSF_GCD_assumption_rm}).
\emph{The ILP problem in the standard form of the co-dimension $k$} is formulated as follows:
\begin{align}
    &c^\top x \to \max \notag\\
    &\begin{cases}
    A x = b\\
    x \in \ZZ_{\geq 0}^n.
    \end{cases}\label{ILP-SF}\tag{\(\ILPSF\)}
\end{align}
For simplicity, we assume that $\dim(\PC) = n-k$ for the corresponding polyhedra $\PC = \{x \in \RR_{\geq 0}^n \colon A x = b\}$.
\end{definition}

\begin{definition}
Let $A \in \ZZ^{(n + k)\times n}$, $\rank(A) = n$, $c \in \ZZ^n$, $b \in \ZZ^{n+k}$. \emph{The ILP problem in the canonical form with $n+k$ constraints} is formulated as follows:
\begin{align}
    &c^\top x \to \max \notag\\
    &\begin{cases}
    A x \leq b\\
    x \in \ZZ^n.
    \end{cases}\label{ILP-CF}\tag{\(\ILPCF\)}
\end{align}
Again, for simplicity, we assume that $\dim(\PC) = n$ for the corresponding polyhedra $\PC = \{x \in \RR^n \colon A x \leq b\}$.
\end{definition}

We study the computational complexity of these problems with respect to $k$ and the absolute values of sub-determinants of $A$. Values of sub-determinants are controlled, using
\begin{definition}
For a matrix $A \in \ZZ^{k \times n}$ and $j \in \intint{k}$, by $$
\Delta_j(A) = \max\left\{\abs{\det (A_{\IC \JC})} \colon \IC \subseteq \intint k,\, \JC \subseteq \intint n,\, \abs{\IC} = \abs{\JC} = j\right\},
$$ we denote the maximum absolute value of determinants of all the $j \times j$ sub-matrices of $A$. 
By $\Delta_{\gcd}(A,j)$, we denote the greatest common divisor of determinants of all the $j \times j$ sub-matrices of $A$.
% By $\Delta_{\gcd}(A,j)$ and $\Delta_{\lcm}(A,j)$, we denote the greatest common divisor and the least common multiplier of nonzero determinants of all the $j \times j$ sub-matrices of $A$, respectively. 
Additionally, let $\Delta(A) = \Delta_{\rank(A)}(A)$ and $\Delta_{\gcd}(A) = \Delta_{\gcd}(A,\rank(A))$. A matrix $A$ with $\Delta(A) \leq \Delta$, for some $\Delta > 0$, is called \emph{$\Delta$-modular}. Note that $\Delta_1(A) = \|A\|_{\max}$. 
% An $\Delta$-modular matrix $A$ with totally non-zero $\rank(A) \times \rank(A)$ sub-determinants is called the \emph{strongly $\Delta$-modular}. 
\end{definition}

\begin{remark}\label{WhuTwoFormulations_rm}
    Why do we consider both formulations in this work? While Problems~\ref{ILP-SF} and~\ref{ILP-CF} are mutually transformable, all known trivial transformations alter at least one of the key parameters $\bigl(k, d, \Delta(A)\bigr)$. The existence of a parameter-preserving transformation is a nontrivial question, resolved by Gribanov et al. \cite{OnCanonicalProblems_Grib}.

    Interestingly, Problem~\ref{ILP-CF} proves to be strictly more general than Problem~\ref{ILP-SF}. For exact equivalence between these formulations, one must augment Problem~\ref{ILP-SF} with additional Abelian group constraints (see Definition \ref{ModularILPSF_def}). Section~\ref{connection_sec} provides detailed discussion of this relationship.

    This duality motivates our use of both formulations: while Problem~\ref{ILP-CF} offers greater generality, Problem~\ref{ILP-SF} remains more prevalent in the ILP literature.
\end{remark}

For the sake of simplicity, we will use the shorthand notations $\Delta := \Delta(A)$, $\Delta_1 := \Delta_1(A)$, and $\Delta_{\gcd} := \Delta_{\gcd}(A)$ in the further text. Additionally, we use the notation $d$ to denote the dimension of a corresponding polyhedron and $\phi$ to denote the input size of the corresponding problems. For the problem \ref{ILP-SF}, we can assume that $\phi = O\bigl(k \cdot n \cdot \log \alpha \bigr)$, where $\alpha$ is the maximum absolute value of elements of $A$, $b$, and $c$. Similarly, for the problem \ref{ILP-CF}, we can assume that $\phi = O\bigl(n \cdot (n + k) \cdot \log \alpha \bigr)$. 

The paper \cite{DiscConvILP} of K.~Jansen \& L.~Rohwedder states that the problem \ref{ILP-SF} can be solved with 
\begin{equation}\label{JR_complexity_eq}
    O(k)^{k} \cdot \Delta_1^{2k} / 2^{\Omega\bigl(\sqrt{\log \Delta_1}\bigr)} + T_{\LP}
\end{equation}
operations, where $T_{\LP}$ denotes the computational complexity to solve the relaxed LP problem. Putting $k = 1$, it leads to an $O\bigl( n+ \Delta_1^2 / 2^{\Omega(\sqrt{\log \Delta_1})}\bigr)$-time algorithm for the \emph{unbounded knapsack problem}. For the feasibility variant of the problem \ref{ILP-SF}, the paper of K.~Jansen \& L.~Rohwedder presents an algorithm, which runs with
\begin{equation}\label{JR_feasibility_complexity_eq}
    O\bigl(\sqrt{k} \cdot \Delta_1\bigr)^{k} \cdot (\log \Delta_1)^2 + T_{\LP}
\end{equation}
operations. 
% Due to T.~Chan \cite{fixed_m_LP_Chan}, $T_{\LP} = O(k)^{k/2} \cdot \log^{3k}(k) \cdot n$, so, the former complexity bounds can be rewritten in a simpler form
% \begin{gather*}
%     O(k)^{k} \cdot \Delta_1^{2k} / 2^{\Omega\bigl(\sqrt{k \cdot \log(k \Delta_1 )}\bigr)} \\
%     O\bigl(\sqrt{k} \cdot \Delta_1\bigr)^{k + o(k)}.
% \end{gather*}
Putting $k=1$, it leads to an $O\bigl(n + \Delta_1 \cdot \log^2 (\Delta_1)\bigr)$-time algorithm for the \emph{unbounded subset-sum problem}. The paper of K.~Jansen \& L.~Rohwedder uses a new dynamic programming technique, composed with the seminal upper bound on the \emph{hereditary discrepancy of $A$} by J.~Spencer\cite{SixDeviations_Spencer}. Additional speedup is achieved, using the \emph{fast tropical convolution} algorithm, due to R.~Williams \cite{APSPViaCircuitComplexity} and D.~Bremner et al. \cite{NecklacesConv}, for the optimization variant of the problem, and using a \emph{fast FFT-based Boolean convolution} algorithm for the feasibility variant of the problem, respectively. Additionally, K.~Jansen \& L.~Rohwedder proved that their computational complexity bounds are optimal with respect to the parameter $\Delta_1$, providing conditional lower bounds.

It is interesting to estimate computational complexity of the considered problems with respect to $\Delta$ instead of $\Delta_1$. Due to the Hadamard's inequality, it seems that such computational complexity bounds could be more general than the computational complexity bounds with respect to $\Delta_1$. Following to this line of research, the paper \cite{OnCanonicalProblems_Grib}, due to D.~Gribanov, I.~Shumilov, D.~Malyshev \& P.~Pardalos, presents an algorithm for the problem \ref{ILP-SF}, which is parameterized by $\Delta$. Actually, the paper \cite{OnCanonicalProblems_Grib} gives an algorithm for some generalized problem, which, in turn, is equivalent to the problem \ref{ILP-CF}.
\begin{theorem}[Corollary~9, Gribanov et al. \cite{OnCanonicalProblems_Grib}]\label{GribanovCanonicalWork_compplexity}
    Any of the problems \ref{ILP-CF} and \ref{ILP-SF} can be solved with 
    $$
    O(\log k)^{2k^2 + o(k^2)} \cdot \Delta^2 \cdot \log^2(\Delta) + \poly(\phi) \quad{operations}.
    $$
\end{theorem}
Unfortunately, the paper \cite{OnCanonicalProblems_Grib} does not propose faster computational complexity bounds for feasibility variants of the problems \ref{ILP-CF} and \ref{ILP-SF}, because of non-triviality of tropical and Boolean convolution-type problems, arising in the computational complexity analysis with respect to the parameter $\Delta$. In the current paper, we resolve these difficulties and provide refined computational complexity bounds for the problems \ref{ILP-CF} and \ref{ILP-SF}, and their feasibility variants. More precisely, we prove
\begin{restatable}{theorem}{mainILPTh}
\label{main_ILP_th}
    Any of the problems \ref{ILP-CF} and \ref{ILP-SF} can be solved with
    $$
    O(\log k)^{2k} \cdot \Delta^2 / 2^{\Omega\bigl(\sqrt{\log \Delta}\bigr)} + 2^{O(k)} \cdot \poly(\phi)\quad\text{operations.}
    $$
    The feasibility variant of the problem can be solved with\footnote{Considering only the problem \ref{ILP-SF}, the multiplicative term $\log^3 \Delta$ can be replaced by $\log^2 \Delta$.}
    $$
    O(\log k)^k \cdot \Delta \cdot \log^3 \Delta + 2^{O(k)} \cdot \poly(\phi) \quad\text{operations.}
    $$
\end{restatable}
% \begin{theorem}\label{main_ILP_th}
%     Any of the problems \ref{ILP-CF} and \ref{ILP-SF} can be solved with
%     $$
%     O(\log k)^{2k} \cdot \Delta^2 / 2^{\Omega\bigl(\sqrt{\log \Delta}\bigr)} + 2^{O(k)} \cdot \poly(\phi)\quad\text{operations.}
%     $$
%     % $$
%     % 2^{O(k)} \cdot \left( (f_k \cdot \Delta)^2 / 2^{\Omega\bigl(\sqrt{\log(f_k \cdot \Delta)}\bigr)} + \poly(\phi) \right)
%     % $$ operations, where $f_k = (\log k)^k$.
%     % , where
%     % $$
%     % f_{k,d} = \min \left\{ \begin{aligned}
%     % &k^{k/2},\\
%     % &\bigl(\log k \cdot \log (k+d)\bigr)^{k/2}\\
%     % \end{aligned}\right.
%     % $$ and $d$ is the dimension of the corresponding polyhedron ($d = n$ for \ref{ILP-CF} and $d = n-k$ for \ref{ILP-SF}).
%     The feasibility variant of the problem can be solved with 
%     $$
%     O(\log k)^k \cdot \Delta \cdot \log^3 \Delta + 2^{O(k)} \cdot \poly(\phi) \quad\text{operations.}
%     $$
%     % $$
%     % 2^{O(k)} \cdot \left( f_{k} \cdot  \Delta \cdot \log^3 \Delta + \poly(\phi) \right) \quad\text{operations.}
%     % $$
% \end{theorem}
The proof is given in \Vref{main_ILP_th_proof}. The proposed computational complexity bounds are much better than the bound of \Cref{GribanovCanonicalWork_compplexity}, for all values of the parameters $\Delta$ and $k$, and, especially, for the feasibility-type problems. We further note that by applying Hadamard's inequality to the complexity bounds of \Cref{main_ILP_th}, we obtain complexity bounds
\begin{gather*}
    O(k)^{k + o(k)} \cdot \Delta_1^{2k} + 2^{O(k)} \cdot \poly(\phi), \quad\text{and}\\
    O(k)^{k/2 + o(k)} \cdot \Delta_1^k \cdot (\log \Delta_1)^3 + 2^{O(k)} \cdot \poly(\phi),
\end{gather*}
which are close to the complexity bounds \eqref{JR_complexity_eq} and \eqref{JR_feasibility_complexity_eq} of the work \cite{DiscConvILP}. Thus, in a certain sense, \Cref{main_ILP_th} generalizes \eqref{JR_complexity_eq} and \eqref{JR_feasibility_complexity_eq}.
The complexity bounds of \Cref{main_ILP_th} are conditionally optimal with respect to the parameter $\Delta$, see \Cref{CLB_subs}.

\subsection{Summary of the Obtained Results}\label{summary_subs}

\begin{itemize}
    \item We propose new algorithms for the Problems \ref{ILP-CF}, \ref{ILP-SF}, and their feasibility variants, parametrized by $k$ and $\Delta$. All details could be found in Theorem \ref{main_ILP_th}, a comparison with previous works could be found in Tables \ref{opt_results_tb} and \ref{fea_results_tb}. Table \ref{fea_results_tb} considers the feasibility variants of the problems \ref{ILP-CF} and \ref{ILP-SF}. 
    Additionally, we show that the obtained complexity bounds are conditionally optimal with respect to the parameter $\Delta$, see Subsection \ref{CLB_subs}.

    \item We apply the results above to different problems, which include ILP formulations with $k=0$ and $k=1$:
    \begin{itemize}
        % \item{\bf Gomory's group minimization problem.} We show that Gomory's group minimization problem can be solved in $O\bigl(n + \Delta^2/2^{\Omega(\sqrt{\Delta})}\bigr)$ operations, where $n$ denotes the number of elements and $\Delta$ is the group order. The new computational complexity bound is better than the known bound $O(n \Delta)$ for the cases, when $n = \Omega\bigl(\Delta / \polylog(\Delta)\bigr)$;
        
        \item The expected ILP computational complexity with respect to a varying right-hand side $b$;
        % {\bf The expected ILP computational complexity with respect to a varying right-hand side $b$.} In this direction, we show that Gomory's group minimization problem can be solved in $O\bigl(\Delta^2/2^{\Omega(\sqrt{\Delta})} + \poly(\varphi)\bigr)$ operations. The new complexity As a consequence, assuming that 

        \item ILP problems with generic constraint matrices;

        \item ILP problems on simplices.
    \end{itemize} 
    % \begin{itemize}
    %     \item The expected ILP computational complexity of the problem \ref{ILP-CF} and \ref{ILP-SF} with respect to a varying right-hand side $b$. In this direction we propose 
    % \end{itemize}
    These applications with the corresponding refined complexity bounds are discussed in Section \ref{app_sec}.

    \item As a result of independent interest, we propose an $n^2/2^{\Omega(\sqrt{\log n})}$-operations algorithm for the generalized tropical convolution with respect to Abelian groups of order $n$. The formal definition of the problem could be found in Section \ref{conv_sec}, the corresponding result is formulated in Theorem \ref{group_min_conv_th}. 
    
    \item As a second result of independent interest, we present a complete, self-contained error analysis of the generalized Discrete Fourier Transform for Abelian groups with respect to the Word-RAM computational model. As a consequence, we show that generalized Boolean convolution with respect to Abelian groups of order $n$ can be solved in $O( n \log n )$ operations with rational numbers if size $O(\log n)$. We were unable to locate this result in the literature. The formal definition of the problem could be found in Subsection \ref{generalized_DFT_subs}, the corresponding result is formulated in Theorem \ref{group_DFT_th}.
\end{itemize}

\begin{table}[h!]
    \TABLE
    {The computational complexity bounds for the problems \ref{ILP-CF} and \ref{ILP-SF}\label{opt_results_tb}}
    {\begin{tabular}{||m{12em}|m{14em}|m{12em}||}
    \hline
    \hline
    Reference: & Time: & Remark  \\
    \hline
    \hline
    Jansen \& Rohwedder \cite{DiscConvILP} & $O(k)^{k} \cdot \Delta_1^{2k} / 2^{\Omega\bigl(\sqrt{\log \Delta_1 }\bigr)}$ & only for \ref{ILP-SF} \\
    \hline
    Gribanov et al. \cite{OnCanonicalProblems_Grib} & $O(\log k)^{2k^2 + o(k^2)} \cdot \Delta^2 \cdot \log^2(\Delta)$ & \\
    \hline
    {\color{red} this work} & $O(\log k)^{2k} \cdot \Delta^2 / 2^{\Omega\bigl(\sqrt{\log \Delta}\bigr)}$ & \\
    % \hline
    % {\color{red} this work} & $2^{O(k)} \cdot (f_{k} \cdot \Delta)^2 / 2^{\Omega\bigl(\sqrt{\log(f_{k} \cdot \Delta)}\bigr)},$ & the additive term $2^{O(k)} \cdot \poly(\phi)$ is skipped \\
    % & where $f_{k} = (\log k)^k$ &  \\
    % \hline
    % {\color{red} this work} & $2^{O(k)} \cdot (g_{k,\Delta} \cdot \Delta)^2 / 2^{\Omega\bigl(\sqrt{\log(g_{k,\Delta} \cdot \Delta)}\bigr)}$ & $g_{k,\Delta} = \bigl(\log k \cdot \log (k \Delta)\bigr)^{k/2}$ \\
    \hline
    \hline
    \end{tabular}}{the additive term $2^{O(k)} \cdot \poly(\phi)$ is skipped}
    % {The additive factor $2^{O(k)} \cdot \poly(\phi)$ is skipped}
\end{table}

\begin{table}[h!]
    \TABLE
    {The computational complexity bounds for the feasibility variants of the problems \ref{ILP-CF} and \ref{ILP-SF}\label{fea_results_tb}}
    {\begin{tabular}{||m{12em}|m{14em}|m{12em}||}
    \hline
    \hline
    Reference: & Time: & Remark  \\
    \hline
    \hline
    Jansen \& Rohwedder \cite{DiscConvILP} & $O\bigl(\sqrt{k} \cdot \Delta_1\bigr)^{k} \cdot (\log \Delta_1)^2$ & only for \ref{ILP-SF} \\
    \hline
    Gribanov et al. \cite{OnCanonicalProblems_Grib} & $O(\log k)^{2k^2 + o(k^2)} \cdot \Delta^2 \cdot \log^2(\Delta)$ & \\
    \hline
    {\color{red} this work} & $O(\log k)^{k} \cdot \Delta \cdot \log^3 \Delta$  & \\
    % \hline
    % {\color{red} this work} & $2^{O(k)} \cdot (g_{k,\Delta} \cdot  \Delta) \cdot \log^3(g_{k,\Delta} \cdot  \Delta)$ & $g_{k,\Delta} = \bigl(\log k \cdot \log (k \Delta)\bigr)^{k/2}$ \\
    \hline
    \hline
    \end{tabular}}{the additive term $2^{O(k)} \cdot \poly(\phi)$ is skipped}
    % {The additive factor $2^{O(k)} \cdot \poly(\phi)$ is skipped}
\end{table}

% The summary of our main results, detailed in Theorem \ref{main_ILP_th}, 
% % and \ref{main_ILP_logD_th}, 
% could be found in Tables \ref{opt_results_tb} and \ref{fea_results_tb}.

% We apply these results to different partial cases with $k=0$ and $k=1$: the expected ILP computational complexity of the problem \ref{ILP-CF} and \ref{ILP-SF} with respect to a varying right-hand side $b$, ILP problems with generic constraint matrices, ILP problems on simplices. These applications with the corresponding refined complexity bounds are discussed in Section \ref{app_sec}.

% Finally, as a result of independent interest, we propose an $n^2/2^{\Omega(\sqrt{\log n})}$-time algorithm for the tropical convolution problem on sequences, indexed by an Abelian group of the order $n$. The formal definition of the problem could be found in Section \ref{conv_sec}, the corresponding result is formulated in Theorem \ref{group_min_conv_th}. Additionally, we present a complete, self-contained error analysis of the generalized Discrete Fourier Transform for Abelian groups with respect to the Word-RAM computational model. The formal definition of the problem could be found in Subsection \ref{generalized_DFT_subs}, the corresponding result is formulated in Theorem \ref{group_DFT_th}.

\subsection{Conditional Lower Bounds}\label{CLB_subs}
It turns out that the computational complexity bounds of Theorem \ref{main_ILP_th} can not be significantly reduced with respect to the parameter $\Delta$, assuming the correctness of some strong hypothesis from the fine-grained complexity theory. This fact is a straightforward consequence of the Hadamard's inequality and following Theorems, due to K.~Jansen \& L.~Rohwedder \cite{DiscConvILP}.

\begin{theorem}[K.~Jansen \& L.~Rohwedder \cite{DiscConvILP}]
    Let $k \in \ZZ_{\geq 1}$. For any $\varepsilon > 0$ and any computable function $f$, there is not an algorithm that solves the problem \ref{ILP-SF} in time $f(k) \cdot \bigl( n^{2-\varepsilon} + (\Delta_1 + \norm{b}_{\infty})^{2k-\varepsilon} \bigr)$, unless there exists a truly sub-quadratic algorithm for the tropical convolution.
\end{theorem}

\begin{theorem}[K.~Jansen \& L.~Rohwedder \cite{DiscConvILP}]
    Let $k \in \ZZ_{\geq 1}$. If the SETH holds, then, for every $\varepsilon > 0$ and every computable function $f$, there is not an algorithm that solves the feasibility variant of the problem \ref{ILP-SF} in time $n^{f(k)} \cdot (\Delta_1 + \norm{b}_{\infty})^{k - \varepsilon}$.
\end{theorem}

The following corollaries give the desired conditional lower bounds with respect to the parameter $\Delta$.
\begin{corollary}\label{optLB_cor}
    Let $k \in \ZZ_{\geq 1}$. For any $\varepsilon > 0$ and any computable function $f$, there is not an algorithm that solves any of the problems \ref{ILP-SF} and \ref{ILP-CF} in time $f(k) \cdot ( n^{2-\varepsilon} + \Delta^{2-\varepsilon} )$, unless there exists a truly sub-quadratic algorithm for the tropical convolution.
\end{corollary}

\begin{corollary}\label{feaLB_cor}
    Let $k \in \ZZ_{\geq 1}$. If the SETH holds, then, for every $\varepsilon > 0$ and every computable function $f$, there is not an algorithm that solves the feasibility variant of any the problems \ref{ILP-SF} and \ref{ILP-CF} in time $n^{f(k)} \cdot \Delta^{1 - \varepsilon}$.
\end{corollary}
As it was already noted, the proof is a straightforward consequence of the Hadamard's inequality, Theorems above, and the reductions between the problems \ref{ILP-CF} and  \ref{ILP-SF} (see Section \ref{connection_sec}).

\subsection{Proof Sketch of Main \Cref{main_ILP_th}}\label{proof_outline_sec}

To simplify our proof of \Cref{main_ILP_th} and to generalize our results, we reduce both the \ref{ILP-CF} and \ref{ILP-SF} problems to a single formulation, which constitutes the formulation \ref{ILP-SF} with an additional Abelian group constraint:
\begin{align}
    &c^\top x \to \max\notag\\
    &\begin{cases}
        A x = b\\
        \sum\limits_{i=1}^n x_i \cdot g_i = g_0\\
        x \in \ZZ_{\geq 0}^n.
    \end{cases}\tag{\(\GenILPSF\)}\label{GEN-ILP-SF-Duplicate}
\end{align}
Here $g_0,g_1, \dots, g_n$ are elements of a finite Abelian group.
% We then apply \Cref{genILP_main_th} to work with this problem. 
A connection between these problems is discussed in \Cref{connection_sec}, where the corresponding generalized analogue of \Cref{main_ILP_th} is also presented as \Cref{genILP_main_th}.

However, the exposition style we have chosen might obscure the main idea behind the proof of \Cref{main_ILP_th}. To provide more clarity, we now briefly outline key steps of the proof of \Cref{main_ILP_th} for the following special case, involving only the \ref{ILP-SF} problem
\begin{align}
    &c^\top x \to \max\notag\\
    &\begin{cases}
    A x = b\\
    x \in \ZZ_{\geq 0}^n.
    \end{cases}\tag{\(\ILPSF\)}\label{ILP-SF-Duplicate}
\end{align}

We employ a modification of the approach from Jansen \& Rohwedder \cite{DiscConvILP}, which combines dynamic programming principles with solution rounding techniques, based on discrepancy theory. To present our algorithm and compare it with the approach of \cite{DiscConvILP}, we first briefly outline the required background from discrepancy theory. Recall that the discrepancy and hereditary discrepancy of a matrix $A$ are defined as follows:
\begin{gather*}
\disc(A) = \min_{z \in \{-1,\, 1\}^n} \left\| A z  \right\|_\infty,\\
\herdisc(A) = \max_{\IC \subset \intint n} \disc(A_{* \IC}).
\end{gather*}
We will need the following result from Jiang \& Reis \cite{TightDiscDetBound}, which together with \cite[Corollary 13.3.3]{AlonSpencerBook} states that:
\begin{gather}
\disc(A) = O\Bigl( \log k \cdot \detlb(A) \Bigr), \quad\text{where}\label{DiscDetBound_Duplicate_eq}\\
\detlb(A) = \max\limits_{t \in \intint k} \sqrt[t]{\Delta_t(A)}.\notag
\end{gather} 
A more detailed exposition of the discrepancy theory results we employ in this work can be found in \Cref{prediscrep_sec}.

The difference between our approach and the approach of \cite{DiscConvILP} is that \cite{DiscConvILP} utilizes discrepancy upper bounds for the matrix $A$. In our work, we construct a certain regularization factor $B \in \ZZ^{k \times k}$ and use a discrepancy upper bound for the matrix $B^{-1} \cdot A$ instead of $A$ alone. This leads to the following theorem. 
\begin{theorem}\label{DP_th_Duplicate}
    Let $\rho \in \ZZ_{>0}$ be a given value, such that $\norm{z^*} \leq (6/5)^\rho$, for some optimal integer solution of the problem \ref{ILP-SF-Duplicate}. Additionally, for a given nondegenerate $B \in \ZZ^{k \times k}$, let $M = B^{-1} \cdot A$, $\eta \in \ZZ_{> 0}$ be an upper bound on $\herdisc(B^{-1} \cdot A)$, and $\delta = \abs{ \det B}$. Then the problem \ref{ILP-SF-Duplicate} can be solved in 
    \begin{equation*}
        \rho \cdot \tau^2 / 2^{\Omega\bigl(\sqrt{\log \tau}\bigr)} + O(n)
    \end{equation*}
    operations with vectors from $\ZZ^k$, where $\tau = O(\eta)^k \cdot \delta$. The feasibility variant of the problem can be solved in 
    \begin{equation*}
        O(\rho \cdot \tau \cdot \log \tau + n) \quad\text{operations.}
    \end{equation*}
\end{theorem}
We omit its proof since a substantially more general \Cref{DP_th} will be proved later in this paper. Briefly, the proof utilizes a dynamic program $\DP(i,b')$, parameterized by $i$ and $b'$ from the following sets:
\begin{gather*}
    i \in \intint[0]{\rho},\\
    b' \in \ZZ^k \cap \left( 2^{i-\rho} b + (\gamma B) \cdot [-1,1]^k\right),
\end{gather*}
where $\gamma = O(\eta)$.
% and $[-1,1]^k$ denotes the unit $\ell_\infty$-ball. 
% As it was already noted, the work \cite{DiscConvILP} does not use a regularization factor $B$. 
% In other words, in the context of this work one could simply set $B = I$.
In the context of the work \cite{DiscConvILP}, one could simply set $B$ equal to the identity matrix.
To enumerate values of $b'$, we use an algorithmic variant of the following well known fact.
% (see, for example, Seb\H{o} \cite{IntroEmptyLatticeSimplicies}).
\begin{proposition}
    For an arbitrary nondegenerate $B \in \ZZ^{k \times k}$, we have
    \begin{equation*}
        \abs{\ZZ^k \cap B \cdot [-1,1]^k} = \vol\bigl(B \cdot [-1,1]^k\bigr) = \abs{\det B}.
    \end{equation*}
\end{proposition}
To compute the values of $\DP(i,b')$, the authors of \cite{DiscConvILP} proposed to use algorithms for fast tropical and Boolean sequence convolution. However, for handling values of $b'$ from our generalized sets, standard convolution algorithms are unsuitable, and we require generalized tropical and Boolean convolution algorithms, relative to Abelian groups. These algorithms are described in \Cref{conv_sec}. In turn, \Cref{tiling_group_sec} focuses on an algorithmic construction of the corresponding Abelian groups that will serve a basement for the generalized convolution algorithms.

Assuming now that \Cref{DP_th_Duplicate} holds, we explain how to properly bound the parameters $\rho$, $\delta$, and $\eta$, appearing in its statement. The parameter $\rho$ can be bounded, using a standard technique, based on the proximity result, stating that the optimal integer solution $z^*$ of \ref{ILP-SF-Duplicate} is close to the optimal solution $x^*$ of the corresponding relaxed linear program. In our case, we use the result from Lee et al. \cite{ProximityViaSparsity}, which guarantees the existence of such $z^*$, satisfying $\norm{x^* - z^*}_1 = O(k^2 \cdot \Delta^{1 + 1/k})$. The latter allows us to set $\rho = O(\log(k \Delta))$. For the generalized problem formulation \ref{GEN-ILP-SF-Duplicate}, we use an extension of this result from Gribanov et al. \cite{OnCanonicalProblems_Grib}.

The selection of parameters $\eta$ and $\delta$ is a more challenging task that essentially reduces to choosing a regularization factor $B$ that minimizes $\tau = O(\eta)^k \cdot \delta$. Recall that $\delta = \abs{\det B}$ and $\eta$ is an upper bound on $\herdisc(B^{-1} \cdot A)$. A good choice for $B$ is a square $k \times k$ submatrix of $A$ such that $\abs{\det B} = \Delta(A)$. Indeed, it is easy to see that in this case, for any $i \in \intint k$, we have $\Delta_i(B^{-1} \cdot A) \leq 1$. In other words, the matrix $B^{-1} \cdot A$ is "totally unimodular". Here, the term "totally unimodular" is in quotes because $B^{-1} \cdot A$ might not be integer-valued. In this case, according to the hereditary discrepancy bound \eqref{DiscDetBound_Duplicate_eq}, $\herdisc(B^{-1} \cdot A) = O(\log k)$, and we can set $\eta = O(\log k)$.

Thus, setting $\rho = O(\log(k \Delta))$, $\delta = \Delta$, and $\eta = O(\log k)$, we obtain $\tau = O(\log k)^k \cdot \Delta$. The main result (\Cref{main_ILP_th} for \ref{ILP-SF-Duplicate} problems) then follows from \Cref{DP_th_Duplicate} by substituting the value of $\tau$.

However, there remains one unresolved caveat: finding a $k \times k$ submatrix $B$ of $A$ satisfying $\abs{\det B} = \Delta(A)$ is an NP-hard problem. A viable workaround is to relax the condition $\abs{\det B} = \Delta(A)$. In this work, we employ a result from Nikolov \cite{LargestSimplex_Nikolov}, which states that a submatrix $B$ with $\abs{\det B} \geq e^{-k} \cdot \Delta(A)$ can be computed by a polynomial-time algorithm. By choosing $B$ this way, we ensure that $\Delta(B^{-1} \cdot A) \leq e^k$. This property of matrix $B^{-1} \cdot A$ can be extended to all its submatrices by the cost of additional $2^k \cdot \poly(\phi)$ operations, where $\phi$ is the input length (we omit these technical details here). The hereditary discrepancy bound \eqref{DiscDetBound_Duplicate_eq} still guarantees $\herdisc(B^{-1} \cdot A) = O(\log k)$ in this case, which suffices to complete our proof sketch.

\subsection{Complexity Model and Other Assumptions}\label{assumptions_subs}
All the algorithms that are considered in our work correspond to the \emph{Word-RAM} computational model. In other words, we assume that additions, subtractions, multiplications, and divisions with rational numbers of the specified size, which is called the \emph{word size}, can be done in $O(1)$-time. In our work, we chose the word size to be equal to some fixed polynomial on $\lceil \log n \rceil + \lceil \log k \rceil + \lceil \log \alpha \rceil$, where $\alpha$ is the maximum absolute value of elements of $A$, $b$, and $c$ in the problem formulations. 
Sometimes, when it is important to specify the exact size of variables, we do this explicitly. Following \cite[Section~4.1]{KorteBook}, we define the bit-encoding size of an integer number $z \in \ZZ$ by the formula $\size(z) = 1 + \bigl\lceil\log_2\bigl(\,\abs{z}+1\bigr)\bigr\rceil$. For a rational number $r = p/q$, where $p$ and $q$ are coprime integers, the size is defined by the formula $\size(r) = \size(p) + \size(q)$.  

\begin{remark}\label{ILPSF_GCD_assumption_rm}
    Let us clarify the assumption $\Delta_{\gcd}(A) = 1$, which was done in the \ref{ILP-SF} problems' definition. Assume that $\Delta_{\gcd}(A) = d > 1$, and let us show that the original problem can be reduced to an equivalent new problem with $\Delta_{\gcd}(A') = 1$, using a polynomial-time reduction. 

    Let $A = P \cdot \bigl(S\,\BZero\bigr) \cdot Q$, where $\bigl(S\,\BZero\bigr) \in \ZZ^{k \times n}$, be the \emph{Smith Normal Form (the SNF, for short)} of $A$ and $P \in \ZZ^{k \times k}$, $Q \in \ZZ^{n \times n}$ be unimodular matrices. We multiply rows of the original system $A x = b,\, x \geq \BZero$ by the matrix $(P S)^{-1}$. After this step, the original system is transformed to the equivalent system $\bigl(I_{k \times k}\,\BZero\bigr) \cdot Q\,x = b^\prime$, $x \geq \BZero$. In the last formula, $b^\prime$ is integer, because in the opposite case the original system is integrally infeasible. Clearly, the matrix $\bigl(I_{k \times k}\,\BZero\bigr)$ is the SNF of $\bigl(I_{k \times k}\,\BZero\bigr)\,Q$, so its $\Delta_{\gcd}(\cdot)$ is equal to $1$. Finally, note that the computation of the SNF is a polynomial-time solvable problem, see, for example, A.~Storjohann \cite{SNFOptAlg}. 
\end{remark}

\section{Applications}\label{app_sec}

\subsection{The Case \(k=1\) and Generic \(\Delta\)-modular Matrices}

In this Subsection, we consider ILP problems with the additional assumption that all $\rank(A) \times \rank(A)$ sub-determinants of $A$ are non-zero. It was proposed by B.~Kriepke, G.~Kyureghyan \& M.~Schymura \cite{generic_DeltaModular} to call such matrices and corresponding ILP problems as \emph{generic}. 
% \begin{definition}[\cite{generic_DeltaModular}]
%     An integer matrix $A$ is called \emph{generic}, if all its $\rank(A) \times \rank(A)$ sub-determinants are non-zero, that is, every set of $rk(A)$ columns or rows of $A$ is linearly independent.
% \end{definition}
Generic ILP problems became known due to S.~Artmann et al. \cite{NonDegenerateMinors}, where it was shown that generic problems can be solved by a polynomial-time algorithm, for any fixed value of $\Delta$.
\begin{theorem}[S.~Artmann et al. \cite{NonDegenerateMinors}]\label{Artmanns_generic_th}
    Let us consider any of the problems \ref{ILP-CF} and \ref{ILP-SF}. Assuming, that the matrix $A$ is generic and $\Delta$-modular, for some fixed value of $\Delta$, the problem is polynomial-time solvable.
\end{theorem}

The key fact that helps to prove Theorem \ref{Artmanns_generic_th} is described in
\begin{theorem}[S.~Artmann et al. \cite{NonDegenerateMinors}]\label{Artmanns_m_th}
    Let $A \in \ZZ^{m \times n}$ be a generic matrix with $\rank(A) = n$. Denote $\Delta := \Delta(A)$. There exists a function $f$, such that if $n \geq f\bigl(\Delta(A)\bigr)$, then $m \leq n + 1$.
\end{theorem}
The original work of S.~Artmann et al. gives a very rough estimate on the growth rate of $f$. It was subsequently improved in the papers \cite{OnCanonicalProblems_Grib,IntegralityNumber,ABCModular,generic_DeltaModular}. An asymptotically optimal bound is presented in the work \cite{generic_DeltaModular} of B.~Kriepke, G.~Kyureghyan \& M.~Schymura. We partially refer their result in the following
\begin{theorem}[Theorem~1.3, B.~Kriepke et al. \cite{generic_DeltaModular}]\label{Kriepke_m_th}
    Let $A \in \ZZ^{m \times n}$ be a generic matrix with $\rank(A) = n$. If $n \geq 2 \Delta(A) - 1$, then $m \leq n + 1$.
\end{theorem}

It follows from Theorem \ref{Kriepke_m_th} that any instance of the problem \ref{ILP-CF} with a generic matrix $A$ and $n \geq 2 \Delta - 1$ has $k \leq 1$. Therefore, due to Theorem \ref{GribanovCanonicalWork_compplexity}, it can be solved with $O\bigl(\Delta^2 \cdot \log^3(\Delta) + \poly(\phi)\bigr)$ operations. In the case $n \leq 2 \Delta - 1$, the problem can be solved by any polynomial-time in a fixed dimension algorithm. We refer to the recent result, due to V.~Reis \& T.~Rothvoss \cite{log_ILP}, which states that any ILP can be solved with $(\log n)^{O(n)} \cdot \poly(\phi)$ operations. Due to Lemmas \ref{ILPCF_to_ILPSF_lm} and \ref{ILPSF_to_ILPCF_lm} of Section \ref{connection_sec}, the same reasoning holds for the problems \ref{ILP-SF}. As a corollary, we have
\begin{theorem}[Consequence of Theorem \ref{Kriepke_m_th}, Theorem \ref{GribanovCanonicalWork_compplexity}, and the work \cite{log_ILP}]
    Any of the problems \ref{ILP-CF} and \ref{ILP-SF} with a generic matrix $A$ can be solved with
    $$
    (\log \Delta)^{O(\Delta)} \cdot \poly(\phi)\quad \text{operations.}
    $$
\end{theorem}
As an application of our results, we present a better computational complexity bound for the case $n \geq 2 \Delta - 1$ with respect to Theorem \ref{GribanovCanonicalWork_compplexity}, especially for the feasibility variant.
\begin{theorem}\label{generic_app_cor}
    Any of the problems \ref{ILP-CF} and \ref{ILP-SF} with generic matrix $A$ and $n \geq 2 \Delta - 1$ can be solved with
    $$
    \Delta^2 / 2^{\Omega(\sqrt{\log \Delta})} + \poly(\phi)\quad \text{operations.}
    $$
    The feasibility variant of the problem can be solved with $$
    O\bigl(\Delta \cdot \log^3(\Delta)\bigr) + \poly(\phi)\quad \text{operations.}
    $$
\end{theorem}

\subsection{The Case \(k=0\), R.~Gomory's Asymptotic Algorithm, and Expected ILP Complexity}\label{Gomory_sec}
% In the current subsection, we describe algorithms that have polynomial dependence both on $k$ and $n$. By this reason, there is no significant difference between the problems \ref{ILP-CF} and \ref{ILP-SF}, since they can be easily reduced to
In the works \cite{GomoryRelation,GomoryIntegerFaces,GomoryCombinatorialPoly}, R.~Gomory introduces his seminal \emph{asymptotic ILP algorithm} (for the full exposition of Gomory's work and additional algorithmic modifications, see the book \cite{HuBook} of T.~Hu). Applied to the problem \ref{ILP-SF}, the Gomory's asymptotic algorithm searches an optimal base $\BC$ of the LP relaxation and finds an optimal solution of the \emph{local ILP problem}
\begin{align}
&c^\top x \to \max\notag\\
&\begin{cases}
A  x = b\label{local_SF_prob}\tag{$\LocalILPSF$}\\
x_{\BC} \in \ZZ^{k}\\
x_{\NotBC} \in \ZZ_{\geq 0}^{n - k},
\end{cases}
\end{align} which can be obtained from the original problem, omitting the inequalities $x_{\BC} \geq 0$ (here, $\NotBC := \intint n \setminus \BC$). The optimal solution of \ref{local_SF_prob}, in turn, in many situations agrees with some optimal solution of the original problem \ref{ILP-SF}. This fact was empirically shown by R.~Gomory in his works. We will not cite exact sufficient conditions, described by R.~Gomory, and just mention that, for any integer $t \geq t_0$, where $t_0$ depends on $A$ and $b$, the asymptotic algorithm successfully finds a correct optimal solution of the problem $\max\{c^\top x \colon A x = t \cdot b,\, x \in \ZZ^n_{\geq 0}\}$. 

To solve the problem \ref{local_SF_prob}, R.~Gomory reduces it to the following \emph{group minimization problem:}
\begin{align}
    &w^\top x \to \min\notag\\
    &\begin{cases}
        \sum\limits_{i = 1}^d x_i \cdot g_i = g_0\\
        x \in \ZZ^{d}_{\geq 0},
    \end{cases}\label{group_min_prob}\tag{$\GroupMin$}
\end{align}
where $g_0,g_1, \dots, g_d$ are elements of an Abelian group $\GC$, $d = n - k$, and $\abs{\GC} \leq \Delta$. Next, R.~Gomory shows that  the problem \ref{group_min_prob} can be solved by a DP-based algorithm with $O\bigl(d \cdot \abs{\GC}\bigr)$ group operations. The group $\GC$ is explicitly represented as a direct sum of at most $\log_2 \Delta$ cyclic groups. Consequently, the total arithmetic complexity of the Gomory's asymptotic algorithm can be upper bounded by
\begin{equation*}
    T_{LP} + O\bigl( (n-k) \cdot \Delta \cdot \log \Delta\bigr)
\end{equation*}
where $T_{\LP}$ is the computational complexity to solve the relaxed LP problem. Due to the straightforward sorting argument, we can slightly refine the last bound with
\begin{equation}\label{Gomorys_alg_complexity}
    T_{LP} + O\bigl( \min\{n-k,\Delta\} \cdot \Delta \cdot \log \Delta\bigr).
\end{equation}
% The main idea of the algorithm is embedding of the problem into a special group, whose order is bounded by $\Delta$. Elements of the group can be represented as integer vectors with components in $\intint[0]{\Delta-1}$ and dimension at most $\log \Delta$, that is the reason of the $\log \Delta$ multiplicative term in the complexity bound. The complexity bounds can be slightly improved, making elements of the group unique, which can be done using any optimal sorting algorithm, which, in turn, is negligible with respect to $T_{LP}$. An improved complexity bound becomes
% \begin{equation*}
%     T_{LP} + O\bigl( \min\{n,\Delta\} \cdot \Delta \cdot \log \Delta \bigr).
% \end{equation*}
% % \begin{multline*}
% %     T_{LP} + O\bigl( \min\{n,\Delta\} \cdot \Delta \cdot \log \Delta \bigr) = \\
% %     = O\bigl( \min\{n,\Delta\} \cdot \Delta \cdot \log \Delta \bigr) + \poly(\phi).
% % \end{multline*}
Since we do not worry too much about the values of $k$ and $n$, when we are working with the asymptotic algorithm, it can be easily applied to the problem \ref{ILP-CF}, using a straightforward reduction. To finish our brief survey of the asymptotic algorithm, we note that the asymptotic algorithm is much easier to describe and implement, when it is initially applied to the problem \ref{ILP-CF}. The full exposition of the R.~Gomory's approach from this point of view can be found in the work \cite[pages~48--61]{OnCanonicalProblems_Grib} of D.~Gribanov et al. 

Further, it is natural to wonder about the applicability of the Gomory's asymptotic algorithm in some probabilistic or averaging setting. Research in this direction was done in the works \cite{IntegralityNumber,SparsityAverage,DistributionsILP,OnCanonicalProblems_Grib}. We will follow to Gribanov et al. \cite{OnCanonicalProblems_Grib}, which gives a slightly refined analysis with respect to the paper \cite{IntegralityNumber}, due to J.~Paat, M.~Schl{\"o}ter \& R.~Weismantel, where a probabilistic approach was applied for the first time. 
\begin{definition} %Let $\BB$ be an optimal base of the relaxation of $\IP_{\leq}(A,b,c)$.
Let us consider the \ref{ILP-CF} problem. Let $\BC$ be some optimal base of the corresponding relaxed LP problem. The \ref{ILP-CF} problem is \emph{local} if the set of its optimal solutions coincides with the set of optimal solutions of the following local problem 
\begin{align}
&c^\top x \to \max\notag\\
&\begin{cases}
A_{\BC} x \leq b_{\BC}\label{local_CF_prob}\tag{$\LocalILPCF$}\\
x \in \ZZ^n,
\end{cases}
\end{align}
which is constructed from the original problem, omitting all the inequalities, except $A_{\BC} x \leq b_{\BC}$.
\end{definition}

\begin{definition}\label{probability_def}
Let $\Omega_{n,t} = \{b \in \ZZ^{n}\colon \|b\|_\infty \leq t\}$. Then, for $\AC \subseteq \ZZ^{n}$, we define
$$
\prob_{n,t}(\AC) = \cfrac{\abs{\AC \cap \Omega_{n,t}}}{\abs{\Omega_{n,t}}} \text{ and } \prob_n(\AC) = \liminf\limits_{t \to \infty} \prob_{n,t}(\AC).
$$
The \emph{conditional probability of $\AC$} with respect to $\GC$ is denoted by the formula 
$$
\prob_n(\AC\mid\GC) = \frac{\prob_n(\AC \cap \GC)}{\prob_n(\GC)}.
$$
\end{definition}
As it was noted in \cite{IntegralityNumber}, the functional $\prob_n(\AC)$ is not formally a probability measure, but rather a lower density function, found from number theory. Additionally, we will use the symbol $\prob$ instead of $\prob_n$, when the value of $n$ will be clear from the context.

\begin{theorem}[Theorem 17 in D.~Gribanov et al. \cite{OnCanonicalProblems_Grib}]\label{local_prob_th}
    For fixed $A$ and $c$, denote the problem \ref{ILP-CF} with the right-hand side $b$ by $\ILPCF(b)$. Let
    \begin{gather*}
        \FC = \bigl\{ b \colon \text{$\ILPCF(b)$ is feasible}  \bigr\},\\
        \GC = \bigl\{ b \colon \text{$\ILPCF(b)$ is local}  \bigr\}.
    \end{gather*}
Then $\prob(\GC \mid \FC) = 1$.
\end{theorem}
Note that a similar Theorem is valid for the \ref{ILP-SF} problems. We only need to define the class of local \ref{ILP-SF} problems in accordance with the problem \ref{local_SF_prob}.

By definition, any local \ref{ILP-CF} problem is equivalent to the problem \ref{local_CF_prob}. In turn, the last problem can be solved using Theorem \ref{main_ILP_th}. Combining the algorithm of Theorem \ref{main_ILP_th} with the Gomory's DP-based group minimization algorithm, we get the following 
\begin{corollary}\label{exp_app_cor}
In terms of Theorem \ref{local_prob_th}, with $\prob = 1$, any of the problems $\ILPCF(b)$ and $\ILPSF(b)$ can be solved by an algorithm with the arithmetic complexity bound
$$
O\left(\min\left\{ 
    \substack{
        d \cdot \Delta \cdot \log \Delta, \\ 
        \Delta^2/2^{\const \cdot \sqrt{\log \Delta}}
    } 
\right. \right) + \poly(\phi),
$$ where $d := n$ for \ref{ILP-CF} and $d := n-k$ for \ref{ILP-SF}.
\end{corollary}
The new computational complexity bound is better than the bound \eqref{Gomorys_alg_complexity} for the cases, when $d = \Omega\bigl(\Delta / \polylog(\Delta)\bigr)$.

% \begin{definition} %Let $\BB$ be an optimal base of the relaxation of $\IP_{\leq}(A,b,c)$.
% Let us consider the \ref{ILP-SF} problem. Let $\BC$ be some optimal base of the corresponding relaxed LP problem. The \ref{ILP-SF} problem is \emph{local} if the set of its optimal solutions coincides with the set of optimal solutions of the following local problem 
% \begin{gather}
% c^\top x \to \max\notag\\
% \begin{cases}
% A  x = b\label{local_SF_prob}\tag{LOCAL-ILP-SF}\\
% x_{\BC} \in \ZZ^{k}\\
% x_{\NotBC} \in \ZZ_{\geq 0}^{n - k}
% \end{cases}
% \end{gather}
% that is induced from the original problem by omitting the inequalities $x_{\BC} \geq \BZero$.
% \end{definition}

% the arithmetic complexity of the Gomory's asymptotic algorithm can be estimated as $O\bigl(n\cdot \log n + \min\{n,\Delta\} \cdot \Delta \cdot \log \Delta\bigr)$ finds an optimal integer point of the simplified ILP problem associated with the corner polyhedron of the corresponding relaxed LP problem. 

\subsection{The Case \(k=1\) and ILP on Simplices}

One of the simplest ILP problems, which still remains NP-hard, is the ILP problem on an $n$-dimensional simplex. We consider the simplices, defined by systems of the type \ref{ILP-CF}, since this definition is more general than \ref{ILP-SF} (a clarification of this is given in Section \ref{connection_sec}). We call a simplex, defined by a $\Delta$-modular system \ref{ILP-CF}, as a \emph{$\Delta$-modular simplex.}

Surprisingly, many NP-hard problems, restricted on $\Delta$-modular simplices with $\Delta = \poly(\phi)$, can be solved by a polynomial-time algorithms. 
% Definitely, due to Gribanov et al. \cite[Theorem~14]{OnCanonicalProblems_Grib} the integer feasibility problem can be solved with $O\bigl(\min\{n,\Delta\} \cdot \Delta \cdot \log(\Delta) + \poly(\phi) \bigr)$ operations. 
Due to Theorem \ref{GribanovCanonicalWork_compplexity}, the integer linear optimization problem can be solved with $O\bigl(\Delta^2 \cdot \log^3(\Delta)  + \poly(\phi)\bigr)$ operations. 
Due to Gribanov \& Malyshev \cite{Counting_FPT_Delta} and a modification from Gribanov, Malyshev \& Zolotykh \cite{SparseILP_Gribanov}, the number of integer points in the simplex can be counted with $O(n^4 \cdot \Delta^3)$ operations. Due to Basu \& Jiang \cite{IntVertexEnumSimplex}, the vertices of $\Delta$-modular simplex can be enumerated by a polynomial-time algorithm, which can be used in convex optimization. A refined computational complexity bound is given in \cite[Section~3.3]{OnCanonicalProblems_Grib}. Due to Gribanov, Malyshev, Pardalos \& Veselov \cite{FPT_Grib} (see also \cite{WidthSimplex_Grib}), the lattice width can be computed with $\poly(n,\Delta,\Delta_{ext})$ operations, where $\Delta_{ext} = \Delta(A\,b)$ and $(A\,b)$ is the system's extended matrix. Additionally, for empty simplices, the computational complexity dependence on $\Delta_{ext}$ can be avoided, which gives the bound $\poly(n,\Delta)$. A similar result for lattice-simplices, defined by convex hulls of their vertices, is presented in Gribanov, Malyshev \& Veselov \cite{WidthConv_Grib}. Finally, it was shown in Gribanov \cite{SimplexEquiv_Gribanov} that all unimodular equivalence classes of $\Delta$-modular simplices can be enumerated by a polynomial-time algorithm, as well as the problem to check the unimodular equivalence of two given simplices. 

In this work, we give better computational complexity bounds for ILP on simplices with respect to Theorem \ref{GribanovCanonicalWork_compplexity}, especially for the feasibility problem.
\begin{theorem}
    Assuming that the system $A x \leq b$ defines a simplex, the ILP problem \ref{ILP-CF} can be solved with
    $$
    \Delta^2 / 2^{\Omega(\sqrt{\log \Delta})} + \poly(\phi) \quad \text{operations.}
    $$
    The feasibility problem can be solved with
    $$
    O\bigl(\Delta \cdot \log^3 (\Delta) \bigr) + \poly(\phi) \quad \text{operations.}
    $$
\end{theorem}
This Theorem is the direct consequence of Theorem \ref{main_ILP_th}. Surprisingly, the feasibility problem on a simplex can be parameterized by a different parameter $\Delta_{\min}$, which denotes the minimum non-zero absolute value of $\rank(A) \times \rank(A)$ sub-determinants of $A$. Definitely, assuming that the rows of the system $A x \leq b$ define facets of a simplex, we can find a base $\BC$ of $A$, such that $\abs{\det(A_{\BC})} = \Delta_{\min}$. Then the feasibility problem on the simplex is equivalent to the minimization problem $\min\{c^\top x \colon A_{\BC} x \leq b_{\BC},\, x \in \ZZ^n\}$, where $c$ agrees with an outer-normal vector, corresponding to the facet that is `opposite' to $\BC$. This problem can be solved, using Theorem \ref{main_ILP_th} or the Gomory's group minimization algorithm (see Section \ref{Gomory_sec}), which proves 
\begin{theorem}
    Assume that the system $A x \leq b$ defines a simplex $\SC$ and lines of the system define its facets. Then the integer feasibility problem on $\SC$ can be solved with
    $$
    O\left(\min\left\{ 
        \substack{
            n \cdot \Delta_{\min} \cdot \log \Delta_{\min},\\
            \Delta^2_{\min}/2^{\const \cdot \sqrt{\log \Delta_{\min}}} 
        } 
    \right. \right) + \poly(\phi) \quad\text{operations.}
    $$
\end{theorem}

\section{Connection of the Problems \texorpdfstring{\ref{ILP-CF} and \ref{ILP-SF}}{ILP-CF and ILP-SF}
}\label{connection_sec}

In the current Section, we are going to clarify our decision to use two definitions of ILP problems: \ref{ILP-CF} and \ref{ILP-SF}. The most common formulation, which can be found in most of the works, devoted to the ILP topic, is the \ref{ILP-SF} formulation. It is natural to consider the computational complexity of \ref{ILP-SF} with respect to the number $k$ of lines in the system, which was pioneered in the seminal work \cite{Papadimitriou}, due to C.~Papadimitriou. In turn, the formulation \ref{ILP-CF} is clearer from the geometrical point of view, but it can be easily transformed to \ref{ILP-SF}, introducing some new integer variables. However, this trivial reduction has a downside: it changes the value of the parameter $k$ and the dimension of the corresponding polyhedra.

A less trivial reduction that preserves the parameters $k$, $\Delta$, and the dimension of the corresponding polyhedra, is described in Gribanov et al. \cite{OnCanonicalProblems_Grib}. This reduction connects the problem \ref{ILP-CF} with the \emph{ILP problem in the standard form with modular constraints}, which strictly generalizes the problem \ref{ILP-SF} and can be formulated in the following way.
\begin{definition}\label{ModularILPSF_def}
    Let $A \in \ZZ^{k \times n}$ and $G \in \ZZ^{(n-k)\times n}$, such that $\binom{A}{G}$ is an integer $n \times n$ unimodular matrix. Additionally, let $S \in \ZZ^{(n-k)\times(n-k)}$ be a matrix, reduced to the SNF, $g \in \ZZ^{n-k}$, $b \in \ZZ^k$, $c \in \ZZ^n$. The \emph{ILP problem in the standard form of the co-dimension $k$ with modular constraints} is formulated as follows:
    \begin{align}
        & c^\top x \to \max\notag\\
        & \begin{cases}
            A x = b\\
            G x \equiv g \pmod{S \cdot \ZZ^n}\\
            x \in \ZZ^n_{\geq 0}.
        \end{cases}\tag{$\ModularILPSF$}\label{MOD-ILP-SF}
    \end{align}
    Here the notation $G x \equiv g \pmod{S \cdot \ZZ^n}$ denotes that, for each $i \in \intint{(n-k)}$, there exists $z \in \ZZ$ such that $G_{i *} x = g_i + S_{i i} z$.
\end{definition}
% \begin{definition}
%     Let $A \in \ZZ^{k \times n}$ and $G \in \ZZ^{(n-k)\timesn}$, such that $\binom{A}{G}$ is an integer $n \times n$ unimodular matrix. Let, additionally, $S \in \ZZ^{(n-k)\times(n-k)}$ be a matrix reduced to the SNF
    
%     $\rank(A) = k$, $c \in \ZZ^n$, $b \in \ZZ^k$. Additionally, let $\GC$ be a finite Abelian group in additive notation and $g_0, g_1, \dots, g_n$ be elements of $\GC$. The \emph{generalized ILP problem in the standard form of codimension $k$} is formulated as follows:
%     \begin{align}
%         & c^\top x \to \max\notag\\
%         & \begin{cases}
%             A x = b\\
%             \sum\limits_{i=1}^n x_i g_i = g_0\\
%             x \in \ZZ^n_{\geq 0}.
%         \end{cases}\tag{$\GenILPSF$}\label{GEN-ILP-SF}
%     \end{align}
% \end{definition}
Due to this reduction, the problem \ref{ILP-CF} is strictly more general, since some exemplars of the \ref{ILP-CF} problem can be reduced to \ref{MOD-ILP-SF} problems, equipped with a non-trivial matrix $G$, which, in turn, can not be represented by \ref{ILP-SF} type problems, preserving $k$, $\Delta$ and the dimension. The corresponding example could be found in \cite[Reamark~4]{OnCanonicalProblems_Grib}. From this point of view, it is more consistent to consider the \ref{ILP-CF} problems. 

Therefore, keeping in mind both the facts that the formulation \ref{ILP-SF} is more user-friendly and the formulation \ref{ILP-CF} is more general, we have decided to consider the both formulations in our work. Let us recall the formal description of the outlined reduction, which is given by the following Lemmas.
\begin{lemma}[Lemma~4, Gribanov et al. \cite{OnCanonicalProblems_Grib}]\label{ILPCF_to_ILPSF_lm}
    For any instance of the \ref{ILP-CF} problem, there exists an equivalent instance of the \ref{MOD-ILP-SF} problem
    \begin{align*}
        & \hat c^\top x \to \min\\
        & \begin{cases}
            \hat A x = \hat b\\
            G x \equiv g \pmod{S \cdot \ZZ^n}\\
            x \in \ZZ_{\geq 0}^{n+k},
        \end{cases}
    \end{align*}
    with $\hat A \in \ZZ^{k \times (k+n)}$, $\rank(\hat A) = k$, $\hat b \in \ZZ^k$, $\hat c \in \ZZ^{n+k}$, $G \in \ZZ^{n \times (n+k)}$, $g \in \ZZ^n$, $S \in \ZZ^{n \times n}$. Moreover, the following properties hold:
    \begin{enumerate}
        \item $\hat A \cdot A = \BZero_{k \times n}$, $\Delta(\hat A) = \Delta(A)/\Delta_{\gcd}(A)$;
        \item $\abs{\det(S)} = \Delta_{\gcd}(A)$;
        \item There exists a bijection between rank-order sub-determinants of $A$ and $\hat A$;
        \item The map $\hat x = b - A x$ is a bijection between integer solutions of both problems;
        \item If the original relaxed LP problem is bounded, then we can assume that $\hat c \geq \BZero$;
        \item The reduction is not harder than the computation of the SNF of $A$.
    \end{enumerate}
\end{lemma}

\begin{lemma}[Lemma~5, Gribanov et al. \cite{OnCanonicalProblems_Grib}]\label{ILPSF_to_ILPCF_lm}
    For any instance of the \ref{MOD-ILP-SF} problem, there exists an equivalent instance of the \ref{ILP-CF} problem
    \begin{align*}
        &\hat c^\top x \to \max\\
        &\begin{cases}
            \hat A x \leq \hat b\\
            x \in \ZZ^d
        \end{cases}
    \end{align*}
    with $d = n - k$, $\hat A \in \ZZ^{(d+k)\times d}$, $\rank(\hat A) = d$, $\hat c \in \ZZ^d$, and $b \in \ZZ^{d+k}$. Moreover, the following properties hold:
    \begin{enumerate}
        \item $A \cdot \hat A = \BZero_{k \times d}$, $\Delta(\hat A) = \Delta(A) \cdot \abs{\det(S)}$;
        \item $\Delta_{\gcd}(\hat A) = \abs{\det(S)}$;
        \item There exists a bijection between rank-order sub-determinants of $A$ and $\hat A$;
        \item The map $x = \hat b - \hat A \hat x$ is a bijection between integer solutions of both problems;
        \item The reduction is not harder than the inversion of an integer unimodular $n \times n$ matrix $\binom{A}{G}$.
    \end{enumerate}
\end{lemma}

\subsection{Generalized ILP Problem and the Proof of Theorem \ref{main_ILP_th}}\label{GenILP_subs}

To prove our main results and to make the proofs shorter and clearer, we solve a bit more general and natural problem, which, due to the described Lemmas, generalizes all the problems \ref{ILP-CF}, \ref{MOD-ILP-SF} and \ref{ILP-SF}. 
\begin{definition}\label{GroupILPSF_def}
    Let $A \in \ZZ^{k \times n}$ with $\rank(A) = k$ and $\Delta_{\gcd}(A) = 1$, $b \in \ZZ^k$, $c \in \ZZ^n$. Additionally, let $\GC$ be a finite Abelian group in the additive notation  represented by its basis and let $g_0,g_1, \dots, g_n \in \GC$. The \emph{generalized ILP problem is the standard form of the co-dimension $k$} is formulated as follows:
    \begin{align}
        &c^\top x \to \max\notag\\
        &\begin{cases}
            A x = b\\
            \sum\limits_{i=1}^n x_i \cdot g_i = g_0\\
            x \in \ZZ_{\geq 0}^n.
        \end{cases}\tag{$\GenILPSF$}\label{GEN-ILP-SF}
    \end{align}
\end{definition}

It is not straightforwardly clear how to define a linear programming relaxation of the \ref{GEN-ILP-SF} problem. So, we emphasize it as a stand-alone definition:
\begin{definition}\label{BLPSF_def}
In terms of the previous definition, the problem 
\begin{gather*}
    c^\top x \to \max \notag\\
    \begin{cases}
    A x = b\\
    x \in \RR_{\geq 0}^n.
    \end{cases}
\end{gather*} is called \emph{the linear programming relaxation} of the \ref{GEN-ILP-SF} problem.
\end{definition}

The problem \ref{GEN-ILP-SF} problem can be solved, using the following generalized Theorem, whose proof requires an additional effort and will be given in \Vref{genILP_main_th_proof}. 
\begin{restatable}{theorem}{mainGenILPTh}
\label{genILP_main_th}
    Consider the problem \ref{GEN-ILP-SF}. Denote $r = \abs{\GC}$, the problem can be solved with 
    $$
        O(\log k)^{2k} \cdot (r \Delta)^2 / 2^{\Omega\bigl(\sqrt{\log(r\Delta)}\bigr)} + 2^{O(k)} \cdot \poly(\phi) 
    $$
    operations with elements from $\ZZ^k \times \GC$. The feasibility variant of the problem can be solved with
    $$
    O(\log k)^{k} \cdot (r \Delta) \cdot \log^2(r\Delta) + 2^{O(k)} \cdot \poly(\phi) \quad\text{operations in $\ZZ^k \times \GC$.}
    $$
\end{restatable}

Now, let us show how the proof of our main \Cref{main_ILP_th} 
% and \ref{main_ILP_logD_th} 
can be deduced from \Cref{genILP_main_th}. Recall the definition of \Cref{main_ILP_th}.

\mainILPTh*

\begin{myproof}[Proof of Theorem \ref{main_ILP_th}]\label{main_ILP_th_proof}
    Assume first that the problem \ref{ILP-SF} is considered. Clearly, it is a partial case of the problem \ref{GEN-ILP-SF} with the trivial group $\GC$, consisting of a neutral element. Since the cost of a single operation in $\ZZ^k \times \GC$ is $k$ and since $O(\log k)^{2k}$ dominates $k$, Theorem \ref{main_ILP_th} is the straight corollary of Theorem \ref{genILP_main_th}, for the considered case.

    Consider now the problem \ref{ILP-CF}. Due to Lemma \ref{ILPCF_to_ILPSF_lm}, it can be reduced to an equivalent instance of the problem \ref{MOD-ILP-SF} with $\hat n = n + k$, $\hat d = n$ and $\Delta(\hat A) = \Delta / \Delta_{\gcd}$, which, in turn, is a partial case of the problem \ref{GEN-ILP-SF} with the group $\GC$, represented as $\GC = \ZZ^n / S \cdot \ZZ^n$. Clearly, $r = \abs{\GC} = \det(S) = \Delta_{\gcd} \leq \Delta$. Since an element of $\ZZ^k \times \GC$ can be represented by an integer vector with at most $k + \log_2 r \leq k \cdot \log_2 \Delta$ components, which is dominated by $O(\log k)^{2k} \cdot 2^{\Omega\bigl(\sqrt{\log \Delta}\bigr)}$, the resulting number of operations can be expressed by the formula
    $$
    O(\log k)^{2k} \cdot \Delta^2 / 2^{\Omega\bigl(\sqrt{\log \Delta}\bigr)}
    $$ 
    modulo the additive term $2^{O(k)} \cdot \poly(\phi)$.

    With respect to the feasibility variant of the problem \ref{ILP-CF}, the number of operations can be expressed by the formula 
    \begin{multline*}
        O(\log k)^k \cdot \Delta \cdot \log^2 \Delta \cdot (k + \log r) = \\
        O(\log k)^k \cdot \Delta \cdot \log^3 \Delta.
    \end{multline*}
    The formulas above satisfy the desired computational complexity bounds, which finishes the proof of Theorem \ref{main_ILP_th}.
\end{myproof}

% \proof{\bf Proof of Theorem \ref{main_ILP_logD_th}.\label{main_ILP_logD_th_proof}}
%     The proof of Theorem \ref{main_ILP_logD_th} is completely similar with the only difference that we use the conclusion of Remark \ref{n_delta_bound_rm} instead of the direct application of Theorem \ref{genILP_main_th}. \Halmos
% \endproof

\section{Convolution and Group Ring}\label{conv_sec}

In the current section we recall definitions of the generalized tropical and Boolean convolution problems with respect to finite Abelian groups. These problems will arise later in the proof of \Cref{genILP_main_th} as subproblems. The results of this section are interesting in their own.

Let $\GC$ be a finite Abelian group (in the additive notation). Consider a \emph{group ring $\RC[\GC]$} over a commutative ring $\RC = \bigl(\RC, \oplus, \otimes\bigr)$, which is a free $\RC$-module and the same time a ring. 
% Denote the ring operations by $\oplus$ and $\otimes$. 
The \emph{group ring} $\RC[\GC]$ is the set of functions $\alpha \colon \GC \to \RC$ of a finite support, where the module scalar product $c \otimes \alpha$ of a scalar $c \in \RC$ and a function $\alpha$ is defined as the mapping $x \to c \otimes \alpha(x)$, and the module sum $\alpha \oplus \beta$ of two mappings 
$\alpha$ and $\beta$ is defined as the mapping $x \to \alpha(x) \oplus \beta(x)$. The multiplication (convolution) $\alpha \star \beta$ of two functions $\alpha$ and $\beta$ is defined by the mapping 
\begin{equation*}
x \to \bigoplus\limits_{\substack{g_1 + g_2 = x\\ g_1,g_2 \in \GC}} \alpha(g_1) \otimes \beta(g_2) = \bigoplus\limits_{g \in \GC} \alpha(g) \otimes \beta(x - g).     
\end{equation*}
Throughout the text, we relax the requirements on $\RC$ and assume that $\RC$ is a commutative semiring, which does not affect the correctness of the definition. 

% \begin{definition}
%     Let $\GC$ be an Abelian group, represented as the direct sum of subgroups $\GC = \QC \oplus \HC$. For a function $\alpha \in \RC[\GC]$ and element $h \in \HC$, the \emph{relative support $\supp_{h}(\alpha)$} is defined by the formula
% $$
% \supp_{h}(\alpha) = \bigl\{q \in \QC \colon \alpha(q + h) \not= 0 \bigr\}.
% $$
% The relative support can be defined with respect to the whole sub-group $\HC$:
% $$
% \supp_{\HC}(\alpha) = \bigcup\limits_{h \in \HC} \supp_{h}(\alpha).
% $$
% % The following formula golds:
% % \begin{equation}\label{subgroup_decomp_eq}
% %     \abs{\supp(\alpha)} = \sum\limits_{h \in \HC} \abs{\supp_{h}(\alpha)}.
% % \end{equation}
% \end{definition}

Additionally, let us recall the definition of a \emph{basis of an Abelian group}.
\begin{definition}
We say that an Abelian group $\GC$ is defined by a \emph{basis $b_1, b_2, \dots, b_s$}, if any element $g \in \GC$ can be uniquely represented as a combination
$$
g = i_1 \cdot b_1 + i_2 \cdot b_2 + \dots + i_s \cdot b_s,
$$ where $i_j \in \intint[0]{\rank(b_j)-1}$, for any $j \in \intint s$.
\end{definition}

\subsection{The Tropical Semiring}
In the current subsection, we consider the group semiring $\RC_{(\min,+)}[\GC]$, where $\RC_{(\min,+)} = \bigl(\ZZ \cup \{+\infty\},\min,+\bigr)$ is tropical semiring, also known as the $(\min,+)$-semiring, and $\GC$ is an arbitrary finite Abelian group. Everywhere in the current subsection, we assume that elements $\alpha \in \RC_{(\min,+)}[\GC]$ are represented by $w$-bit integers on their support. 
The main nontrivial fact that we use is given by the following
\begin{theorem}[R.~Williams \cite{APSPViaCircuitComplexity}, T.~Chan \& R.~Williams \cite{DerandomRazbSmol}]\label{Williams_th}
    Let $A,B \in (\ZZ \cup \{+\infty\})^{n \times n}$. Then
    % assuming that each input integer number can be encoded with $w$-bit machine word,  
    the product $A \cdot B$ in $\RC_{(\min,+)}$ can be computed with $O\bigl(n^3 / 2^{\Omega(\sqrt{\log n})} \bigr)$ operations with $O(w)$-bit integers.
\end{theorem}

Due to D.~Bremner et al. \cite{NecklacesConv}, the convolution of two finite sequences can be reduced to $\RC_{(\min,+)}$-multiplication of square matrices. The reduction complexity is $O\bigl( T(\sqrt{n}) \cdot \sqrt{n} \bigr)$, where $T(\cdot)$ is the computational complexity of a square matrix multiplication in $\RC_{(\min,+)}$. In turn, due to Theorem \ref{Williams_th}, $T(n) = O(n^3/2^{\Omega(\sqrt{\log n})})$, which leads to an $O(n^2/2^{\Omega(\sqrt{\log n})})$-time algorithm for the standard $(\min,+)$-convolution. More formally:
\begin{corollary}[D.~Bremner et al. \cite{NecklacesConv}]\label{standard_minconv_th}
    Let $\alpha = \{\alpha_i\}_{i = 0}^{n-1}$ and $\beta = \{\beta_i\}_{i = 0}^{n-1}$ be input sequences with elements in $\ZZ \cup \{+\infty\}$. Then
    % assuming that each input integer number can be encoded with $w$-bit machine word, 
    the convolution $\alpha \star \beta$ can be computed with $O\bigl(n^2 / 2^{\Omega(\sqrt{\log n})} \bigr)$ operations with $O(w)$-bit integers.
\end{corollary}

The following simple Lemma, due to Gribanov et al. \cite{OnCanonicalProblems_Grib}, gives a fast convolution algorithm in $\RC_{(\min,+)}[\GC]$ with respect to an arbitrary cyclic group $\GC$. For completeness, we present a short proof.
\begin{lemma}\label{cyclic_minconv_lm}
    Let $\GC = \langle g \rangle$ and $n := \abs{\GC} < \infty$. Given $\alpha, \beta \in \RC_{(\min,+)}[\GC]$, the convolution $\gamma = \alpha \star \beta$ can be computed with $O(n^2/2^{\Omega(\sqrt{\log n})})$ operations with $O(w)$-bit integers and group-operations in $\GC$. 
\end{lemma}
\begin{myproof}
    The problem can be easily reduced to the standard convolution of two sequences, indexed by numbers in $\intint[0]{n-1}$. Let $\{\hat \alpha_i\}_{i =0}^{2n-1}$ and $\{\hat \beta_i\}_{i = 0}^{2n-1}$ be the sequences, defined by
    \begin{gather*}
        \hat \alpha_i = \alpha(i \cdot g)\\
        \hat \beta_j = \begin{cases}
            \beta(j \cdot g),\text{ for } j \in \intint[0]{n-1}\\
            +\infty,\text{ for other values of $j$}.
        \end{cases}
    \end{gather*}
    Let $\{\hat \gamma_i\}_{i=0}^{2n - 1}$ be the convolution of $\hat \alpha$ and $\hat \beta$. Then we can easily see that
    \begin{multline}\label{cyclic_reduction}
        \gamma(j \cdot g) = \min\limits_{i \in\intint[0]{n-1}}\bigl\{\alpha\bigl((j-i) \cdot g\bigr) + \beta(i\cdot g)\bigr\} = \\
        = \min\limits_{i \in\intint[0]{n-1}}\bigl\{\alpha\bigl((n+j-i) \cdot g\bigr) + \beta(i\cdot g)\bigr\} = \min\limits_{i \in\intint[0]{n+j}}\bigl\{\hat \alpha_{n+j-i} + \hat \beta_{i}\bigr\} = \\
        = \hat \gamma_{n+j}, \quad \text{for } j \in \intint[0]{n-1}.
    \end{multline} 
    Due to Corollary \ref{standard_minconv_th}, $\hat \gamma$ can be computed with $O\bigl(n^2 / 2^{\Omega(\sqrt{\log n})}\bigr)$ operations. Finally, due to the formula \eqref{cyclic_reduction}, we can compute $\gamma$, using $O(n)$ operations.
\end{myproof}

The following Lemma can be used to develop fast convolution algorithm with respect to direct sums of Abelian groups under assumption that there exists a fast convolution algorithm for one of the summands.
\begin{lemma}\label{conv_decomp_lm}
     Let $\GC = \QC \oplus \HC$, where the Abelian groups $\QC$ and $\HC$ are represented by their bases. Denote $n = \abs{\QC}$, $m = \abs{\HC}$, and assume that the convolution in $\RC_{(\min,+)}[\QC]$ can be computed in $T(n)$ operations with $T(n) = \Omega(n)$.
     
     Then, for given $\alpha,\beta \in \RC_{(\min,+)}[\GC]$, the convolution $\gamma = \alpha \star \beta$ can be computed in $O\bigl(T(n) \cdot m^2\bigr)$ operations with $O(w)$-bit integers and group-operations in $\GC$. 
\end{lemma}
\begin{myproof}
    Let us fix an element $h^* \in \HC$ and show how to compute $\gamma(q + h^*)$ through all $q \in \QC$ with only $O\bigl(T(n) \cdot m\bigr)$ operations. For any $h \in \HC$, we define the functions $\hat \alpha_h \in \RC_{(\min,+)}[\QC]$ and $\hat \beta_h \in \RC_{(\min,+)}[\QC]$ by the formulae $\hat \alpha_h(q) = \alpha(q + h^*-h)$ and $\hat \beta_h(q) = \beta(q + h)$, and let $\hat \gamma_h = \hat \alpha_h \star \hat \beta_h$. We have
    \begin{multline}\label{decomp_conv}
        \gamma(q+h^*) = \min\limits_{h' \in \HC} \min\limits_{q'\in \QC} \bigl\{ \alpha(q-q' + h^*-h') + \beta(q' + h') \bigr\} = \\
        = \min\limits_{h' \in \HC} \min\limits_{q'\in \QC} \bigl\{ \hat \alpha_{h'}(q-q') + \hat \beta_{h'}(q')\bigr\} = \\= \min\limits_{h' \in \SC_{\HC}} \hat \gamma_{h'}(q).
    \end{multline}
    The algorithm is as follows. We compute $\hat \gamma_h$, for each $h \in \HC$, which takes $O\bigl(T(n) \cdot m\bigr)$ operations. After that, we compute $\gamma(q + h^*)$, for each $q \in \QC$, using the formula \eqref{decomp_conv}, which takes $O(n \cdot m)$ operations. Assuming that $T(n) = \Omega(n)$, we have the desired computational complexity bound for any fixed $h^*$. Enumerating $h^* \in \HC$, the total computational complexity becomes $O\bigl(T(n) \cdot m^2\bigr)$.
\end{myproof}

The following technical Lemma and its corollary reduce the multiplication of rectangular matrices to the multiplication of square matrices. Such reductions are simple and well known.
\begin{lemma}\label{mm_lm}
    Let $A \in \RC_{(\min,+)}^{m \times n}$, $B \in \RC_{(\min,+)}^{n \times t}$ and $t \leq m \leq n$. Then the multiplication $A \cdot B$ can be done with $O\bigl(\frac{n}{t} \cdot T(t,m) + n \cdot m\bigr)$ operations in with $O(w)$-bit integers, where $T(m,t)$ denotes the multiplication complexity of two matrices of the orders $t \times t$ and $t \times m$ respectively. 
\end{lemma}
\begin{myproof}
    We decompose $A$ on approximately $n/t$ blocks of the size $m \times t$. Here, we can assume that $t \mid m$, since we can augment the last blocks by $+\infty$. Similarly, we decompose $B$ on $n/t$ blocks of the size $t \times t$. It is illustrated by the formula
    $$
    A \cdot B = \begin{pmatrix} A_1 & A_2 & \dots & A_{\lceil n/t \rceil}\end{pmatrix} \cdot \begin{pmatrix}
        B_1 \\
        B_2 \\
        \dots \\
        B_{\lceil n/t \rceil}
    \end{pmatrix}.
    $$

    The products $C_1 := A_1 \cdot B_1, C_2 = A_2 \cdot B_2, \dots, C_{\lceil n/t \rceil} := A_{\lceil n/t \rceil} \cdot B_{\lceil n/t \rceil}$, which are $(m \times t)$-order matrices, can be constructed with $O(n/t) \cdot T(t,m)$ operations. Finally, it can be directly checked that 
    $$
    (A \cdot B)_{i j} = \min\limits_{k \in \intint{\lceil n/t \rceil}} (C_k)_{i j}.
    $$ 
    Therefore, the answer matrix $A \cdot B$ can be constructed with additional $O\bigl(m \cdot t \cdot \frac{n}{t}\bigr) = O(n \cdot m)$ operations. 
\end{myproof}

\begin{corollary}\label{mm_cor}
    Let $A \in \RC_{(\min,+)}^{m \times n}$, $B \in \RC_{(\min,+)}^{n \times t}$, and $t \leq m \leq n$. Then the multiplication $A \cdot B$ can be done with $O\bigl(\frac{n \cdot m}{t^2} \cdot T(t) + n \cdot m\bigr)$ operations with $O(w)$-integers, where $T(t)$ denotes the multiplication complexity of two matrices of the order $t \times t$.
\end{corollary}

\begin{myproof}
    Due to the previous Lemma, the original problem can be reduced to multiplication of matrices of the orders $t \times t$ and $t \times m$, respectively. Such multiplications, in turn, are equivalent to $m/t$ multiplications of $(t \times t)$-order matrices. The total computational complexity can be expressed by the following formula, which gives the desired computational complexity bound: 
    \begin{multline*}
    O\left( \frac{n}{t} \cdot T(t,m) + n \cdot m \right) = O\left( \frac{n}{t} \cdot\bigl( \frac{m}{t} \cdot T(t) \bigr) + n \cdot m \right) = \\
    = O\left( \frac{n \cdot m}{t^2} \cdot T(t) + n \cdot m\right).
    \end{multline*}
\end{myproof}
    
\begin{theorem}\label{group_min_conv_th}
Assume that the group $\GC$ is finite and the basis $b_1, b_2, \dots, b_s$ of $\GC$ is given. Then, for given $\alpha, \beta \in \RC_{(\min,+)}[\GC]$, the convolution $\gamma = \alpha \star \beta$ can be computed with $O(n^2/2^{\Omega(\sqrt{\log n})})$ operations with $O(w)$-bit integers and group-operations in $\GC$, where $n = \abs{\GC}$. 
\end{theorem}
\begin{myproof}
    Denoting $r_i = \rank(b_i)$, for $i \in \intint s$, assume that $r_1 \geq r_2 \geq \dots \geq r_s$. Let the value $k \geq 1$ be chosen such that 
    \begin{equation*}
        \begin{cases}
            r_1 \cdot r_2 \cdot \ldots \cdot r_k \geq \sqrt{n},\\
            r_1 \cdot r_2 \cdot \ldots \cdot r_{k-1} < \sqrt{n}.
        \end{cases}
    \end{equation*}
    Here, we assume that the empty product is equal to $0$. Let us consider the cases: $k = 1$ and $k \geq 2$.
    
    {\bf Case 1: $k=1$ and $r_1 \geq \sqrt{n}$.} In this case, we look on the group $\GC$ as $\GC = \langle b_1 \rangle \oplus \HC$, where $\HC = \langle b_2 \rangle \oplus \dots \oplus \langle b_s \rangle$. Due to Lemma \ref{cyclic_minconv_lm}, the convolution in $\RC_{(\min,+)}[\langle b_1 \rangle]$ can be done with $O(r_1^2 / 2^{\Omega(\sqrt{\log r_1})}) = O(r_1^2 / 2^{\Omega(\sqrt{\log n})})$ operations. Consequently, Lemma \ref{conv_decomp_lm} gives the desired computational complexity bound.
    
    {\bf Case 2: $k \geq 2$.} Since $r_k \leq r_i$, for each $i \in \intint k$, and, due to the construction of $k$, it follows that 
    \begin{gather*}
    n^{1/4} \leq t := r_1 \cdot r_2 \cdot \ldots \cdot r_{k-1} \leq n^{1/2} \quad\text{and}\\
    n^{1/2} \leq m := r_k \cdot r_{k+1} \cdot \ldots \cdot r_s \leq n^{3/4}.    
    \end{gather*}
    Now, we look on our group as $\GC = \QC \oplus \HC$, where $\QC = \langle b_1 \rangle \oplus \dots \oplus \langle b_{k-1} \rangle$ and $\HC = \langle b_k \rangle \oplus \dots \oplus \langle b_s \rangle$. Note that $\abs{\QC} = t$ and $\abs{\HC} = m$. Construct the matrices $A \in \RC_{(\min,+)}^{m \times n}$ and $B \in \RC_{(\min,+)}^{n \times k}$ in the following way: the rows of $A$ are indexed by elements of $\QC$ and the columns of $B$ are indexed by elements of $\HC$. The elements of $A$ and $B$ are given by the following formulae:
    \begin{gather*}
        A_{q *} = \Bigl(\alpha(q - g)\Bigr)_{g \in \GC},\\
        B_{* h} = \Bigl(\beta(h + g)\Bigr)_{g \in \GC}.
    \end{gather*}
    Construction of $A$ and $B$ costs $O\bigl(n \cdot (m + t)\bigr) = O(n^{7/4})$ operations. In turn, the elements of $A \cdot B$ are indexed by pairs $(q,h) \in \QC \times \HC$ and it directly follows that
    \begin{equation*}
        (A \cdot B)_{q, h} = \min\limits_{g \in \GC} \bigl\{ \alpha(q - g) + \beta(h + g) \bigr\} = \gamma(q + h).
    \end{equation*}
    Therefore, since any element $g \in \GC$ can be uniquely represented as $g = q + h$, the computation of $\gamma$ has been reduced to the computation of $A \cdot B$.
    
    Note that $t \leq m \leq n$. Consequently, due to Corollary \ref{mm_cor}, the computational complexity to compute $A \cdot B$ is $O\bigl(\frac{n\cdot m}{t^2} \cdot T(t) + n \cdot m\bigr)$. Due to Theorem \ref{Williams_th}, $T(t) = O\bigl(t^3/2^{\Omega(\sqrt{\log t})}\bigr)$. Since $t \geq n^{1/4}$, $m \leq n^{3/4}$ and $m \cdot t = n$, we have
    $$
    \frac{n \cdot m}{t^2} \cdot T(t) + O(n \cdot m) = \frac{n \cdot m}{t^2} \cdot \frac{t^3}{2^{\Omega(\sqrt{\log t})}} + O(n \cdot m) = \frac{n^2}{2^{\Omega(\sqrt{\log n})}}.
    $$
    Remembering that the construction of $A$ and $B$ costs $O(n^{7/4})$ operations, it  finishes the proof.
\end{myproof}

\subsection{The Boolean Semiring
% Convolution in \texorpdfstring{\(\BB[\GC]\)}{Z2[G]}
}\label{generalized_DFT_subs}
 In this Subsection, we consider a group semiring $\BB[\GC]$, where $\BB = \bigl(\{0,1\},\vee,\wedge\bigr)$ is the Boolean semiring, and $\GC$ is an arbitrary finite Abelian group. The standard way of making convolution in $\BB[\GC]$ is embedding into $\ZZ[\GC]$, which consequently can be embedded into $\CC[\GC]$. Note that $\CC[\GC]$ becomes an algebra on a $\CC$-vector space. A natural basis of $\CC[\GC]$ is given by the indicator functions of the group elements. Identifying each group element with its indicator function, $\CC[\GC]$ can be viewed as the space of all formal sums $\sum_{g \in \GC} c_g \cdot g$ with coefficients in $\CC$. The multiplication (convolution) in $\CC[\GC]$ can be effectively reduced to the \emph{(generalized) Discrete Fourier Transform}, which maps elements of $\CC[\GC]$ into $\CC[\widehat{\GC}]$, where $\widehat{\GC}$ is defined as follows. 
\begin{definition}
  The set 
  $$    
        \widehat\GC = \Hom\bigl(\GC, \SC^1\bigr),\quad \text{for } \SC^1 = \bigl\{ x \in \CC \colon \abs{z} = 1 \bigr\}
  $$
  of group homomorhisms of $\GC$ into $\SC^1$, considered with the group operation
  $$
        (\chi_1 \cdot \chi_2)(g) = \chi_1(g) \cdot \chi_2(g),\quad\text{for }\chi_1,\chi_2 \in \widehat\GC \text{ and } g \in \GC,
  $$ is called the \emph{group of characters of $\GC$}, which is also known as the \emph{Pontryagin dual of $\GC$}.
\end{definition}
It is easy to see that $\chi^{-1} = \bar \chi$. It is known that, for an arbitrary Abelian group $\GC$, the group $\widehat{\GC}$ is isomorphic to $\GC$.

\begin{definition}
    For $\alpha \in \CC[\GC]$, the \emph{discrete Fourier transform (the DFT, for short)} of $\alpha$ is the function $\hat \alpha \in \CC[\widehat \GC]$, given by
    $$
        \hat \alpha(\chi) = \sum\limits_{g \in \GC} \alpha(g) \cdot \bar \chi(g). 
    $$
    The \emph{inverse DFT} is given by 
    \begin{equation}\label{inverse_DFT_def}
      \alpha(g) = \frac{1}{n} \cdot  \sum\limits_{\chi \in \widehat \GC} \hat \alpha(\chi) \cdot  \chi(g).
    \end{equation}
\end{definition}
Since $\widehat \GC$ is isomorphic to $\GC$, it consists of exactly $n$ elements, the so-called \emph{characters of $\GC$}, which can be viewed as the basis of $\CC[\widehat \GC]$ in the same way that we did it with $\CC[\GC]$. Consequently, if some order of elements in $\GC$ and $\widehat{\GC}$ is chosen, the DFT can be viewed as the matrix-vector multiplication on an $n \times n$ complex matrix $W$, whose columns and rows are indexed by elements of $\GC$ and $\widehat \GC$, respectively. 
 
Similarly, due to \eqref{inverse_DFT_def}, the inverse DFT can be handled as the DFT over $\CC[\widehat \GC]$ and reduced to a matrix multiplication in the same way. The following Theorem, due to U.~Baum, M.~Clausen \& B.~Tietz \cite{AbelianFFTImproved_inC}, gives an efficient way to compute such a matrix-vector multiplication.
 \begin{theorem}[U.~Baum, M.~Clausen \& B.~Tietz \cite{AbelianFFTImproved_inC}]\label{AbelianDFT_th} Assume that a basis of $\GC$ is given. Let $\alpha \in \CC[\GC]$, then $\hat \alpha \in \CC[\widehat \GC]$ can be computed with $8 n \cdot \log_2 n$ group operations and operations in $\CC$.
 \end{theorem}
 It can be directly checked that 
 $$
 \widehat{\alpha \star \beta}(\chi) = \hat \alpha(\chi) \cdot \hat \beta(\chi). 
 $$ Consequently, due to Theorem \ref{AbelianDFT_th} and our discussion, for an input $\alpha, \beta \in \CC[\GC]$, the convolution $\alpha \star \beta \in \CC[\GC]$ can be computed with $O(n \cdot \log n )$ operations in $\CC$.

 Recall that we were originally interested to make the convolution in $\BB[\GC]$, which can be easily reduced to the convolution in $\ZZ[\GC]$ and, consequently, in $\CC[\GC]$. Since it is impossible to implement exact complex arithmetic in the Word-RAM model, it is natural to represent complex numbers approximately as pairs of rational numbers. It can be shown that, if the input values are given with a sufficiently good additive accuracy $\varepsilon$, then the output values of the DFT algorithm of Theorem \ref{AbelianDFT_th} and, consequently, of the convolution calculation will be returned with an additive accuracy $\varepsilon \cdot n^{O(1)}$. Therefore, taking the size of rational representation of complex numbers proportional to $O(\log n)$, we can compute the convolution with an additive accuracy strictly less than $1/2$, which will give the correct integer answer after rounding to the nearest integer. Therefore, the convolution in $\BB[\GC]$ can be computed using $O(n \log n)$ operations with rational numbers of size $O(\log n)$. 
 Unfortunately, we did not able to find a complete error analysis for the generalized DFT in the literature with respect to the Word-RAM model. By this reason, we give a complete analysis of the Abelian case in our work.
 Before we formulate a main theorem, let us make a technical definition:
 \begin{definition}
  We call a complex number in the form $a + b \cdot i$ with $a,b \in \QQ$ as a \emph{rational complex number}. The \emph{set of rational complex numbers} is denoted by $\CCQ$. The \emph{encoding size of a rational complex number $z = a + b \cdot i$} is defined as $\size(z) = \size(a) + \size(b)$.
\end{definition}
 
 \begin{theorem}[The DFT with respect to the Word-RAM model]\label{group_DFT_th}
 Let us assume that a finite Abelian group $\GC$ is represented by its basis. Let $\alpha' \in \CCQ[\GC]$ be the input function and $\varepsilon \in \QQ \cap (0,1)$ be the input accuracy with $\varepsilon < 1/n^{C}$, for a sufficiently big absolute constant $C$. Assume additionally that there exists $\alpha \in \CC[\GC]$, such that $\norm{\alpha-\alpha'}_\infty \leq \varepsilon$. Then, there exists a polynomial-time algorithm that returns $\hat \alpha' \in \CCQ[\widehat{\GC}]$ such that
 \begin{enumerate}
   \item $\norm{\hat \alpha' - \hat \alpha}_\infty = n^{O(1)} \cdot \varepsilon$;
   \item The algorithm needs $O\bigl(n \cdot (\log n + \gamma + \log 1/\varepsilon )\bigr)$ operations with rational numbers of the size $O\bigl(\log n + \gamma + \size(\varepsilon \bigr))$, where $\gamma = \max_{g \in \GC} \size(\alpha'(g))$ is the maximum of component sizes of $\alpha'$.
 \end{enumerate} 
 \end{theorem}
 The proof could be found in Appendix Subsection \ref{group_DFT_proof}.
 
 \begin{corollary}[Convolution with respect to the Word-RAM model]\label{Z2Conv_cor}
     Let us assume that a finite Abelian group $\GC$ is represented by its basis. Given $\alpha,\beta \in \BB[\GC]$, the convolution $\alpha \star \beta \in \BB[\GC]$ can be computed with $O(n \log n)$ operations with rational numbers of the size $O(\log n)$.
 \end{corollary}
 \begin{myproof}
   Consider $\alpha$ and $\beta$ as the members of $\CC[\GC]$ with components in $\{0,1\}$. Choose a rational value $\varepsilon \leq 1/n^{C}$, for a sufficiently large absolute constant $C$. Using Theorem \ref{group_DFT_th}, we compute approximate versions $\hat \alpha'$ and $\hat \beta'$ of $\hat \alpha$ and $\hat \beta$ with the input accuracy $\varepsilon$. Due to Theorem \ref{group_DFT_th}, it can be done, using $O(n \log n)$ operations with rational numbers of the size $O(\log n)$. Additionally, we have $\delta_{\hat \alpha} := \norm{\hat\alpha-\hat\alpha'}_\infty = n^{O(1)} \cdot \varepsilon$ and $\delta_{\hat\beta} := \norm{\hat\beta-\hat\beta'}_\infty = n^{O(1)} \cdot \varepsilon$. Next, we compute $\hat \psi' := \hat \alpha' \cdot \hat \beta'$, which can be done with $O(n)$ operations with rational numbers of the size $O(\log n)$. Denoting $\hat \psi = \hat \alpha \cdot \hat \beta$, let us estimate the error $\delta_{\hat\psi} := \norm{\hat \psi - \hat \psi'}_\infty$:
   \begin{equation*}
     \delta_{\hat\psi} \leq \delta_{\hat\alpha} \cdot \norm{\hat \beta}_\infty + \delta_{\hat\beta} \cdot \norm{\hat \alpha}_\infty + \delta_{\hat\alpha} \cdot \delta_{\hat\beta} = n^{O(1)} \cdot \varepsilon,
   \end{equation*}
   where the last equality holds, because $\norm{\hat \beta}_\infty \leq n$ and $\norm{\hat \alpha}_\infty \leq n$.
   
   Recall that $\alpha \star \beta$ can be calculated, using the inverse DFT to $\hat \alpha \cdot \hat \beta$, which, due to \eqref{inverse_DFT_def}, is equivalent to the DFT in $\CC[\widehat\GC]$. Consequently, we compute the approximate version $\eta' \in \CCQ[\GC]$ of $\alpha \star \beta$, using Theorem \ref{group_DFT_th} to $\hat \psi'$ with an accuracy, corresponding to the error bound for $\delta_{\hat\psi}$, and dividing components of the resulting vector by $n$. Note that, due to Theorem \ref{group_DFT_th}, the component sizes of $\hat \psi'$ are bounded by $O(\log n)$. Hence, the last step costs of $O(n \log n )$ operations with rational numbers of the size $O(\log n)$. Additionally, 
   $$
        \norm{\alpha \star \beta - \eta'} \leq \delta_{\hat \psi} \cdot n^{O(1)} \cdot \varepsilon = n^{C_0} \cdot \varepsilon, \quad \text{for some constant $C_0$}.
   $$ 
   
   Choose $C$, such that $n^{C_0} \cdot \varepsilon \leq 1/3$ (for $n \geq 3$, we can put $C := C_0 + 1$). Since $\alpha \star \beta$ has integer elements, rounding to the nearest integer in components of $\eta'$ will give the correct answer, which finishes the proof.
 \end{myproof}

\section{Tiling Group}\label{tiling_group_sec}

In the current section we describe an algorithmic Abelian group construction, which will be used later in the proof of \Cref{genILP_main_th} as a basement for generalized convolution subproblems. The group is induced by a factorization of $\ZZ^n$ using tiles with integer sides.

Let $v \in \QQ^n$, $A \in \ZZ^{n \times n}$ with $\Delta = \abs{\det(A)} > 0$. Consider a set $\GC = v + A \cdot [-1,1)^n$. In other words, $\GC$ is an affine image of $[-1,1)^n$. Since $\RR^n$ can be tiled by $[-1,1)^n$ and its parallel copies, it follows that any point $y \in \RR^n$ has the following unique representations:
\begin{gather}
    \exists! z \in (2 \cdot \ZZ)^n,\; \exists! x \in [-1,1)^n \colon y = z + x,\label{unique_repr_BC_plus}\\ 
    \exists! z \in (2\cdot\ZZ)^n,\; \exists! x \in [-1,1)^n \colon y = z - x.\label{unique_repr_BC_minus}
\end{gather}
Using these representations, we define two functions $\lfloor \cdot \rceil$ and $\lceil \cdot \rfloor$ that map $\RR^n$ to $[-1,1)^n$. More precisely, for $y \in \RR^n$, we put $\lfloor y \rceil = x$, according to the representation \eqref{unique_repr_BC_plus}, and $\lceil y \rfloor = x$, according to the representation \eqref{unique_repr_BC_minus}, respectively. 
Note that, for any $z \in (2\cdot\ZZ)^n$ and $x \in \RR^n$, the following properties hold
\begin{gather}
    \lfloor z + x \rceil = \lfloor x \rceil,\label{round_fun_prop_1}\\
    \lceil z - x \rfloor = \lceil x \rfloor.\label{round_fun_prop_2}
\end{gather}
Additionally, it is useful to note that, for any $x,y \in \RR^n$: 
\begin{gather}
    \bigl\lfloor x + y \bigr\rceil = \bigl\lfloor \lfloor x \rceil + y \bigr\rceil = \bigl\lfloor x + \lfloor y \rceil \bigr\rceil. \label{round_fun_prop_3}
\end{gather}

Since the set $\GC$ with its parallel copies form a tiling of $\RR^n$, it is natural to consider an Abelian group that identifies points in parallel copies of $\GC$. Since $A$ is an integer matrix, the group is correctly defined and isomorphic to the group $\RR^n/(2A)\cdot \ZZ^n$. According to its background, we call this group as the \emph{tiling group}, induced by $\GC$. In the former text, we need an explicit construction of the tiling group, which can be given as follows.

Denote $t_v = A^{-1} \cdot v$, and note that any point $y \in \RR^n$ can be uniquely represented as $y = A\cdot(t_v  + t_y)$, for unique $t_y \in \RR^n$. Clearly, if $y \in \GC$, then $t_y \in [-1,1)^n$.
\begin{definition}
    For $y,z \in \GC$ with representations $y = A\cdot(t_v + t_y)$ and $z = A\cdot(t_v + t_z)$, consider a binary operation $\oplus$, defined by the formula 
    \begin{equation}\label{tiling_group_op}
    y \oplus z = A \cdot \bigl( t_v + \lfloor t_v + t_y + t_z \rceil \bigr).
    \end{equation}
    The pair $\bigl(\GC,\oplus\bigr)$ forms an Abelian group, called the \emph{tiling group, induced by $\GC$}.
\end{definition}
Let us check that $(\GC,\oplus)$ is indeed an Abelian group. Since the commutativity and associativity are straightforward, we need only to show the existence of the neutral element $0_{\GC}$ and an inverse element $\ominus z \in \GC$, for an arbitrary $z \in \GC$. 
% Note that, due to the commutativity of $\oplus$ and group axioms, we not need to check the uniqueness of the neutral and inverse elements. So, our notations are correctly defined. 
The neutral element $0_{\GC}$ is given by the formula:
\begin{equation}\label{tiling_group_neutral}
    0_{\GC} = A \cdot \bigl(t_v + \lceil t_v \rfloor\bigr).
\end{equation}
    Definitely, let $z \in \GC$ be represented as $z = A\cdot(t_v + t_z)$, then 
    \begin{multline*}
        z \oplus 0_{\GC} = A\bigl(t_v + \bigl\lfloor t_v + t_z + \lceil t_v \rfloor \bigr\rceil\bigr) = \\
        = A\bigl(t_v + \lfloor t_z \rceil\bigr) = A(t_v + t_z) = z,
    \end{multline*}
    where the second equality holds, because $t_v + \lceil t_v \rfloor \in 2\cdot\ZZ$ and, due to the property \eqref{round_fun_prop_1}. The third equality follows, because $t_z \in [-1,1)^n$.
    
    Now, let us give an explicit formula for the inversion. For an element $z \in \GC$, represented as $z = A(t_v + t_z)$, the inverse element $\ominus z$ is given by the formula:
    \begin{equation}\label{tiling_group_inv}
        \ominus z = A\bigl(t_v + \lceil t_v + t_z - t_0 \rfloor\bigr),\quad\text{where $t_0 = \lceil t_v \rfloor$.}
    \end{equation}
    Let us check the correctness:
    \begin{multline*}
        z \oplus (\ominus z) = A\bigl(t_v + \bigl\lfloor t_v + t_z + \lceil t_v + t_z - t_0 \rfloor \bigr\rceil\bigl) = \\
        = A\bigl(t_v + \bigl\lfloor t_0 + (t_v + t_z - t_0) + \lceil t_v + t_z - t_0 \rfloor \bigr\rceil\bigl) =\\
        = A\bigl(t_v + \lfloor t_0 \rceil \bigr) = A(t_v + t_0) = 0_{\GC},
    \end{multline*}
    where the third equality holds, because $t_v + t_z - t_0 + \lceil t_v + t_z - t_0 \rfloor \in 2\cdot\ZZ$ and, due to the property \eqref{round_fun_prop_1}. The forth equality follows, because $t_0 \in [-1,1)^n$.

    It is natural to consider the canonical homomorphism $\phi_{\GC}$ that maps $\RR^n$ into $\GC$. For $y \in \RR^n$, represented as $y = A \cdot (t_v + t_y)$, it is given by the formula:
    \begin{equation}\label{tiling_group_cgomo}
        \phi_{\GC}(y) = A \cdot \bigl( t_v + \lfloor t_y \rceil \bigr).
    \end{equation}
    Let us check that $\phi_{\GC}(y)$ is indeed a homomorphism. Definitely, for $y,z \in \RR^n$, represented as $y = A \cdot (t_v + t_y)$ and $z = A \cdot (t_v + t_z)$, we have
\begin{multline*}
    \phi_{\GC}(y+z) = A \cdot \bigl(t_v + \lfloor t_v + t_y + t_z \rceil \bigr) = \\
    = A \cdot \Bigl(t_v + \bigl\lfloor t_v + \lfloor t_y \rceil + \lfloor t_z \rceil \bigr\rceil \Bigr) = \phi_{\GC}(y) \oplus \phi_{\GC}(z),
\end{multline*}
where the second equality holds, due to the property \eqref{round_fun_prop_3}.
The group operation can be rewritten, using $\phi_{\GC}$, in the following way:
\begin{equation*}
    y \oplus z = \phi_{\GC}(y + z).
\end{equation*}

The following simple Lemma is important for algorithmic implications of our work.
\begin{lemma}\label{tiling_injection_lm}
Let $\WC \subseteq \RR^n$ be an arbitrary set. If there exists a translation vector $t \in \RR^n$, such that $t + \WC \subseteq \GC$, then the restriction of $\phi_{\GC}$ on $\WC$ is injective. 
\end{lemma}
\begin{myproof}
    Choose arbitrary and different $y,z \in \WC$. Denote $y' = y + t$ and $z' = z + t$. Note that $\phi_{\GC}(y') \ominus \phi_{\GC}(z') = \phi_{\GC}(y) \ominus \phi_{\GC}(z)$. Since $y',z' \in t + \WC \subseteq \GC$, we have $\phi_{\GC}(y') = y'$ and $\phi_{\GC}(z') = z'$. Finally, 
    \begin{equation*}
        \phi_{\GC}(y)\ominus\phi_{\GC}(z) = \phi_{\GC}(y') \ominus \phi_{\GC}(z') = y' \ominus z' \not= 0,
    \end{equation*}
    which proves the injectivity.
\end{myproof}

Since $A$ is an integer matrix, the operation $\oplus$ and the homomorphism $\phi_{\GC}$ map integers to integers. Additionally, it can be directly checked that $0_{\GC} \in \ZZ^n$ and $\ominus z \in \ZZ^n$, for any $z \in \ZZ^n \cap \GC$. In other words, the set $\GC_{I} = \ZZ^n \cap \GC$, equipped by the operation $\oplus$, also admits a group structure, which is isomorphic to  $\ZZ^n / (2A) \cdot \ZZ^n$. 
\begin{definition}
    The pair $\bigl( \GC_{I}, \oplus \bigr)$ forms a group, called the \emph{integer tiling group, induced by $\GC$.} Note that $0_{\GC_{I}} = 0_{\GC}$.
\end{definition} 

This group is especially interesting to us in terms of algorithmic implications, since it is finite and, therefore, admits a finite basis.
\begin{theorem}\label{tailing_group_basis_th}
    Consider an integer tiling group $\GC_{I}$ induced by $\GC$, and let $S = P \cdot A \cdot Q$ be the SNF of $A$, where $P,Q \in \ZZ^{n \times n}$ are unimodular. 
    
    Then the vectors $b_1, b_2, \dots, b_n \in \GC_I$, given by the formula
    $$
    b_k = A \cdot \Bigl( t_v + \bigl\lceil t_v - \lfloor Q_{* k} / S_{k k} \rceil \bigr\rfloor \Bigr), \quad \text{for $k \in \intint n$,}
    $$
    form a basis of $\GC_I$.
\end{theorem}
\begin{myproof}
    Consider an integer tiling group $\GC'$, induced by the set $A \cdot [-1,1)^n$. The following Lemma states that $\GC'$ is isomorphic to $\GC_I$, the proof could be found in Appendix, Section \ref{group_isomorphism_sec}.
    \begin{lemma}\label{group_isomorphism_lm}
        The group $\GC_I$ is isomorphic $\GC'$, the corresponding bijective map $\phi \colon \GC_I \to \GC'$ and its inverse are given by the formulae:
        \begin{gather}
            z = A \cdot (t_v + t_z) \quad\to\quad A \cdot \lfloor t_v + t_z \rceil  =: \phi(z), \label{phi_homomorph_eq}\\
            z' = A \cdot t_{z'} \quad\to\quad A \cdot \bigl( t_v + \lceil t_v - t_{z'} \rfloor \bigr) =: \phi^{-1}(z').\label{phi_inv_homomorph_eq}
        \end{gather}
    \end{lemma}
    
    As the first step, we will construct a basis for $\GC'$. Then the basis of $\GC_I$ can be obtained, using the map $\phi^{-1}$ on the basis of $\GC'$. We claim that the basis $b_1', b_2', \dots, b'_n$ of $\GC'$ is given by the formula
    $$
    b_k' = A\cdot \bigl\lfloor Q_{* j}/S_{j j} \bigr\rceil, \quad \text{for $k \in \intint n$.}
    $$
    Let as show that the combinations 
    $$
    g' = i_1 \cdot b_1' \oplus \dots \oplus i_n \cdot b_n', \quad \text{where $i_j \in \intint[-S_{j j}]{S_{j j}-1}$, for $j \in \intint n$},
    $$ uniquely represent all the elements of $\GC'$. The inclusion $g' \in \GC'$ can be checked, using the following formulae:
    \begin{multline*}
        g' = A \cdot \Bigl\lfloor i_1 \cdot \bigl\lfloor Q_{* 1}/S_{1 1} \bigr\rceil + \dots + i_n \cdot \bigl\lfloor Q_{* n} / S_{n n} \bigr\rceil \Bigr\rceil = \\
        = A \cdot \Bigl\lfloor \frac{i_1}{S_{1 1}} \cdot Q_{* 1} + \dots + \frac{i_n}{S_{n n}} \cdot Q_{* n} \Bigr\rceil =\\
        = A \cdot \Bigl( \frac{i_1}{S_{1 1}} \cdot Q_{* 1} + \dots + \frac{i_n}{S_{n n}} \cdot Q_{* n} \Bigr) + A \cdot z = \qquad \text{`for some $z \in (2\cdot\ZZ)^n$'}\\
        = P^{-1} \cdot (i_1, \dots, i_n)^{\top} + A \cdot z,
    \end{multline*}
    where the second equality uses the property \eqref{round_fun_prop_3}. Therefore, $g'$ is integer and $g' \in \GC'$. Next, let us show that the different combinations give different elements of $\GC'$. For the sake of contradiction, assume that there exists a non-trivial combination, representing $0_{\GC'}$, that is
    $$
    i_1 \cdot b_1' \oplus \dots \oplus i_n \cdot b_n' = 0_{\GC'},
    $$ which can be rewritten as
    \begin{equation*}
        \BZero = A \cdot \Bigl\lfloor i_1 \cdot \bigl\lfloor Q_{* 1}/S_{1 1} \bigr\rceil + \dots + i_n \cdot \bigl\lfloor Q_{* n} / S_{n n} \bigr\rceil \Bigr\rceil.
    \end{equation*}
    Since $A$ is invertible and, due to the property \eqref{round_fun_prop_3}, it is equivalent to
    \begin{equation}\label{zero_i_comb_eq1}
    \BZero = \Bigl\lfloor \frac{i_1}{S_{1 1}} \cdot Q_{* 1} + \dots + \frac{i_n}{S_{n n}} \cdot Q_{* n} \Bigr\rceil.    
    \end{equation}
    The following sequence of equivalences holds:
    \begin{multline*}
        \eqref{zero_i_comb_eq1} \quad\Leftrightarrow\quad \frac{i_1}{S_{1 1}} \cdot Q_{* 1} + \dots + \frac{i_n}{S_{n n}} \cdot Q_{* n} \in (2\cdot\ZZ)^n \quad\Leftrightarrow \\
        \Leftrightarrow\quad Q \cdot S^{-1} \cdot (i_1, i_2, \dots, i_n)^\top \in (2\cdot\ZZ)^n \quad\Leftrightarrow\quad 2 S_{j j} \mid i_j, \quad \text{for j $\in \intint n$.}
    \end{multline*}
    Since $i_j \in \intint[-S_{j j}]{S_{j j} - 1}$, it is only possible if $i_j = 0$, for all $j \in \intint n$, which contradicts to our assumption.
    
    Finally, applying the isomorphism $\phi^{-1}$ to the basis of $\GC'$, we construct a basis $b_1, b_2, \dots, b_n$ of $\GC_I$:
    \begin{equation*}
        b_k = \phi^{-1}(b_k') = A \cdot \Bigl( t_v + \bigl\lceil t_v - \lfloor Q_{* k} / S_{k k} \rceil \bigr\rfloor \Bigr), \quad \text{for $k \in \intint n$.}
    \end{equation*}
\end{myproof}

As a simple corollary, we get
\begin{corollary}\label{tailing_group_card_cor}
    In terms of Theorem \ref{tailing_group_basis_th}, $\abs{\GC_I} = 2^n \cdot \Delta$. Taking the matrix $A \in \ZZ^{n \times n}$ and vector $v \in \QQ^n$ as an input, all the elements of $\GC_I$ can be enumerated with $O\bigl(\abs{\GC_I} \cdot n \bigr)$ operations with rational numbers, whose size is bounded by a polynomial on the input size. The number of group operations is bounded by $O\bigl(\abs{\GC_I}\bigr)$.
\end{corollary}
\begin{myproof}
    Due to A.~Storjohann \cite{SNFOptAlg}, the SNF $S$ of $A$, together with unimodular matrices $P$ and $Q$, can be computed by a polynomial-time algorithm. Consequently, due to Theorem \ref{tailing_group_basis_th}, the basis of $\GC_I$ can be constructed by a polynomial-time algorithm. Therefore, we can enumerate all the elements of $\GC_I$ by enumerating the combinations 
    $$
    i_1 \cdot b_1 \oplus \dots \oplus i_n \cdot b_n,\quad\text{where $i_j \in \intint[-S_{j j}]{S_{j j}-1}$, for $j \in \intint n$.}
    $$  Since any group element can be represented by a vector in $\QQ^n$ of a polynomial size and since the group operation $\oplus$ has a linear computational complexity, the total number of operations is $O\bigl(\abs{\GC_I} \cdot n\bigr)$.
\end{myproof}

\section{Preliminaries from Discrepancy Theory}\label{prediscrep_sec}

As noted in \Cref{proof_outline_sec}, we employ tools from discrepancy theory to prove our main result. Below we provide a brief list of required results and definitions.

\begin{definition}
For a matrix $A \in \RR^{m \times n}$, we define \emph{its discrepancy and its hereditary discrepancy} by the formulas
\begin{gather*}
\disc(A) = \min_{z \in \{-1,\, 1\}^n} \left\| A z  \right\|_\infty,\\
\herdisc(A) = \max_{\IC \subset \intint n} \disc(A_{* \IC}).
\end{gather*}
\end{definition}

% For our analysis, we require the following fundamental bounds on the hereditary discrepancy $\herdisc(A)$. 
% The seminal result of Spencer \cite{SixDeviations_Spencer} establishes that for any matrix $A \in \RR^{m \times n}$,
% \begin{equation}\label{SixDeviations_eq}
%     \disc(A) \leq 6 \cdot \|A\|_{\max} \cdot \sqrt{n}.
% \end{equation}
% Due to the works \cite{HerDisc} and \cite{SixDeviations_Spencer} of Lov\'asz,  Spencer, \& Vesztergombi and Spencer, it is known that
% \begin{equation}\label{SixDeviations_eq}
% \herdisc(A) \leq 2 \disc(A) \leq 12 \sqrt{k} \cdot \|A\|_{\max}.
% \end{equation}
% Due to Beck and Fiala \cite{DiscBeckBound}, the value of $\herdisc(A)$ is bounded by the $l_1$-norm of columns. More precisely,
% \begin{equation}\label{Beck_eq}
%     \herdisc(A) < \|A\|_{\infty}.
% \end{equation}
% Additionally, Beck and Fiala conjectured that $\herdisc(A) = O\bigl(\sqrt{ \|A\|_{\infty} }\bigr)$ and settling this has been an elusive open problem. The best known result in this direction is due to Banaszczyk \cite{BanaszDiscBound}:
% \begin{equation}\label{Banasz_eq}
%     \herdisc(A) = O\Bigl( \sqrt{\|A\|_{\infty} \cdot \log(n)} \Bigr).
% \end{equation}
The important matrix characteristic that is closely related to $\herdisc(A)$ is $\detlb(A)$. Due to Lov\'asz, Spencer, \& Vesztergombi \cite{HerDisc}, it can be defined as follows:
$$
\detlb(A) = \max\limits_{t \in \intint k} \sqrt[t]{\Delta_t(A)},
$$ and it was shown in \cite{HerDisc} that 
$
\herdisc(A) \geq (1/2) \cdot \detlb(A)
$. Matou\v{s}ek in \cite{DiscDetBound} showed that $\detlb(A)$ can be used to produce tight upper bounds on $\herdisc(A)$. The result of Matou\v{s}ek was improved by Jiang \& Reis in \cite{TightDiscDetBound}:
\begin{equation}\label{DiscDetBound_eq}
\disc(A) = O\Bigl( \detlb(A) \cdot \sqrt{\log k \cdot \log n} \Bigr).    
\end{equation} 

We must also recall an important property concerning the discrepancy of matrices $A \in \RR^{k \times n}$ when $k \leq n$.
\begin{lemma}[{ Alon\&Spencer \cite[Corollary 13.3.3]{AlonSpencerBook} }]\label{DiscLowk_lm}
    Suppose that $\disc(A_{* \IC}) \leq H$ for every subset $\IC \in \intint n$ with $\abs{\IC} \leq k$. Then, $\disc(A) \leq 2 H$.
\end{lemma}
Originally, this statement was proved only for the discrepancy of hypergraphs. However, it is straightforward to see from the original proof that it extends to matrices as well.
Combining Lemma \ref{DiscLowk_lm} with
the upper bound \eqref{DiscDetBound_eq}, 
% the upper bounds \eqref{SixDeviations_eq} and \eqref{DiscDetBound_eq}, 
we get
\begin{equation}\label{DiscDetBoundReduced_eq}
    \herdisc(A) = O\left(\log k \cdot \detlb(A)\right).
\end{equation}
% \begin{gather}
%     \herdisc(A) \leq 12 \cdot \abs{A}_{\max} \cdot \sqrt{k},\label{SixDeviationsReduced_eq}\\
%     \herdisc(A) = O\left(\log k \cdot \detlb(A)\right).\label{DiscDetBoundReduced_eq}
% \end{gather}

The following key Lemma, due to K.~Jansen \& L.~Rohwedder\cite{DiscConvILP}, connects results of the discrepancy theory with the theory of integer linear programs with a bounded co-dimension.
\begin{lemma}[K.~Jansen \& L.~Rohwedder \cite{DiscConvILP}]\label{disc_lm} Let $A \in \RR^{k \times n}$ with $\rank(A) = k$ and $x \in \ZZ^n_{\geq 0}$. Then there exists a vector $z \in \ZZ^n_{\geq 0}$ with
\begin{enumerate}
    \item $z \leq x$;
    \item $\frac{1}{6} \cdot \norm{x}_1 \leq \norm{z}_1 \leq \frac{5}{6} \cdot \norm{x}_1$;
    \item $
\norm{A(z - x/2)}_\infty \leq 2 \cdot \herdisc(A)
$.
\end{enumerate}
\end{lemma}

\section{Proof of \Cref{genILP_main_th}}\label{genILP_main_th_proof}

Recall the definition of \Cref{genILP_main_th}.

\mainGenILPTh*

The proof consists of three parts: In the first part, we describe our dynamic programming algorithm; In the second part, we estimate the parameters of the dynamic program; Finally, in the third part, we put things together and provide the final computational complexity bound.  

\subsection{Dynamic Program}

In seminal work \cite{DiscConvILP}, K.~Jansen \& L.~Rohwedder provide a new class of dynamic programming algorithms for ILP problems, which uses results of the discrepancy theory and fast algorithms for tropical convolution on sequences. Our dynamic programming algorithm follows to the same pattern, but it also has sufficient differences, and we need to solve more general tropical and Boolean convolution problems on a special group ring. The algorithm is presented in the following theorem.

\begin{theorem}\label{DP_th}
Consider the \ref{GEN-ILP-SF} problem. Let $r = \abs{\GC}$ and $\rho \in \ZZ_{>0}$ be the value, such that $\norm{z^*}_1 \leq (6/5)^\rho$, for some optimal integer solution $z^*$ of the problem. 
Additionally, for a given nondegenerate $B \in \ZZ^{n \times n}$, let $M = B^{-1} \cdot A$, $\eta \in \ZZ_{\geq 1}$ be an upper bound on $\herdisc(M)$, and $\delta = \abs{\det B}$.
% Additionally, for a base $\BC \subseteq \intint n$ of $A$, let $M = (A_{\BC})^{-1} \cdot A$, $\eta \in \ZZ_{\geq 1}$ be an upper bound on $\herdisc(M)$, and $\delta = \abs{\det(A_{\BC})}$. 
Then the problem can be solved with
$$
\rho \cdot \tau^2 / 2^{\Omega\bigl(\sqrt{\log \tau}\bigr)} + O(n)
$$ operations with elements of $\ZZ^k \times \GC$, where $\tau = (16 \eta)^k \cdot r \cdot \delta$.
The feasibility variant of the problem can be solved with
$$
O(\rho \cdot \tau \cdot \log \tau + n) \quad \text{operations in $\ZZ^k \times \GC$.}
$$ 
\end{theorem}
\begin{myproof}
    % Denote $v^* = (A_{\BC})^{-1} b$ and define
    Denote $v^* = B^{-1} b$ and define
    \begin{multline*}
        \MC(i,\gamma) = \ZZ^k \cap \left\{ B x \colon \norm{x - 2^{i - \rho} v^*}_\infty \leq \gamma \right\} = \\
        \ZZ^k \cap \left( 2^{i-\rho} b + B \cdot [-\gamma, \gamma]^k \right) = \ZZ^k \cap \left( 2^{i-\rho} b + (\gamma B) \cdot [-1, 1]^k \right).
    \end{multline*} 
    For every $i \in \intint[0]{\rho}$, $b^\prime \in \MC(i, 4 \eta)$, and every $g^\prime \in \GC$, we solve the problem
    \begin{gather}
    c^\top x \to \min\notag\\
    \begin{cases}
    A x = b^{\prime}\\
    x_1 \cdot g_1 + \dots + x_n \cdot g_n = g^{\prime}\\
    \norm{x}_1 \leq \left(\frac{6}{5}\right)^i\\
    x \in \ZZ^n_{\geq 0}.
    \end{cases}\tag{$\DP(i,b^\prime,g^\prime)$}\label{norm_prob}
    \end{gather} 
    Following to the paper \cite{DiscConvILP} of K.~Jansen \& L.~Rohwedder, we iteratively derive solutions for $i$, using pairs of solutions for $i-1$. Finally, an optimal solution of the original problem can be derived from $\DP(i,b',g')$, taking $i := \rho$, $b' := b$, and $g' := g_0$.
    
    We decompose the computation into $\rho+1$ levels, where the $i$-th level means the solution of all problems of the type $\DP(i,\cdot,\cdot)$. The computation in the level $\DP(0,\cdot,\cdot)$ is straightforward, because its solutions correspond exactly to the elements $\binom{A_{* 1}}{g_1}, \binom{A_{* 2}}{g_2}, \dots, \binom{A_{* n}}{g_n}$. So, the $0$-th level can be computed with $O(n)$ operations in $\ZZ^k \times \GC$. 
    
    Fix some $i \geq 1$. Let us estimate the computational complexity of the level $\DP(i,\cdot,\cdot)$. For $b^\prime \in \MC(i, 4 \eta)$ and $g^\prime \in \GC$, let $x^*$ be an optimal solution of \ref{norm_prob}. The equality $A x^* = b^\prime$ can be rewritten as $M x^* = B^{-1} b^\prime$. By Lemma \ref{disc_lm}, there exists a vector $0 \leq z \leq x^*$ with $\norm{M z - \frac{1}{2} B^{-1} b^\prime}_{\infty} \leq 2 \eta$ and 
    $$
    \norm{z}_1 \leq \frac{5}{6} \cdot \norm{x^*}_1 \leq \frac{5}{6} \cdot \left(\frac{6}{5}\right)^i = \left(\frac{6}{5}\right)^{i -1},\quad\text{if $\norm{x^*}_1 > 1$,}
    $$
    or $\norm{z}_1 \leq \norm{x^*}_1 \leq 1 \leq (6/5)^{i-1}$, otherwise. The same holds for $x^* - z$. Therefore, $z$ is an optimal solution for $\DP(i-1,b^{\prime\prime}, g^{\prime\prime})$, where $b^{\prime\prime} = A z$ and $g^{\prime\prime} = z_1 \cdot g1 + z_2 \cdot g_2 + \ldots + z_n \cdot g_n$. Likewise, $x^* - z$ is an optimal solution for $\DP(i-1,b^\prime - b^{\prime\prime}, g^\prime - g^{\prime\prime})$. We claim that $b^{\prime\prime} \in \MC(i-1,4 \eta)$ and $b^{\prime} - b^{\prime\prime} \in \MC(i-1, 4 \eta)$. 
    Definitely, \begin{multline*}
        \norm{B^{-1} b^{\prime\prime} - 2^{(i-1)-\rho} v^*}_{\infty} = \norm{M z - \frac{1}{2} B^{-1} b^{\prime} + \frac{1}{2} B^{-1} b^{\prime} - 2^{(i-1)-\rho} v^*}_{\infty} \leq \\
        \leq \norm{M z - \frac{1}{2} B^{-1} b^{\prime}}_{\infty} + \norm{\frac{1}{2} B^{-1} b^{\prime} - 2^{(i-1)-\rho} v^*}_{\infty} \leq \\
        \leq 2 \cdot \eta + \frac{1}{2} \norm{B^{-1} b^{\prime} - 2^{i - \rho} v^*}_{\infty} \leq 4 \cdot \eta,
    \end{multline*}
    which proves that $b'' \in \MC(i-1, 4 \eta)$. The proof of the inclusion $b'-b'' \in \MC(i-1, 4 \eta)$ is completely similar. 
    
    This Claim implies that we can combine pairs of optimal solutions of $\DP(i-1, b^{\prime\prime}, g^{\prime\prime})$ and $\DP(i-1,b^\prime - b^{\prime\prime}, g^\prime - g^{\prime\prime})$, to compute an optimal solution for $\DP(i,b^\prime, g^\prime)$. More precisely, the following formula holds:
    \begin{equation}\label{DP_formula}
    \DP\bigl(i,b^\prime,g^\prime\bigr) = \min_{
    \substack{
    b^{\prime\prime} \in \MC(i-1,4 \eta)\\
    g^{\prime\prime} \in \GC
    }
    } \Bigl\{\DP\bigl(i-1,b^{\prime\prime},g^{\prime\prime}\bigr) + \DP\bigl(i-1,b^\prime-b^{\prime\prime},g^\prime-g^{\prime\prime}\bigr)\Bigr\}.
    \end{equation}
    Based on the formula \eqref{DP_formula}, let us show that the computation of the level $\DP(i,\cdot,\cdot)$ can be reduced to the convolution in a group ring $\RC_{(\min,+)}[\QC]$, where $\RC_{(\min,+)} = \bigl(\ZZ \cup \{+\infty\}, \min, +\bigr)$ is the $(\min,+)$-semiring and $\QC$ is a specially constructed group. 
    
    Consider the set $\HC$, defined by the formula
    \begin{multline*}
        \HC = \left\{ 2 B x \colon \bigl(x - 2^{i-1-\rho} v^*\bigr) \in [-4 \eta, 4 \eta)^k \right\} = \\
        = 2^{i - \rho} b + 2B \cdot [-4 \eta, 4 \eta)^k = 2^{i - \rho} b + (8 \eta B) \cdot [-1,1)^k.
    \end{multline*}
    Due to the facts, described in Section \ref{tiling_group_sec}, we can look at this set as a group, called the \emph{integer tiling group}.
    
    We claim that $\MC(i,4 \eta) \subseteq \HC$. Definitely, let $y \in \MC(i,4\eta)$, then $y = B x$, for $\norm{x - 2^{i - \rho} v^*}_{\infty} \leq 4 \eta$. The last is equivalent to $\norm{x/2 - 2^{i -1- \rho} v^*}_{\infty} \leq 2 \eta$. Since $y = 2 B (x/2)$, we have $y \in \HC$, which proves the Claim. Additionally, note that there exists a vector $t \in \QQ^n$, such that $t + \MC(i-1,4 \eta) \subseteq \HC$ (one can put $t = 2^{i - 1 - \rho} b$). Consequently, due to Lemma \ref{tiling_injection_lm}, the canonical homomorphism $\varphi := \phi_{\HC} \colon \ZZ^k \to \HC$ injectively maps the set $\MC(i-1,4 \eta)$ into $\HC$. Therefore, the map $\varphi^{-1}$ is correctly defined on the set $\varphi\bigl(\MC(i-1, 4\eta)\bigr)$. 
    % Let us denote the corresponding bijection between the sets $\MC(i-1, 4 \eta)$ and $\phi_{\HC}\bigl(\MC(i-1, 4 \eta)\bigr)$ by $\varphi$.
    
    Now, we construct the group $\QC$ as the direct sum $\QC = \HC \oplus \GC$ of the groups $\HC$ and $\GC$. To compute the level $\DP(i,\cdot,\cdot)$, we define $\alpha \in \RC_{(\min,+)}[\QC]$ by the formula
    $$
    \alpha(h + g) = \begin{cases}
        \DP\bigl(i-1,\varphi^{-1}(h),g\bigr),\quad\text{if $h \in \varphi\bigl(\MC(i-1,4 \eta)\bigr)$}\\
        +\infty,\quad\text{in the opposite case.}
    \end{cases}
    $$
    We claim that
    $$
    \DP(i,h,g) = \alpha \star \alpha(h + g),\quad\text{for $h \in \MC(i,4 \eta)$ and $g \in \GC$.}
    $$
    Definitely, for $h \in \MC(i,4 \eta)$ and $g \in \GC$, we have
    \begin{multline*}
        \alpha \star \alpha(h + g) = \min\limits_{q \in \QC} \bigl\{\alpha(q) + \alpha(h+g-q)\bigr\} = \\ = \min\limits_{\substack{h' \in \HC\\ g' \in \GC}} \bigl\{ \alpha(h'+g') + \alpha(h-h'+g-g') \bigr\} = \\
        = \min\limits_{\substack{h' \in \varphi(\MC(i-1,4 \eta))\\ g' \in \GC}} \bigl\{ \alpha(h'+g') + \alpha(h-h'+g-g') \bigr\} = \\ = \min\limits_{\substack{h' \in \varphi(\MC(i-1,4 \eta))\\ g' \in \GC}} \bigl\{ \DP\bigl(i-1,\varphi^{-1}(h'),g'\bigr) + \DP\bigl(i-1, \varphi^{-1}(h-h'),g-g'\bigr) \bigr\} = \\
        = \min\limits_{\substack{h' \in \MC(i-1,4 \eta)\\ g' \in \GC}} \bigl\{ \DP\bigl(i-1,\varphi^{-1}(\varphi(h')),g'\bigr) + \DP\bigl(i-1, \varphi^{-1}(h-\varphi(h')),g-g'\bigr) \bigr\} =\\
        = \min\limits_{\substack{h' \in \MC(i-1,4 \eta)\\ g' \in \GC}} \bigl\{ \DP\bigl(i-1,h',g'\bigr) + \DP\bigl(i-1, h-h',g-g'\bigr) \bigr\} = \DP(i,h,g),
    \end{multline*}
    where the equality of the last lines holds, because $h-h' \in \MC(i-1,4\eta)$ and $\varphi^{-1}(h - \varphi(h')) = \varphi^{-1}(\varphi(h) - \varphi(h')) = \varphi^{-1}(\varphi(h - h')) = h - h'$. Here, the equality $\varphi(h) = h$ holds, because $h \in \MC(i,4 \eta) \subseteq \HC$.   
    
    Therefore, assuming that $i \in \intint \rho$ is fixed, we have shown that computation of the level $\DP(i,\cdot,\cdot)$ can be reduced to the convolution $\alpha \star \alpha$ for a specially constructed $\alpha \in \RC_{(\min,+)}[\QC]$. Assuming that we know some basis of $\QC$, due to Theorem \ref{group_min_conv_th}, the convolution $\alpha \star \alpha$ can be computed with $O\Bigl(\abs{\QC}^2 /2^{\Omega\bigl(\sqrt{\log \abs{\QC}}\bigr)}\Bigr)$ operations. Due to Theorem \ref{tailing_group_basis_th}, a basis of $\HC$ can be computed by a polynomial-time algorithm, so the same fact hods for the whole group $\QC$. Hence, to finish the computational complexity analysis for the level $\DP(i,\cdot,\cdot)$, we need to estimate the value of $\abs{\QC}$. 
    
    Due to Corollary \ref{tailing_group_card_cor}, $\abs{\HC} = (16 \eta)^k \cdot \delta $ and, consequently, $\abs{\QC} = \tau = (16 \eta)^k \cdot r \cdot \delta$, where $\tau$ is the constant from the Theorems' definition. Therefore, the level $\DP(i,\cdot,\cdot)$ can be computed with $O\bigl(\tau^2 / 2^{\Omega(\sqrt{\log \tau})}\bigr)$ operations with integers and group elements in $\QC$. Note that elements of $\QC$ can be explicitly represented as elements of $\ZZ^k \times \GC$. Finally, since there are exactly $\rho+1$ levels, the total computational complexity becomes $\rho \cdot \tau^2 / 2^{\Omega(\sqrt{\log \tau})} + O(n)$, where $O(n)$ is the computational complexity of the $0$-th level.
    
    The proof for the feasibility variant of the problem \ref{GEN-ILP-SF} is completely similar, with the only difference that the convolution in the $(\min,+)$-semiring is replaced by the Boolean convolution. In other words, the convolution is defined with respect to the group algebra $\BB[\QC]$, where $\BB = \bigl(\{0,1\}, \vee, \wedge \bigr)$ is the Boolean semiring. Since we know a basis of $\QC$, due to Corollary \ref{Z2Conv_cor}, the convolution in $\BB[\QC]$ can be computed with $O(\tau \log \tau)$ operations. Therefore, the total computational complexity for the feasibility problem is $O(\rho \cdot \tau \cdot \log \tau + n)$.
\end{myproof}

\subsection{Estimating the Parameters \(\rho\), \(\eta\), and \(\delta\) of the Dynamic Program}

Let us first estimate the parameter $\rho$ of Theorem \ref{DP_th}. Due to the next Lemma, we can assume that $\rho = O\bigl(\log(k \cdot r \cdot \Delta)\bigr)$.
\begin{lemma}\label{rho_estimate_lm}
    Any instance of the problem \ref{GEN-ILP-SF} can be transformed to an equivalent instance with the following property: if the problem is feasible and bounded, then there exists an optimal solution $z^*$, such that $\norm{z^*}_1 = (k \cdot \Delta \cdot r)^{O(1)}$, where $r = \abs{\GC}$.
\end{lemma}
\begin{myproof}
    Let $v$ be an optimal vertex solution of the relaxation, which can be found by a polynomial-time algorithm. Due to \cite[Corollary~2]{OnCanonicalProblems_Grib}, there exists an optimal solution $\hat z$ of the original problem, such that $\norm{v-\hat z}_1 \leq \chi$, for $\chi = O\bigl(k^2 \cdot (r \Delta) \cdot \sqrt[k]{r \Delta}\bigr)$. Following to \cite{DiscConvILP}, we change the variables $x^\prime = x - y$, where $y_i = \min\{0, \lceil v_i \rceil - \chi\}$. Note that $\hat z \geq y$, and since $v$ has at most $k$ non-zero components, we have $\|\hat z - y\|_1 = O(k \cdot \chi)$. Since $\hat z \geq y$, after the change of variables, the original problem transforms to an equivalent instance of the problem \ref{GEN-ILP-SF} with an optimal solution $z^* = \hat z - y$, satisfying
    $$
    \norm{z^*}_1 = O(k \cdot \chi) = (k \cdot \Delta \cdot r)^{O(1)}.
    $$
\end{myproof}

To give a good bound for the parameter $\eta$, we use the inequality \eqref{DiscDetBoundReduced_eq}.
% To give a good bound for the parameter $\eta$, we use the inequalities \eqref{SixDeviationsReduced_eq} and \eqref{DiscDetBoundReduced_eq}. 
To this end, we seek for a base $\BC$ of $A$, which will simultaneously minimize the values $\Delta_i\bigl(M(\BC)\bigr)$, for $i \in \intint k$, where $M(\BC) := A_{\BC}^{-1} \cdot A$. Taking $\BC$, such that $\abs{\det(A_{\BC})} = \Delta$, we can make $\Delta_i\bigl(M(\BC)\bigr) = 1$, for all $i \in \intint k$. But it is an NP-hard problem to compute this $\BC$. Instead, we will settle for an approximate solution that can be obtained by a polynomial-time algorithm. The following Theorem, due to A.~Nikolov, gives an asymptotically optimal approximation ratio. 
\begin{theorem}[A.~Nikolov \cite{LargestSimplex_Nikolov}]\label{maxdet_apr_th}
    Let $A \in \ZZ^{k \times n}$, $\rank(A) = k$ and $\Delta := \Delta(A)$. Then there exists a deterministic polynomial-time algorithm that computes a base $\BC$ of $A$ with $\Delta / \abs{\det(A_{\BC})} \leq e^k$.
\end{theorem}

The next Lemma uses an algorithm, due to A.~Nikolov, to compute a relatively good base $\BC$.
\begin{lemma}\label{detlb_submat_lm}
Let $A \in \ZZ^{k \times n}$, $\rank(A) = k$, and $\Delta := \Delta(A)$. 
% For a base $\BC \subseteq \intint n$ of $A$, denote $M(\BC) := (A_{\BC})^{-1} \cdot A$. 
Then there exists a base $\BC$, such that
\begin{enumerate}
    % \item $\Delta_1\bigl(M(\BC)\bigr) \leq 1$;
    
    \item for each $i \in \intint k$, $\Delta_i\bigl(M(\BC)\bigr) \leq e^{i+1}$;

    \item the base $\BC$ can be computed by an algorithm with the computational complexity bound
    $$
    O\bigl( k \cdot 2^k \cdot T_{\text{apr}} \bigr),
    $$ where $T_{\text{apr}}$ is the computational complexity of the algorithm of Theorem \ref{maxdet_apr_th} with an input $A$. Writing the complexity bound, we make an additional assumption that the approximation problem is harder than the matrix inversion.
\end{enumerate}
\end{lemma}
\begin{myproof}
    In the initial step, we compute a base $\BC$, such that $\Delta\bigl(M(\BC)\bigr) \leq e^k$, using Theorem \ref{maxdet_apr_th}.  
    Next, we repeatedly perform the following iterations:
    \begin{algorithmic}[1]
    \State $M \gets M(\BC)$
    \For{$\JC \subseteq \intint k$}
        \State $i \gets \abs{\JC}$
        \State using Theorem \ref{maxdet_apr_th}, compute a base $\IC$ of $M_{\JC *}$, such that $\abs{\det(M_{\JC \IC})} \cdot e^i \geq \Delta_i(M_{\JC *})$
        \If{$\abs{\det(M_{\JC \IC})} > e$}
            \State $\BC \gets \BC \setminus \JC \cup \IC$
            \State \textbf{break}
        \EndIf
    \EndFor
    \end{algorithmic}
    Note that $\bigl(M(\BC)\bigr)_{\BC} = I$, where $I$ is the $k \times k$ identity matrix.
    Hence, if the condition $\abs{\det(M_{\JC \IC})} > e$ will be satisfied, for some $\IC$ and $\JC$, then the value of $\abs{\det(A_{\BC})}$ will grow at least by $e$. Therefore, since initially $ e^{-k} \cdot \Delta(A) \leq \abs{\det(A_{\BC})} \leq \Delta(A)$, it is sufficient to run the described procedure exactly $k$ times. More precisely, we can stop at the moment, when the cycle in Line 2 will be completely finished without calling the $\textbf{break}$-operator of Line 7. After that the condition $\Delta_i\bigl(M(\BC)\bigr) \leq e^{i+1}$ will be satisfied, for all $i \in \intint k$. Clearly, the total computational complexity is bounded by $O(k \cdot 2^k \cdot T_{\text{apr}})$.
\end{myproof}

Assume that the algorithm of Lemma \ref{detlb_submat_lm} has been applied to the matrix $A$ setting $B = A_{\BC}$ for the resulting base $\BC$. Then
$$
\detlb\bigl(M(\BC)\bigr) \leq \max\limits_{t \in \intint k} e^{\frac{t+1}{t}} = e^2,
$$
hence, due to the inequality \eqref{DiscDetBoundReduced_eq},
\begin{equation*}\label{MB_DiscDet_bound}
    \herdisc\bigl(M(\BC)\bigr) = O(\log k).    
\end{equation*}
% Similarly, due to the inequality \eqref{SixDeviationsReduced_eq},
% \begin{equation*}\label{MB_SixDev_bound}
%     \herdisc\bigl(M(\BC)\bigr) = O\bigl(\sqrt{k}\bigr).    
% \end{equation*}

Note that $\abs{\det(A_{\BC})} \geq \Delta/e^k$. Therefore, recalling Lemma \ref{rho_estimate_lm} and assuming that an $2^k \cdot \poly(\phi)$-time pre-processing is done, we can assume that 
\begin{gather}
    \rho = O\bigl(\log(k\cdot r \cdot \Delta)\bigr),\notag\\
    \eta = O(\log k),\label{eta_delta_bounds}\\
    \delta \geq \Delta/e^k.\notag
    % \eta = O\Bigl(\min\bigl\{\sqrt{k}, \sqrt{\log k \cdot \log n}\bigr\}\Bigr),\label{eta_delta_bounds}\\
\end{gather}

\subsection{Putting Things Together}

Let us consider the computational complexity guaranties of our dynamic programming algorithm, given by Theorem \ref{DP_th}. Due to \eqref{eta_delta_bounds}, assuming that an $2^{O(k)} \cdot \poly(\phi)$-time preprocessing has been done, we can assume that 
$$
    \eta = O(\log k),\;\rho = O\bigl(\log(k\cdot r \cdot \Delta)\bigr)
$$ 
and $\delta \geq \Delta/e^k$. Recall that $\tau = (16 \eta)^k \cdot r \cdot \delta$. Since $\delta \geq \Delta/e^k$ and $\eta \geq 1$, we have $\tau = O(\log k)^k \cdot r \cdot \Delta$ and $\tau \geq 2^{(4-\ln 2)k} \cdot r \cdot \Delta \geq r \cdot \Delta$. Hiding the term $2^{O(k)} \cdot \poly(\phi)$, the total number of operations in $\ZZ^k \times \GC$ can be estimated by
\begin{multline*}
    \rho \cdot \tau^2 / 2^{\Omega\bigl(\sqrt{\log \tau}\bigr)} = \log(k r \Delta) \cdot O(\log k)^{2k} \cdot (r \Delta)^2 / 2^{\Omega\bigl(\sqrt{\log(r\Delta)}\bigr)} =  \\
    = O(\log k)^{2k} \cdot (r \Delta)^2 / 2^{\Omega\bigl(\sqrt{\log(r \Delta)}\bigr)}.
\end{multline*}
For the feasibility problem, we have
\begin{multline*}
    \rho \cdot \tau \cdot \log \tau = \log(k r \Delta) \cdot O(\log k)^{k} \cdot (r\Delta) \cdot \log(r \Delta)\\
    = O(\log k)^{k} \cdot (r\Delta) \cdot \log^2(r \Delta).
\end{multline*}
The above formulas satisfy the desired computational complexity bounds, which finishes the proof of Theorem \ref{genILP_main_th}.

%% file: parts/acknowledgements.tex
Main results of our work described in Section \ref{intro_sec} were supported by the Ministry of Science and Higher Education of the Russian Federation (Goszadaniye) No. 075-03-2024-117, project No. FSMG-2024-0025. Secondary results of the Sections \ref{app_sec} and \ref{conv_sec} were prepared within the framework of the Basic Research Program at the National Research University Higher School of Economics (HSE).

%% file: parts/appendices.tex
\section{Proof of Lemma \ref{group_isomorphism_lm}}\label{group_isomorphism_sec}

Let us check that the map $\phi$ is correctly defined, i.e. $\phi(z) \in \GC'$, for each $z \in \GC_I$. Definitely, since $z = A \cdot (t_v + t_z) = A \cdot (q + \lfloor t_v + t_z \rceil)$, for some $q \in \ZZ^n$, it follows that $A \cdot \lfloor t_v + t_z \rceil$ is integer. Consequently, since $\lfloor \cdot \rceil \in [-1,1)^n$ by its definition, we have $\phi(z) \in \GC'$.

Similarly, let us check that $\phi(z') \in \GC'$, for each $z' \in \GC'$. Again, it is only sufficient to prove that $\phi^{-1}(z') \in \ZZ^n$. We have
\begin{multline*}
    \phi^{-1}(z') = A \cdot \bigl( t_v + \lceil t_v - t_{z'} \rfloor \bigr) = \\
    = A \cdot \bigl( t_v + (t_v - t_{z'}) + \lceil t_v - t_{z'} \rfloor - (t_v - t_{z'}) \bigr) = \\
    = A \cdot (q + t_{z'}),\quad\text{for some $q \in (2\cdot\ZZ)^n$.}
\end{multline*}
Therefore, since $z' = A \cdot t_{z'}$ is integer, the same holds for $\phi^{-1}(z')$. 

Let us check that $\phi \circ \phi^{-1} = \phi^{-1} \circ \phi = 1$, as a consequence, it means that $\phi$ is a bijection. For $z = A \cdot (t_v  +t_z) \in \GC_I$, we have
\begin{multline*}
    \phi^{-1}(\phi(z)) = A \cdot \Bigl( t_v + \bigl\lceil t_v - \lfloor t_v + t_z \rceil \bigr\rfloor \Bigr) = A \cdot \Bigl( t_v + \bigl\lceil t_v + (t_v+t_z) - \lfloor t_v + t_z \rceil - (t_v + t_z) \bigr\rfloor \Bigr) = \\
    = A \cdot \Bigl( t_v + \bigl\lceil q - t_z \bigr\rfloor \Bigr),\quad\text{for some $q \in (2\cdot\ZZ)^n$.}
\end{multline*}
Due to the relation \eqref{round_fun_prop_2}, we have
\begin{equation*}
    A \cdot \Bigl( t_v + \bigl\lceil q - t_z \bigr\rfloor \Bigr) = A \cdot \bigl(t_v + \lceil t_z \rfloor\bigr) = A \cdot \bigl(t_v + t_z\bigr) = z,
\end{equation*}
where the second equality holds, because $t_z \in [-1,1)^n$. Similarly, for $z' = A \cdot t_{z'} \in \GC'$, we have
\begin{multline*}
    \phi(\phi^{-1}(z')) = A \cdot \Bigl\lfloor t_v + \bigl\lceil t_v - t_{z'} \bigr\rfloor \Bigr\rceil = A \cdot \Bigl\lfloor t_v + (t_v - t_{z'}) + \bigl\lceil t_v - t_{z'} \bigr\rfloor - (t_v - t_{z'}) \Bigr\rceil =\\
    = A \cdot \bigl\lfloor q + t_{z'} \bigr\rceil,\quad\text{for some $q \in (2\cdot\ZZ)^n$}.
\end{multline*}
Due to the relation \eqref{round_fun_prop_1}, we have
\begin{equation*}
    A \cdot \bigl\lfloor q + t_{z'} \bigr\rceil = A \cdot \lfloor t_{z'} \rceil = A \cdot t_{z'} = z',
\end{equation*}
where the second equality holds, because $t_{z'} \in [-1,1)^n$.

Finally, let us check that $\phi$ is a homomorphism:
\begin{multline*}
    \phi(y \oplus_{\GC_I} z) = \phi\Bigl(A \cdot \bigl(t_v + \lfloor t_v + t_y + t_z \rceil\bigr)\Bigr) = A \cdot \bigl\lfloor t_v + \lfloor t_v + t_y + t_z \rceil \bigr\rceil = \\
    = A \cdot \bigl\lfloor 2t_v + t_y + t_z \bigr\rceil = A \cdot \bigl\lfloor \lfloor t_v + t_y \rceil + \lfloor t_v + t_z \rceil \bigr\rceil = \phi(y) \oplus_{\GC'} \phi(z),
\end{multline*}
where the third and forth equalities use the property \eqref{round_fun_prop_3}.

% \section{Proof of Lemma \ref{tiling_injection_lm}}\label{tiling_injection_sec}
% Let us show that the restriction of $\phi$ onto $\WC$ is injective. Choose arbitrary and different $y,z \in \WC$. Denote $y' = y + t$ and $z' = z + t$. Note that $\phi(y') \ominus \phi(z') = \phi(y) \ominus \phi(z)$. Since $y',z' \in t + \WC \subseteq \PC$, we have $\phi(y') = y'$ and $\phi(z') = z'$. Finally, 
% \begin{equation*}
%     \phi(y)\ominus\phi(z) = \phi(y') \ominus \phi(z') = y' - z' \not= 0,
% \end{equation*}
% which proves the injectivity.

\section{Error Analysis for the Generalized DFT With Respect to the Word-RAM Model}

\subsection{Preliminaries}
%\begin{theorem}[Khinchin]\label{chfr_th}
%  Let $a \in \QQ$ and $n \in \ZZ_{>0}$ be given input values. Then there exists a polynomial-time algorithm returning a rational value $p/q$ such that
%  \begin{enumerate}
%    \item $q < n$;
%    \item the value $\abs{a - p/q}$ is minimized between the all $p/q$ with $q < n$;
%    \item $\abs{a - p/q} \leq \frac{1}{q n}$;
%    \item arithmetic cost of the algorithm is $O(\log(n))$.
%  \end{enumerate}
%\end{theorem}

The following Lemma is very useful for rounding of rationals.
\begin{lemma}\label{Reps_lm}
  Let $a \in \QQ$ and $\varepsilon \in \QQ \cap (0,1)$ be given input numbers. Then, there exists a polynomial-time algorithm that returns a rational number $p/q$ with $q = \lceil 1/\varepsilon \rceil$,
  %represented as $b = z + p/q$, where $z = \lfloor a \rfloor$, 
  such that
  \begin{enumerate}
    \item $\abs{a - p/q} \leq \varepsilon$;
    \item $\size(p/q) = O\bigl(\log\bigl(\abs{a}+1\bigr)+\log(1/\varepsilon)\bigr)$;
    \item the algorithm needs $O(1)$ operations with rational numbers of the size $O\bigl(\size(a) + \size(\varepsilon)\bigr)$.
  \end{enumerate}
\end{lemma}
\begin{myproof}
    Assume that $a$ is represented as $a = p'/q'$ for co-prime integers $p'$ and $q'$. It is sufficient to put $p := \lfloor \frac{p' \cdot q}{q'} \rfloor$. Definitely,
    \begin{multline*}
      \abs{p'/q'-p/q} = \frac{1}{q} \abs{\frac{p' \cdot q}{q'} - \bigl\lfloor \frac{p' \cdot q}{q'} \bigr\rfloor} \leq 1/q = 1/\lceil 1/\varepsilon \rceil \leq \varepsilon.\\
    \end{multline*}
    Let us prove the second Claim. Note that $\size(q) = O\bigl(\log(1/\varepsilon)\bigr)$. Represent $p$ as $p = z \cdot q + r$, for $z \in \ZZ$ and $r \in \intint[0]{q-1}$. Since $\abs{a - z - r/q} = \abs{a - p/q} \leq \varepsilon < 1$, we have $\abs{z} \leq \abs{a} + \abs{a-z} \leq \abs{a} +2$. Finally, $\size(p/q) = O\bigl(\size(z)+\size(q)\bigr) = O\bigl(\log\bigl(\abs{a} + 1\bigr)+\log(1/\varepsilon)\bigr)$.
\end{myproof}

It is very useful to have a variant of Lemma \ref{Reps_lm} for complex numbers.
\begin{lemma}\label{Ceps_lm}
  Let $c = a + i \cdot b \in \CCQ$ and $\varepsilon \in \QQ \cap (0,1)$ be given input numbers. Then, there exists a polynomial-time algorithm that returns a number $z = x + i \cdot y \in \CCQ$, such that
  \begin{enumerate}
    \item $\abs{c-z} \leq \varepsilon$;
    \item $\size(z) = O\bigl(\log\bigl(\abs{c}+1\bigr)+\log(1/\varepsilon)\bigr)$;
    \item The algorithm needs $O(1)$ operations with rational numbers of the size $O\bigl(\size(c) + \size(\varepsilon)\bigr)$.
  \end{enumerate}
\end{lemma}
\begin{myproof}
    Let us use Lemma \ref{Reps_lm} separately to $a$ and $b$ with the accuracy $\varepsilon/2$. Let $x$ and $y$ be the output values. Consequently, $\abs{c-z} = \sqrt{(x-a)^2 + (y - b)^2} \leq \sqrt{\varepsilon^2/2} < \varepsilon$. Due to the construction, we have $\size(z) = O\bigl(\log\bigl(\abs{a}+1\bigr) + \log\bigl(\abs{b}+1\bigr)+\log(1/\varepsilon)\bigr)$. Clearly, $\max\bigl\{\abs{a},\abs{b}\bigr\} \leq \abs{c}$. Therefore, $\log\bigl(\abs{a}+1\bigr) + \log\bigl(\abs{b}+1\bigr) = O\bigl( \log\bigl(  \abs{c} \bigr) \bigr)$, which proves the Lemma.
\end{myproof}

The following Lemma is useful, when we need to approximate a finite sum of rational complex numbers.
\begin{lemma}\label{sum_appr_lm}
  Let $m \in \ZZ_{\geq 1}$, a rational complex number $\varepsilon$ with $0 < \varepsilon < 1/2^m$ and $\{a_j \in \CCQ\}_{j \in \intint[0]{m-1}}$ be given input numbers. Additionally, denote $S = \sum_{j} a_j$, $\gamma = \max_{j} \size(a_j)$, and assume that $\abs{a_j} \leq C$, for some absolute constant $C$ and all $j$. Then, there exists a polynomial-time algorithm that returns a number $S' \in \CCQ$, such that
  \begin{enumerate}
    \item $\abs{S - S'} \leq \varepsilon$;
    \item the algorithm needs $O(m)$ operations with rational numbers of the size $O\bigl(m + \gamma + \size(\varepsilon)\bigr)$;
    \item $\size(S') = O\bigl(\log m + \log(1/\varepsilon)\bigr)$.
  \end{enumerate}
\end{lemma}
\begin{myproof}
    Denote the partial sums of $S$ by $S_j$. We compute their approximate values $S'_j$ in the following way: starting from $S'_0 = a_0$, we compute $S'_j$, applying Lemma \ref{Ceps_lm} to $S'_{j-1} + a_j$ with the accuracy $\varepsilon/2^{m-j+1}$. Let us prove by induction that $\abs{S_j - S'_j} \leq \varepsilon/2^{m-j}$, which will also prove the first Claim. Definitely,
    \begin{multline*}
        \abs{S_j - S'_j} \leq \abs{S_j - S'_{j-1} - a_j} + \abs{S'_{j-1} + a_j - S'_j} \leq \\
        \leq \abs{S_{j-1} - S'_{j-1}} + \varepsilon/2^{m-j+1} \leq \varepsilon/2^{m-j}.
    \end{multline*}
  
    Since $\abs{S_j - S'_j} < 1$, we have $\abs{S'_j} \leq \abs{S_j} + 1 \leq C \cdot m + 1$. Therefore, due to Lemma \ref{Ceps_lm} and construction of $S'_j$, we have $\size(S'_j) = O\bigl(\log m + \log(1/\varepsilon)\bigr)$. All together, it proves the second and third Claims.
\end{myproof}

\subsection{The DFT in Cyclic Groups}

Denoting 
$$
    \epsilon_n = e^{i \cdot \frac{2 \pi}{n}},
$$ the following Lemma helps to approximate $\epsilon^k_n$, for any $k \in \intint[0]{n-1}$. 
\begin{lemma}\label{root_exp_appr_lm}
  Let $n \in \ZZ_{\geq 2}$, $k \in \intint[0]{n-1}$, and $\varepsilon \in \QQ \cap (0,1)$ be given input numbers. Then, there exists a polynomial-time algorithm that returns a number $z \in \CCQ$, such that
  \begin{enumerate}
    \item $\abs{z - \epsilon_n^k} \leq \varepsilon$;
    \item the algorithm needs $O\bigl(\log(1/\varepsilon)\bigr)$ operations with rational numbers of the size $O\bigl(\size(n) + \size(\varepsilon)\bigr)$;
    \item $\size(z) = O\bigl(\log(1/\varepsilon)\bigr)$.
  \end{enumerate}
\end{lemma}
\begin{myproof}
    Using the Taylor expansion for $e^x$, we can write $\epsilon_n^k$ in the form of absolutely convergent series:
    \begin{equation}\label{kroot_series}
        \epsilon_n^k = e^{i \cdot \frac{2 \pi k}{n}} = \sum\limits_{j \geq 0} i^j \cdot \left(\frac{2 \pi k}{n}\right)^j \cdot \frac{1}{j!}. 
    \end{equation}
    Denoting the partial sum of the first $j$ terms of \eqref{kroot_series} by $S_j$, we have
    \begin{equation*}
        \abs{\epsilon_n^k - S_{j_0}} \leq \sum\limits_{j \geq j_0} \frac{(2 \pi)^j}{j!}.
    \end{equation*}
    Due to Stirling's approximation, it follows that, for a sufficiently big constant $C$ and $m := \bigl\lceil C \cdot \log_2(1/\varepsilon)\bigr\rceil$, the inequality $\abs{\epsilon_n^k - S_{m}} \leq \varepsilon/3$ will hold. 

    Unfortunately, it is not enough to compute $S_{m}$ by a straight way, because the size of $S_{m}$ and intermediate variables will become too large. Instead, denoting the $j$-th term of \eqref{kroot_series} by $a_j$, let us use the following auxiliary Lemma.
    \begin{lemma}\label{acoeff_appr_lm}
        Let $\varepsilon \in \QQ \cap (0,1)$ and $m \in \ZZ_{\geq 1}$ be the input numbers. Then, there exists a polynomial-time algorithm that returns a sequence $\{a'_j \in \CCQ\}_{j \in \intint[0]{m-1}}$, such that
        \begin{enumerate}
            \item $\abs{a'_j - a_j} \leq \varepsilon$, for any $j \in \intint[0]{m-1}$;
            \item the algorithm needs $O(m)$ operations with rational numbers of the size $O\bigl(\size(n) + \size(\varepsilon)\bigr)$;
            \item $\size(a'_j) = O\bigl(\log(1/\varepsilon)\bigr)$, for any $j \in \intint[0]{n-1}$.
        \end{enumerate}
    \end{lemma}
    \begin{myproof}
        Let $\pi'$ be a rational approximation of $\pi$ with accuracy $\varepsilon/C$, for a sufficiently large constant $C$,  which can be constructed without any difficulties. We calculate $a'_j$ in the following way: starting from $a'_0 = 1$, we calculate $a'_j$, applying Lemma \ref{Ceps_lm} to $\frac{2 \pi' k}{j n} \cdot a_{j-1}$ with the accuracy $\varepsilon/C$. Let us prove the first Claim, using the induction principle. Denoting $\delta_{\pi} = \pi - \pi'$ and $\delta_j = a_j - a'_j$, we have 
        \begin{multline}\label{acoeff_appr}
            \abs{\delta_j} \leq \abs{a_j - \frac{2 \pi' k}{j n} \cdot a'_{j-1}} + \abs{a'_j - \frac{2 \pi' k}{j n} \cdot a'_{j-1}} \leq \\
            \leq \abs{\frac{2 \pi k}{j n} \cdot a_{j-1} - \frac{2 \pi' k}{j n} \cdot a'_{j-1}} + \varepsilon/C \leq \\
             \leq \frac{2}{j} \cdot \Bigl( \pi \cdot \abs{\delta_{j-1}} + \abs{a_{j-1}} \cdot \abs{\delta_{\pi}} + \abs{\delta_{\pi}} \cdot \abs{\delta_{j-1}} \Bigr) + \varepsilon/C.
        \end{multline}
        Let $j_0 \geq C$ be chosen, such that $\abs{a_{j}} \leq 1$, for all $j \geq j_0$. Then, due to \eqref{acoeff_appr}, for all $j \geq j_0$, we have
        \begin{multline*}
            \abs{\delta_j} \leq \frac{2}{j_0} \cdot \Bigl( \pi \cdot \varepsilon + \varepsilon/C + \varepsilon^2/C \Bigr) + \varepsilon/C \leq \frac{2\cdot (\pi+2)}{C} \cdot \varepsilon + \varepsilon/C \leq \varepsilon \cdot \frac{13}{C}, \\
        \end{multline*} 
        which proves the first Claim for $j \geq j_0$, putting $C \geq 13$. Since $\abs{a_j}$ is bounded, $\delta_0 = 0$, and, due to \eqref{acoeff_appr}, the Claim with respect to smaller values of $j$ can be satisfied, taking sufficiently large $C$.
        
        Due to Lemma \ref{Ceps_lm} and due to the construction of $a'_j$, the third Claim is also satisfied. Due to the third Claim, the values $a'_j$ can be computed with $O(m)$ operations with rational numbers of the size $O\bigl(\size(n) + \size(\varepsilon)\bigr)$, which proves the second Claim and this Lemma.
    \end{myproof}

    Let us continue the proof of Lemma \ref{root_exp_appr_lm}. Using Lemma \ref{acoeff_appr_lm}, we compute the sequence $\{a'_j\}_{j \in \intint[0]{m-1}}$ with the accuracy $\varepsilon/(3 m)$, which can be done, using $O(m)$ operations with rational numbers of the size $O\bigl(\size(n) + \size(\varepsilon)\bigr)$. Next, we use Lemma \ref{sum_appr_lm} to compute the resulting value $z$ as an approximation of the sum $S'_m = \sum_{i \in \intint[0]{m-1}} a'_j$ with the accuracy $\varepsilon/3$. Due to Lemma \ref{sum_appr_lm}, the last step needs $O(m)$ operations with rational numbers of the size $O\bigl(\log m + \size(\varepsilon)\bigr) = O\bigl(\size(\varepsilon)\bigr)$. Let us check that $\abs{\epsilon_n^k - z} \leq \varepsilon$. Definitely,
    \begin{multline*}
      \abs{\epsilon_n^k - z} = \abs{\epsilon_n^k - S'_m} + \abs{S'_m - z} \leq \\
      \leq \abs{\epsilon_n^k - S_m} + \abs{S_m - S'_m} + \varepsilon/3 \leq \varepsilon/3 + m \cdot \frac{\varepsilon}{3m} + \varepsilon/3 \leq \varepsilon,
    \end{multline*}
    which finishes the proof.
\end{myproof}

Denote the Vandermonde matrix with respect to the values $(\epsilon_n^0, \epsilon_n^1, \dots, \epsilon_n^{n-1})$ by $V_n$. That is, $V_n \in \CC^{n \times n}$ and 
$$
    \bigl(V_n\bigr)_{i j} = \epsilon_n^{i j}, \quad \text{for $i,j \in \intint[0]{n-1}$}.
$$
Note that the matrix-vector product $V_n \cdot v$ corresponds (with respect to natural ordering of elements in $\GC$ and $\widehat\GC$ and choosing a natural basis in $\CC^n$) to the DFT of $v$ as a member of $\CC[\CCal_n]$, where $\CCal_n$ denotes the cyclic group of order $n$. The seminal Cooley\&Tukey algorithm \cite{CooleyTukey} gives an $O\bigl(n \cdot (p_1 + p_2 + \dots + p_s) \bigr)$-time algorithm in exact complex arithmetic to compute $V_n \cdot v$, where $n = p_1 \cdot p_2 \cdot \ldots \cdot p_s$ is the prime factorisation of $n$. The most popular implementation of the Cooley \& Tukey algorithm assumes that $n = 2^s$. In Algorithm \ref{CT_exact_alg}, we refer to a standard recursive implementation of this case.
\begin{algorithm}
	\caption{Recursive implementation of the Culey\&Tukey algorithm for $n = 2^s$}
    \label{CT_exact_alg}
    \begin{algorithmic}[1]
    \Require Input vector $v \in \CC^n$, for $n = 2^s$;
    \Ensure Output vector $\hat v = V_n \cdot v$;
    \Procedure{CT}{$v$}
        \State $n \gets \dim(v)$;
        \If{$n = 1$}
            \State return $v$;
        \Else
            %\State Let $v_{even}$ and $v_{odd}$ be the vectors composed of the components of $v$ with the even and odd indices respectively 
            \State $v_{even} \gets (v_0, v_2, \dots, v_{n-2})$;
            \State $v_{odd} \gets (v_1, v_3, \dots, v_{n-1})$;
            \State $u_{even} \gets \Call{CT}{v_{even}}$;
            \State $u_{odd} \gets \Call{CT}{v_{odd}}$;
            \State $\hat v \gets \BZero_n$;
            \For{$k \in \intint[0]{n/2-1}$}
                \State $\hat v_{k} \gets (u_{even})_k + \epsilon_{n}^k \cdot (u_{odd})_k$;
                \State $\hat v_{n/2 + k} \gets (u_{even})_k - \epsilon_{n}^k \cdot (u_{odd})_k$;
            \EndFor
            \State \textbf{return} $\hat v$;
        \EndIf
    \EndProcedure
    \end{algorithmic}
\end{algorithm}

\begin{remark}[Approximation of $\epsilon_n^k$]\label{root_appr_rm}
  Before presenting our approximate version of the Cooley\&Tukey algorithm for an input vector $v' \in \CCQ^n$ with $n=2^s$ and an accuracy $\varepsilon \in \QQ \cap (0,1)$, we need to provide good approximations of the values $\epsilon_j^k$, for all $j \in \{1, 2, 4, \dots, 2^s\}$ and $k \in \intint[0]{j-1}$. Denoting the maximum size of components of $v'$ by $\gamma$, we choose an accuracy $\varepsilon' = \varepsilon \cdot \frac{1}{n \cdot 2^{C \cdot \gamma}}$, for a sufficiently large constant $C$. Let us denote the resulting approximate version of $\epsilon_j^k$ with the accuracy $\varepsilon'$ by $c_{j,k}$. 
  
  Due to Lemma \ref{root_exp_appr_lm}, the values of $c_{n,k}$, for all $k \in \intint[0]{n-1}$, can be computed with $O\bigl( n \cdot \bigl(\log n + \log(1/\varepsilon')\bigr)\bigr) = O\bigl(n \cdot \bigl( \log n + \gamma + \log(1/\varepsilon)\bigr)\bigr)$ operations with rational numbers of the size $O\bigl(\log  n + \size(\varepsilon)\bigr)$. The values of $c_{j,k}$, for all $j \in \{1,2,4,\dots,2^s\}$ and $k \in \intint[0]{j-1}$, can be easily derived from the values $c_{n,k}$, so it does not change the complexity. Additionally, note that $\size(c_{j,k}) = O\bigl(\log n + \log(1/\varepsilon)\bigr)$.
\end{remark}

Our approximate variation of the Cooley-Tukey algorithm, described in Algorithm \ref{CT_appr_alg}, differs from the original one just by using approximate calculations of the type $z = a + b \cdot \epsilon^k_j$. Additionally, we use Lemma \ref{Ceps_lm} to guaranty that intermediate variables will have a bounded bit-encoding size.
\begin{algorithm}
    \caption{Approximate version of Algorithm \ref{CT_exact_alg}}
    \label{CT_appr_alg}
    \begin{algorithmic}[1]
        \Require An input vector $v \in \CCQ^n$, for $n = 2^s$, and an input accuracy $\varepsilon \in \QQ \cap (0,1)$;
        \Ensure An output vector $\hat v \in \CCQ$, which is an approximate version of $V_n \cdot v$;
        \Procedure{ApproxCT}{$v,\varepsilon$}
        \State $n \gets \dim(v)$;
        \If{$n = 1$}
            \State return $v$;
        \Else
            \State $v_{even} \gets (v_0, v_2, \dots, v_{n-2})$;
            \State $v_{odd} \gets (v_1, v_3, \dots, v_{n-1})$;
            \State $u_{even} \gets \Call{ApproxCT}{v_{even},\varepsilon}$;
            \State $u_{odd} \gets \Call{ApproxCT}{v_{odd},\varepsilon}$;
            \State $\hat v \gets \BZero_n$;
            \For{$k \in \intint[0]{n/2-1}$}
                \State $\hat v_{k} \gets (u_{even})_k + c_{n,k} \cdot (u_{odd})_k$;
                \State $\hat v_{n/2 + k} \gets (u_{even})_k - c_{n,k} \cdot (u_{odd})_k$;
                \State Apply Lemma \ref{Ceps_lm} to $\hat v_{k}$ and $\hat v_{n/2 + k}$ with the accuracy $\varepsilon$; \label{ApproxCT:line:prunning}
            \EndFor
            \State \textbf{return} $\hat v$;
        \EndIf
        \EndProcedure
    \end{algorithmic}
\end{algorithm}

\begin{lemma}\label{CT_approx_lm}
  Let $v' \in \CCQ^n$ with $n = 2^s$ and a rational number $\varepsilon$ with $\varepsilon < (1/5)^s$ be given input numbers. Additionally, assume that there exists $v \in \CC^n$, such that $\norm{v-v'}_{\infty} \leq \varepsilon$. Then, there exists an Word-RAM approximate implementation of the Cooley-Tukey algorithm that returns a vector $\hat v'$, such that
  \begin{enumerate}
    \item $\norm{\hat v'-\hat v}_{\infty} \leq 5^s \cdot \varepsilon$, where $\hat v = V_n \cdot v$; 
    \item the algorithm needs $O(n \cdot s )$ operations with rational numbers of the size $O\bigl(s + \log(1/\varepsilon) + \gamma\bigr)$, where $\gamma = \max_{i} \size(v'_i)$ is the maximum of component sizes of $v'$.
  \end{enumerate} 
\end{lemma}
\begin{myproof}
    Roughly speaking, the exact version of the Cooley\&Tukey algorithm (Algorithm \ref{CT_exact_alg}) consists of $s$ levels. Fixing some level, each value $z$ of this level is calculated from two values $a,b$ of the previous level by the formula $z = a + b \cdot \epsilon_{j}^k$, for $j = 2^{l-1}$ and $k \in \intint[0]{j-1}$, according to Algorithm \ref{CT_exact_alg}. For the level $l \in \intint s$, let us denote the values, computed on this level by $\LC_l$. Let us show that $\abs{z} = n \cdot 2^{O(\gamma)}$, for any $z \in \LC_l$ and $l \in \intint s$. Definitely, for $z = a + b \cdot \epsilon^k_j \in \LC_l$, where $a,b \in \LC_{l-1}$, since $\abs{z} \leq \abs{a} + \abs{b}$, it directly follows from the induction principle that $\abs{z} \leq 2^{l} \cdot \norm{v}_{\infty} \leq n \cdot \norm{v}_{\infty}$. Since $\norm{v - v'}_{\infty} \leq \varepsilon < 1$, we have $\norm{v}_{\infty} \leq \norm{v'}_\infty + 1 = 2^{O(\gamma)}$. Therefore, we get that $\abs{z} = n \cdot 2^{O(\gamma)}$. 
  
  Similarly, let us denote the values, calculated on the $l$-th level of our approximate algorithm (Algorithm \ref{CT_appr_alg}) by $\LC'_l$, for any $l \in \intint s$. For an arbitrary $l \in \intint s$, let $z = a + b \cdot \epsilon_{j}^k \in \LC_l$ with $a,b \in \LC_{l-1}$, according to Algorithm \ref{CT_exact_alg}. Similarly, let $y = a' + c_{j,k} \cdot b'$ with $a',b' \in \LC'_{l-1}$ and $z' \in \LC'_l$ be computed from $y$, using Lemma \ref{Ceps_lm}, according to Algorithm \ref{CT_appr_alg} (see the line~\ref{ApproxCT:line:prunning}). Note that $z',a',b'$ are approximate versions of $z,a,b$. We need to prove the following Claims, using the induction principle:
  \begin{enumerate}
    \item $\abs{z - z'} \leq 5^{l} \cdot \varepsilon$; 
    \item $\size(z') = O\bigl(\gamma + \log(1/\varepsilon)\bigr)$;
    \item all the values of the level $\LC'_l$ can be computed with $O(n)$ operations with rational numbers of the size $O\bigl(\gamma + \log(1/\varepsilon)\bigr)$.
  \end{enumerate}
  Denote $\delta_a = a - a'$, $\delta_b = b - b'$, $\delta_z = z - z'$, $\delta_y = z - y$ and $\delta_c = \epsilon_j^k - c_{j,k}$. By induction, we have that $\abs{\delta_a} \leq 5^{l-1} \cdot \varepsilon$ and $\abs{\delta_b} \leq 5^{l-1} \cdot \varepsilon$. To prove the first Claim, we need to show that $\abs{\delta_z} \leq 5^{l} \cdot \varepsilon$. We have
  \begin{equation*}
    \abs{\delta_y} \leq \abs{\delta_a} + \abs{b \cdot \epsilon_j^k - b' \cdot c_{j,k}} \leq \abs{\delta_a} + \abs{\delta_c} \cdot \abs{\delta_b} + \abs{\delta_c} \cdot \abs{b} + \abs{\delta_b}.
  \end{equation*}
  Due to Remark \ref{root_appr_rm}, by the construction of $c_{j,k}$, we have $\abs{\delta_c} \leq \varepsilon \cdot \frac{1}{n \cdot 2^{C \cdot \gamma}}$. Recall that $\abs{b} = n \cdot 2^{O(\gamma)}$. Consequently, taking a sufficiently large $C$, we can assume that $\abs{\delta_c} \cdot \abs{b} \leq \varepsilon$. Consequently,
  \begin{equation*}
    \abs{\delta_y} \leq 5^{l-1} \cdot \varepsilon + \frac{5^{l-1}}{n \cdot 2^{C \cdot \gamma}} \cdot \varepsilon^2 + \varepsilon + 5^{l-1} \cdot \varepsilon \leq 4 \cdot 5^{l} \cdot \varepsilon. \\
  \end{equation*}
  By Lemma \ref{Ceps_lm}, $\abs{z'-y} \leq \varepsilon$. Consequently, 
  $$
  \abs{\delta_z} \leq \abs{z - y} + \abs{y - z'} \leq 4 \cdot 5^{l-1} \cdot \varepsilon + \varepsilon \leq 5^{l} \cdot \varepsilon, 
  $$ which proves the first Claim. Let us prove the second Claim. We have
  $$
  \abs{y} \leq \abs{z} + \abs{\delta_z} = n \cdot 2^{O(\gamma)} + 1 = n \cdot 2^{O(\gamma)}.
  $$ By Lemma \ref{Ceps_lm}, 
  $$
    \size(z') = O\bigl(\log\bigl(\abs{y}+1\bigr) + \log(1/\varepsilon)\bigr) = O\bigl(\gamma + \log n + \log(1/\varepsilon)\bigr),
  $$ which proves the second Claim. Due to Lemma \ref{Ceps_lm} and the structure of Algorithm \ref{CT_appr_alg}, we need only $O(1)$ operations to compute the value $z'$ from $a'$, $b'$ and $c_{j,k}$. The size of these values is bounded by $O\bigl( \log n + \gamma + \log(1/\varepsilon)\bigr)$, due to the second Claim and Remark \ref{root_appr_rm}. Since $\abs{\JC'_l} = n$, this reasoning proves the third Claim. Since all the Claims are proven, it finishes the proof of our Lemma.
\end{myproof}

Following Bluestein \cite{BluesteinDFT} and Baum, Clausen \& Tietz \cite{AbelianFFTImproved_inC}, the next Lemma gives an efficient way to multiply $V_n \cdot v$ for general values of $n$.
\begin{lemma}\label{cyclic_DFT_lm}
  Let $n \in \ZZ_{\geq 2}$, $v' \in \CCQ^n$ and a rational number $\varepsilon$ with $0 < \varepsilon < 1/n^{C}$, for a sufficiently large absolute constant $C$, be given input numbers. Additionally, assume that there exists $v \in \CC^n$, such that $\norm{v-v'}_{\infty} \leq \varepsilon$. Then, there exists a polynomial-time algorithm that returns a vector $\hat v'$, such that
  \begin{enumerate}
    \item $\norm{\hat v -\hat v'}_{\infty} = n^{O(1)} \cdot \varepsilon$, where $\hat v = V_n \cdot v$; 
    \item the algorithm needs $O(n \log n)$ operations with rational numbers of the size $O\bigl(\size(n) + \log(1/\varepsilon) + \gamma\bigr)$, where $\gamma = \max_{i} \size(v'_i)$ is the maximum of component sizes of $v'$.
  \end{enumerate} 
\end{lemma}
\begin{myproof}
    Following the description of Blustein's method \cite{BluesteinDFT}, given by Baum, Clausen \& Tietz \cite{AbelianFFTImproved_inC}, we rewrite the formula 
  $$
    \hat v_k = \sum\limits_{0 \leq j < n} v_j \cdot \varepsilon_n^{k j},
  $$ using $2 k j = k^2 + j^2 - (k - j)^2$, to
  $$
    \hat v_k = \epsilon_n^{k^2/2} \cdot \sum\limits_{0 \leq j < n} v_j \cdot \epsilon_n^{j^2/2} \cdot \epsilon_n^{-(k-j)^2/2}.
  $$
  Denoting $a_j = v_j \cdot \epsilon_{2n}^{j^2}$ and $b_k = \epsilon_{2n}^{-k^2}$, we obtain
  $$
  \epsilon_{2n}^{-k^2} \cdot \hat v_k = \sum\limits_{0 \leq j< n} a_j \cdot b_{k-j}.
  $$ In other words, the computation of $\hat v$ can be reduced to the cyclic convolution $a \star b$ of length $n$. It turns out that it is more efficient to simulate this convolution with respect to a cyclic convolution of increased length. Following Bluestein's original approach, let $N = 2^{\lceil \log_2 2n \rceil}$ and $\CCal_N = \langle g \rangle$ be a cyclic group of order $N$ with a generator $g$. Define $\alpha, \beta \in \CC[\CCal_N]$ in the following way:
  \begin{gather*}
    \alpha = \sum\limits_{0 \leq j < n} a_j \cdot g^j, \\
    \beta = \sum\limits_{0 \leq j < n} b_j \cdot \bigl( g^j + (-1)^n \cdot g^{N-n+j} \bigr).
  \end{gather*}
  Noticing that $N \geq 2n$, $g^N = 1$ and $b_{j+n} = (-1)^n \cdot b_{j}$, we have
  \begin{multline*}
    \alpha \star \beta = \sum\limits_{0 \leq k < n} \left( \sum\limits_{\substack{0 \leq i,j < n\\i+j = k}} a_i \cdot b_j + (-1)^n \cdot \sum\limits_{\substack{0 \leq i,j < n\\i+j-n = k}} a_i \cdot b_j \right) \cdot g^k  + g^n \cdot \bigl( \ldots \bigr) = \\
    = \sum\limits_{0 \leq k < n} \left( \sum\limits_{0 \leq j  \leq k} a_j \cdot b_{k-j} + (-1)^n \cdot \sum\limits_{k+1 \leq j < n} a_j \cdot b_{k-j+n} \right) \cdot g^k  + g^n \cdot \bigl( \ldots \bigr) = \\
    = \sum\limits_{0 \leq k < n} \left( \sum\limits_{0\leq j < n} a_j \cdot b_{k-j}\right) \cdot g^k  + g^n \cdot \bigl( \ldots \bigr).
  \end{multline*} 
  Hence, the values of $\alpha \star \beta$ in the points $g^k$ with $k \in \intint[0]{n-1}$ are exactly $\epsilon_{2n}^{-k^2} \cdot \hat v_k$.
  
  Choosing natural ordering and basis in $\GC$, $\widehat\GC$, $\CC^n$ and identifying $\alpha$, $\beta$, $\alpha \star \beta$ with the corresponding vectors in $\CC^n$, we can write $\alpha \star \beta = V_n^{-1}\bigl((V_n \cdot \alpha) \cdot (V_n \cdot \beta)\bigr)$. Therefore, the computation of $\hat v$ in the exact complex arithmetic can be done with $O(N \log N) = O(n \log n)$ operations, using the exact version of the Cooley\&Tukey algorithm (Algorithm \ref{CT_exact_alg}).
  
  Now, let us construct an approximate Word-RAM version of the described method. Firstly, we calculate approximate versions $a'_j = v'_j \cdot c_{2n,j^2}$ and $b'_k = c_{2n,-k^2}$ of $a_j$ and $b_k$. Here, as in the proof of previous Theorem \ref{CT_approx_lm}, the values $c_{j,k}$ are approximate versions of $\epsilon_j^k$ with the accuracy $\varepsilon \cdot \frac{1}{n \cdot 2^{C \cdot \gamma}}$, for a sufficiently large absolute constant $C$ (see Remark \ref{root_appr_rm}). Note that $\size(a'_j) = O\bigl(\gamma + \log n + \log(1/\varepsilon)\bigr)$ and, due to Remark \ref{root_appr_rm}, the same holds for $b'_j$. Denote $\delta_c = \max_{k \in \intint[0]{2n-1}} \abs{\epsilon_{2n}^k - c_{2n,k}}$, $\delta_v = \norm{v - v'}_\infty$, $\delta_a = \norm{a - a'}_\infty$, and $\delta_b = \norm{b - b'}_\infty$. Since $\norm{v}_\infty \leq 2^\gamma$ and due to Remark \ref{root_appr_rm}, we have
  \begin{gather*}
    \delta_a \leq \delta_v + \norm{v}_{\infty} \cdot \delta_c + \delta_v \cdot \delta_c \leq 3 \varepsilon, \quad \text{and} \\
    \delta_b \leq \delta_c \leq \frac{\varepsilon}{n \cdot 2^{C \cdot \gamma}}, \quad \text{for a sufficiently large absolute constant $C$}.
  \end{gather*}
  
   Construct the vectors $\alpha', \beta' \in \CCQ[\CCal_N]$ from the vectors $a'$ and $b'$ in the same way as the vectors $\alpha, \beta$ were constructed from $a$ and $b$. Denoting $\delta_\alpha = \norm{\alpha-\alpha'}_{\infty}$ and $\delta_\beta = \norm{\beta-\beta'}_{\infty}$, by construction, we have $\delta_\alpha \leq \delta_a \leq 3 \varepsilon$ and $\delta_\beta \leq \delta_b \leq \varepsilon$. Let $\hat \alpha, \hat \beta \in \CC[\widehat{\GC}]$ be the images of the DFT for $\alpha, \beta$. Next, construct the vectors $\hat \alpha'$ and $\hat \beta'$, using the approximate version of the Cooley\&Tukey algorithm (Algorithm \ref{CT_appr_alg}) with input vectors $\alpha'$ and $\beta'$, respectively, and the accuracy $\varepsilon$. Due to Lemma \ref{CT_approx_lm}, this step costs of $O(N \log N) = O(n \log n)$ operations with rational numbers of the size $O\bigl(\log N + \gamma + \log(1/\varepsilon)\bigr)$. Denote $\delta_{\hat \alpha} = \norm{\hat \alpha - \hat \alpha'}$ and $\delta_{\hat \beta} = \norm{\hat \beta - \hat \beta'}$. Due to Lemma \ref{CT_approx_lm}, we have
   \begin{gather*}
     \delta_{\hat \alpha} \leq 5^{\log_2(N)}\cdot \delta_\alpha = 3 \cdot 5^{\lceil \log_2(2n) \rceil} \cdot \varepsilon = n^{O(1)} \cdot \varepsilon, \quad \text{and} \\
     \delta_{\hat \beta} = 5^{\log_2(N)}\cdot \delta_\beta = n^{O(1)} \cdot \varepsilon/2^{C \cdot \gamma}.
   \end{gather*}
   
   Now, we calculate $\hat \alpha' \cdot \hat \beta'$, which costs of $O(N)$ operations with rational numbers of the size $O\bigl(\log n + \gamma + \log(1/\varepsilon)\bigr)$. Denoting $\delta_{\hat \alpha, \hat \beta} = \norm{\hat \alpha \cdot \hat \beta - \hat\alpha' \cdot \hat \beta '}_\infty$, we have
   \begin{multline*}
       \delta_{\hat \alpha, \hat \beta} \leq \delta_{\hat \alpha} \cdot \norm{\hat \beta}_\infty + \delta_{\hat \beta} \cdot \norm{\hat \alpha}_\infty + \delta_{\hat \alpha} \cdot \delta_{\hat \beta} = \\
       = n^{O(1)} \cdot N \cdot \varepsilon + \frac{n^{O(1)} \cdot \varepsilon}{2^{C \cdot \gamma}} \cdot N \cdot 2^\gamma + \frac{n^{O(1)} \cdot \varepsilon^2}{2^{C\cdot\gamma}} = n^{O(1)} \cdot \varepsilon.
   \end{multline*}
   Denote $\psi = \alpha \star \beta$. Finally, we compute resulting approximation $\psi'$ of $\psi$, applying Lemma \ref{CT_approx_lm} to $\hat \alpha' \cdot \hat \beta'$ with the accuracy $\varepsilon$ and dividing the resulting value by $N$. Since the size of elements of $\hat \alpha' \cdot \hat \beta'$ is bounded by $O\bigl(\log n + \gamma + \log(1/\varepsilon)\bigr)$, and due to Lemma \ref{CT_approx_lm}, last step costs of $O(N \log N) = O(n \log n)$ operations with the rational numbers of the same size. Denoting $\delta_\psi = \norm{\psi' - \psi}_{\infty}$, we have
   $$
   \delta_\psi = 5^{\log_2(N)} \cdot \delta_{\hat \alpha, \hat \beta} = n^{O(1)} \cdot \varepsilon.
   $$
   To extract the vector $\hat v'$, which is an approximate version of $\hat v$, we take the first $n$ values of the vector $\psi'$ and multiply them by $c_{2n,k^2}$, for $k \in \intint[0]{n-1}$. 
   
   Finally, let us estimate an error in the resulting vector $\hat v'$:
   \begin{multline*}
     \norm{\hat v - \hat v'}_{\infty} \leq \delta_c \cdot \norm{\psi}_{\infty} + \delta_\psi + \delta_\psi \cdot \delta_c = \\
     \frac{\varepsilon}{n \cdot 2^{C\cdot\gamma}} \cdot n \cdot 2^\gamma + n^{O(1)} \cdot \varepsilon + n^{O(1)} \cdot \varepsilon^2 /2^{C\cdot\gamma} = n^{O(1)} \cdot \varepsilon,
   \end{multline*}
   which finishes the proof.
\end{myproof}

\begin{remark}[Complexity to compute the DFT $N$-times]\label{cyclic_DFT_rm}
    Assume that Lemma \ref{cyclic_DFT_lm} was used $N$ times for different input vectors $v \in \CCQ^n$ with the same accuracy $\varepsilon$. Let us estimate the total complexity together with construction of the approximated values $c_{j,k}$ of $\epsilon_j^k$, for $j \in \intint[0]{n-1}$ and $k \in \intint[0]{j-1}$. Due to Lemma \ref{cyclic_DFT_lm} and Remark \ref{root_appr_rm}, we need
    $$
    O\Bigl(n \cdot \bigl( N \cdot \log n + \gamma + \log(1/\varepsilon)\bigr)\Bigr)
    $$ operations with rational numbers of the size
    $$
    O\bigl(\log n + \gamma + \size(\varepsilon)\bigr).
    $$
\end{remark}

\subsection{Proof of Theorem \ref{group_DFT_th}}\label{group_DFT_proof}

% Finally, let us present the proof of Theorem \ref{group_DFT_th}.
\begin{myproof}
    As it was already noted in Subsection \ref{generalized_DFT_subs}, taking a natural order of elements in $\GC$ and $\widehat{\GC}$ and choosing the standard basis in $\CC^n$, we can identify the function $\alpha \in \CC[\GC]$ with the corresponding vectors in $\CC^n$. Moreover, the DFT can be identified with the matrix-vector multiplication $W_{\GC} \cdot \alpha$. For $A \in \CC^{n \times n}$ and $B \in \CC^{m \times m}$, let $A \otimes B \in \CC^{(nm) \times (nm)}$ be their \emph{Kronecker product}. It can be directly checked that, if $\GC = \HC \oplus \QC$, then $W_{\GC} = W_{\HC} \otimes W_{\QC}$. Consequently, the DFT in $\CC[\GC]$ can be reduced to $\abs{\HC}$ uses of the DFT transform in $\CC[\QC]$ plus $\abs{\QC}$ uses of the DFT transform in $\CC[\HC]$. This approach is known as the Good\&Thomas FFT.

  For any $g \in \GC$, denote by $\chi_g$ the isomorphic image of $g$, corresponding to the mentioned isomorphism (see Section \ref{generalized_DFT_subs}). between $\GC$ and $\widehat\GC$ in accordance with the chosen order of elements. In these notations, the precise description of the multiplication $\hat\alpha = W_{\GC} \cdot \alpha = \bigl( W_{\HC} \otimes W_{\QC} \bigr) \cdot \alpha$ is given in Algorithm \ref{kron_prod_alg}.
  % and denoted by $AbelianDFT(\QC \oplus \HC, \alpha)$ on an input vector $\alpha$. 
  \begin{algorithm}[h]
      \caption{The DFT for the group $\QC\oplus\HC$}
      \label{kron_prod_alg}
      \begin{algorithmic}[1]
      \Require Abelian groups $\HC$ and $\QC$, equipped by the matrices $W_{\HC}$ and $W_{\QC}$, an input vector $\alpha \in \CC[\QC\oplus\HC]$; 
      \Ensure Output vector $\hat \alpha = W_{\GC} \cdot \alpha = \bigl( W_{\HC} \otimes W_{\QC} \bigr) \cdot \alpha$;
      \Procedure{AbelianDFT}{$\QC\oplus\HC,\alpha$}
          \For{$h \in \HC$}
            \State Compose a vector $\alpha_h \in \CC[\QC]$ by the rule: $\alpha_h(q) \gets \alpha(h+q)$, for $q \in \QC$;
          \EndFor
          \For{$h \in \HC$}
            \State Apply the DFT with respect to $\CC[\QC]$: $\hat \alpha_h \gets W_{\QC} \cdot \alpha_h$;
          \EndFor
          \For{$q \in \QC$}
            \State Compose a vector $\beta_q \in \CC[\HC]$ by the rule: $\beta_q(h) = \hat \alpha_h(\chi_q)$, for $h \in \HC$;
          \EndFor
          \For{$q \in \QC$}
            \State Apply the DFT with respect to $\CC[\HC]$: $\hat \beta_q \gets W_{\HC} \cdot \beta_q$;
          \EndFor
          \For{$q \in \QC$, $h \in \HC$}
            \State $\hat \alpha(\chi_q \cdot \chi_h) \gets \hat \beta_q(\chi_h)$;
          \EndFor
          \State \textbf{return} $\hat \alpha$;
      \EndProcedure
      \end{algorithmic}
  \end{algorithm}

  Let us prove by induction that Theorem \ref{group_DFT_th} holds for an arbitrary Abelian group, consisting of $s$ or fewer elements in its basis. The precise formulation is follows: {\it There exist absolute constants $I_1$, $I_2$, $I_3$, and $I_4$, such that the following five Claims hold: 
  \begin{enumerate}
    \item[(1)] for any $\varepsilon < 1/n^{I_1}$, $\norm{\hat \alpha - \hat \alpha'}_{\infty} \leq \abs{\GC}^{I_1} \cdot \varepsilon$;
    \item[(2)] the size of intermediate variables is bounded by $I_2 \cdot \bigl( \log_2 \abs{\GC} + \gamma + \size(\varepsilon) \bigr)$;
    \item[(3)] if the algorithm is used $N$ times on different inputs from $\CCQ[\GC]$, the total arithmetic cost is bounded by $I_3 \cdot \abs{\GC} \cdot \bigl(N \cdot \log_2\abs{\GC} + \gamma + \log_2(1/\varepsilon)\bigr)$;
    \item[(4)] the size of elements of the resulting vector $\hat\alpha'$ is bounded by $I_4 \cdot \bigl( \log_2 \abs{\GC} + \gamma + \size(\varepsilon) \bigr)$;
    \item[(5)] the constant $I_4$ is independent on $I_1$, $I_2$, and $I_3$.
  \end{enumerate}
  } %Note that the statement is trivial for $s = 1$.
  
  Assume that the group $\GC$ is represented by a basis $b_1, b_2, \dots, b_s$. Denote $\HC = \langle b_1 \rangle$ and $\QC = \langle b_2 \rangle \oplus \dots \oplus \langle b_s \rangle$. The proof uses an approximate version of the Algorithm \ref{kron_prod_alg}, described as Algorithm \ref{appr_kron_prod_alg}, which returns approximate version $\hat \alpha'$ of $\hat \alpha = W_{\GC} \cdot \alpha$.
  \begin{algorithm}
    \caption{Approximate version of Algorithm \ref{kron_prod_alg}}
      \label{appr_kron_prod_alg}
      \begin{algorithmic}[1]
        \Require Cyclic group $\HC$ and Abelian group $\QC$, represented by their bases, an input vector $\alpha' \in \CCQ[\QC\oplus\HC]$, an input accuracy $\varepsilon \in \QQ\cap(0,1)$; 
        \Ensure Output vector $\hat \alpha'$, which is an approximate version of $\hat \alpha$;
        \Procedure{ApproxAbelianDFT}{$\QC\oplus\HC,\alpha',\varepsilon$}
        \For{$h \in \HC$}
            \State Compose a vector $\alpha'_h \in \CCQ[\QC]$ by the rule: $\alpha'_h(q) \gets \alpha'(h+q)$, for $q \in \QC$;
          \EndFor
          \For{$h \in \HC$}
            \State $y_h \gets \Call{ApproxAbelianDFT}{\QC,\alpha'_h,\varepsilon/2}$;
            \State Use Lemma \ref{Ceps_lm} to $y_h$ with the input accuracy $\varepsilon/2$. Denote the resulting vector by $\hat \alpha'_h$;
          \EndFor
          \For{$q \in \QC$}
            \State Compose a vector $\beta'_q \in \CCQ[\HC]$ by the rule: $\beta'_q(h) = \hat \alpha'_h(\chi_q)$, for $h \in \HC$;
          \EndFor
          \For{$q \in \QC$}
            \State Use approximate algorithm for the cyclic convolution in $\CC[\HC]$, described in Lemma \ref{cyclic_DFT_lm}, to $\beta'_q$ with the input accuracy $\abs{\QC}^{I_1} \cdot \varepsilon$. Denote the resulting vector by $\hat \beta'_q$; 
          \EndFor
          \For{$q \in \QC$, $h \in \HC$}
            \State $\hat \alpha'(\chi_q \cdot \chi_h) \gets \hat \beta'_q(\chi_h)$;
          \EndFor
          \State \textbf{return} $\hat \alpha'$;
      \EndProcedure
      \end{algorithmic}
  \end{algorithm}

    According to Algorithms \ref{kron_prod_alg} and \ref{appr_kron_prod_alg}, we denote the approximate versions of the variables $\{\alpha, \hat \alpha, \alpha_h, \hat \alpha_h, \beta_q, \hat \beta_q \}$ by $\{\alpha', \hat \alpha', \alpha'_h, \hat \alpha'_h, \beta'_q, \hat \beta'_q \}$. We split the analysis of Algorithm \ref{appr_kron_prod_alg} into 4 parts: {\bf accuracy (the $(1)$-th Claim), size of the intermediate variables (the $(2)$-th Claim), number of operations (the $(3)$-th Claim), size of elements of an output vector (the $(4)$-th and $(5)$-th Claims)}.

    {\bf Accuracy analysis.}
   Due to proposed scheme, for each $h \in \HC$, we have
   $$
   \norm{\hat\alpha_h - \hat\alpha'_h}_{\infty} \leq \norm{\hat\alpha_h - y_h }_{\infty} + \norm{y_h - \hat\alpha'_h}_{\infty} \leq \abs{\QC}^{I_1} \cdot \varepsilon.
   $$
   Clearly, the same holds for $\beta'_q$, for each $q \in \QC$:
   $$
   \norm{\beta_q - \beta'_q}_\infty \leq \abs{\QC}^{I_1} \cdot \varepsilon.
   $$
   By Lemma \ref{cyclic_DFT_lm}, there exists an absolute constant $C_1$, such that
   \begin{multline*}
       \norm{\hat\beta_q - \hat\beta'_q}_\infty \leq \norm{\beta_q - \beta'_q}_\infty \cdot \abs{\HC}^{C_1} \leq \abs{\QC}^{I_1} \cdot \abs{\HC}^{C_1} \cdot \varepsilon \leq \\
       \leq \abs{\GC}^{I_1} \cdot \varepsilon \quad \text{(taking $I_1 > C_1$).}
   \end{multline*}
   Consequently,
   $$
   \norm{\hat\alpha - \hat\alpha'} \leq \abs{\GC}^{I_1} \cdot \varepsilon,
   $$
   which satisfies the Claim (1) and finishes the accuracy analysis.

    {\bf Size of the intermediate variables.}
    Due to the induction hypothesis, the size of intermediate variables inside the Lines 1--6 of Algorithm \ref{appr_kron_prod_alg} is bounded by $I_2 \cdot \bigl(\log_2\abs{\QC} + \gamma + \size(\varepsilon)\bigr)$. Due to the Claim (4) of the induction hypothesis and Lemma \ref{Ceps_lm}, the size of variables inside the Line 7 is bounded by $(I_4 + C_2) \cdot \bigl(\log_2\abs{\QC} + \gamma + \size(\varepsilon)\bigr)$, for a sufficiently big absolute constant $C_2$. Note that, for each $q \in \QC$, $\norm{\beta'_q}_\infty \leq \norm{\beta'_q - \beta_q}_\infty + \norm{\beta_q}_\infty \leq 1 + \abs{\QC} \cdot 2^\gamma = \abs{\QC} \cdot 2^{O(\gamma)}$. Hence, due to the Line 7 and Lemma \ref{Ceps_lm}, for any $q \in \QC$, the size of elements of $\beta'_q$ is bounded by $O\bigl(\log_2\abs{\QC} + \gamma + \size(\varepsilon)\bigr)$.
    Therefore, due to Lemma \ref{cyclic_DFT_lm} and choosing sufficiently large value of $C_2$, the size of intermediate variables in the steps 8--18 can be bounded by $C_2 \cdot \bigl(\log\abs{\GC} + \gamma + \size(\varepsilon)\bigr)$.
    Finally, due to proposed reasoning, taking $I_2 > I_4+C_2$, we satisfy the Claim (2) and finish the analysis of the size of intermediate variables.

    {\bf Number of operations.} Assume that Algorithm \ref{appr_kron_prod_alg} is used $N$ times on different inputs from $\CCQ[\QC\oplus\HC]$. Due to the induction hypothesis, the Lines 1--6 need totally \begin{equation}\label{QC_recursive_complexity}
    I_3 \cdot \abs{\QC} \cdot \bigl(N \cdot \abs{\HC} \cdot \log_2\abs{\QC} + \gamma + \log_2(1/\varepsilon)\bigr)
    \end{equation}
    operations. Due to Lemma \ref{Ceps_lm}, the Line 7 needs totally $O\bigl(N \cdot \abs{\QC} \cdot \abs{\HC}\bigr)$ operations. Due to Remark \ref{cyclic_DFT_rm}, the Lines 8--18 cost
    $$
    C_3 \cdot \abs{\HC} \cdot \bigl(N \cdot \abs{\QC} \cdot \log_2\abs{\HC} + \gamma_\beta + \log_2(1/\varepsilon') \bigr) \Bigr)
    $$ operations, where $C_3$ is a sufficiently large absolute constant, $\gamma_\beta$ is the maximum size of elements in $\beta'_q$, for $q \in \QC$, and $\varepsilon' = \abs{\QC}^{I_1} \cdot \varepsilon$. Choosing a sufficiently large value of $C_3$, we can assume that the total computational cost of the Line 7 is hidden inside $C_3$.
    Since $\gamma_\beta \leq C_4 \cdot (\log_2\abs{\QC} + \gamma)$, for some absolute constant $C_4$, the computational cost of the Lines 8--18 can be estimated by
    \begin{multline}
        C_3 \cdot \abs{\HC} \cdot \bigl(N \cdot \abs{\QC} \cdot \log_2\abs{\HC} + \gamma + (C_4-I_1)\cdot \log_2 \abs{\QC} + \log_2(1/\varepsilon) \bigr) \leq \notag\\
        \leq C_3 \cdot \abs{\HC} \cdot \bigl(N \cdot \abs{\QC} \cdot \log_2\abs{\HC} + \gamma + \log_2(1/\varepsilon) \bigr),\quad \text{taking $I_1 > C_4$.} \label{HC_recursive_complexity}
    \end{multline}
    Taking $I_3 > C_3$, since $\abs{\QC} \leq \abs{\GC}/2$ and $\abs{\HC} \leq \abs{\GC}/2$, the total complexity can be estimated by
    \begin{multline*}
        I_3 \cdot N \cdot  \abs{\GC} \cdot \log_2\abs{\GC} + I_3 \cdot \bigl(\gamma + \log_2(1/\varepsilon)\bigr) \cdot \bigl(\abs{\QC} + \abs{\HC}\bigr) \leq \\
        \leq I_3 \cdot  \abs{\GC} \cdot \bigl(N \cdot \log_2\abs{\GC} + \gamma + \log_2(1/\varepsilon)\bigr),
    \end{multline*}
    which proves the Claim (3) and finishes the arithmetic complexity analysis.

    {\bf Output size.}
    Due to Lemma \ref{cyclic_DFT_lm}, the output size is bounded by $C_5 \cdot \bigl(\log_2(\HC) + \gamma_\beta + \log_2(1/\varepsilon')\bigr) \leq (C_5+C_4) \cdot (\log_2(\GC) + \gamma + \log_2(1/\varepsilon)\bigr)$, for some absolute constant $C_5$. Taking $I_4 > C_4 + C_5$, we satisfy the Claims (4) and (5), which finishes analysis of output size. Therefore, all the Claims are satisfied, which finishes the proof of the theorem.
\end{myproof}